\documentclass[11pt]{article}

\usepackage{epsfig} 
\usepackage{amsmath} 
\usepackage{amssymb}  
\usepackage{amsthm}
\usepackage{cite}
\usepackage{color}
\usepackage{enumitem}
\usepackage{url}
\usepackage{subcaption}

\newcommand*{\real}{\mathbb{R}}

\theoremstyle{plain}
\newtheorem{lemma}{Lemma}
\newtheorem{theorem}{Theorem}

\theoremstyle{definition}
\newtheorem{assumption}{Assumption}
\newtheorem{definition}{Definition}
\newtheorem{algorithm}{Algorithm}
\newtheorem{remark}{Remark}

\author{Kin Cheong Sou}
\title{Minimum Equivalent Precedence Relation Systems}
\date{\today}

\begin{document}
\maketitle

\begin{abstract}
In this paper two related simplification problems for systems of linear inequalities describing precedence relation systems are considered. Given a precedence relation system, the first problem seeks a minimum subset of the precedence relations (i.e., inequalities) which has the same solution set as that of the original system. The second problem is the same as the first one except that the ``subset restriction'' in the first problem is removed. This paper establishes that the first problem is NP-hard. However, a sufficient condition is provided under which the first problem is solvable in polynomial-time. In addition, a decomposition of the first problem into independent tractable and intractable subproblems is derived. The second problem is shown to be solvable in polynomial-time, with a full parameterization of all solutions described. The results in this paper generalize those in [Moyles and Thompson 1969, Aho, Garey, and Ullman 1972] for the minimum equivalent graph problem and transitive reduction problem, which are applicable to unweighted directed graphs.
\end{abstract}

\section{Introduction}

\subsection{Statement of problem}
In this paper we consider precedence relation systems with $n$ variables of the form:
\begin{equation} \label{eqn:prec_sys}
x_i - x_j \le c_{ij}, \quad (i,j) \in E,
\end{equation}
where $c_{ij} \in \real$ and $E \subseteq \{1,\ldots,n\} \times \{1,\ldots,n\}$ are given, and $x_1, x_2, \ldots, x_n$ are the variables. $E$ is the index set of all precedence relations in (\ref{eqn:prec_sys}). Let $c(E) \in \real^{|E|}$ be the vector edge weights such that if the $k^\text{th}$ element of $E$ is $(i,j)$, then $c_k = c_{ij}$. Also, let $V = \{1,2,\ldots,n\}$ denote the set of variable indices. Then, system~(\ref{eqn:prec_sys}) is described by the triple $(V,E,c(E))$.


This paper considers two related problems regarding the simplification of precedence relation systems. 

{\bf The first problem -- maximum index set of redundant relations problem:} we seek a \emph{maximum (cardinality) index set of redundant relations}, with the definition of an index set of redundant relations given by:
\begin{definition}[Index set of redundant relations] \label{def:redundant_relations}
Let $(V,E,c(E))$ be a precedence relation system, then $R \subseteq E$ is called an index set of redundant relations of $(V,E,c(E))$ if 
\begin{equation} \label{eqn:redundant_sys_set}
\text{$x \in \real^n$ satisfies $(V,E,c(E))$} \iff \text{$x$ satisfies $(V,(E \setminus R),c(E \setminus R))$}. 
\end{equation}
\end{definition}
In this paper, two precedence relation systems are equivalent (with symbol $\equiv$) if they have the same solution set. Therefore, condition~(\ref{eqn:redundant_sys_set}) can be stated as $(V,(E \setminus R),c(E \setminus R)) \equiv (V,E,c(E))$. In its minimization form, the first problem can be posed as finding a minimum (cardinality) subset $E' \subseteq E$ (i.e., $E' = E \setminus R$) such that $(V,E',c(E')) \equiv (V,E,c(E))$.

{\bf The second problem -- equivalent reduction problem:} we seek an \emph{equivalent reduction} of a given precedence relation system $(V,E,c(E))$ as follows:
\begin{definition}[Equivalent reduction] \label{def:equivalent_reduction}
Let $(V,E,c(E))$ be a precedence relation system. A precedence relation system $(V,E^\prime,c^\prime(E^\prime))$ with $E^\prime \subseteq V \times V$ and $c^\prime(E^\prime) \in \real^{|E^\prime |}$ is an equivalent reduction of $(V,E,c(E))$ if it satisfies the following two conditions:
\begin{enumerate}[label=\ref{def:equivalent_reduction}.\alph*]
\item\label{def:er_equiv} $x \in \real^n$ satisfies $(V,E,c(E)) \iff x$ satisfies $(V,E^\prime,c^\prime(E^\prime))$,
\item with respect to property \ref{def:er_equiv}, $E^\prime$ has the minimum cardinality.
\end{enumerate}
\end{definition}

The maximum index set of redundant relations problem (i.e., the first problem) is a restriction of the equivalent reduction problem (i.e., the second problem). In the search of the minimum equivalent system of precedence relations for the first problem (in its minimization version) we are restricted to a subset of (\ref{eqn:prec_sys}), whereas in the second problem there is no such restriction. In addition, according to definition the values in $c^\prime$ in the second problem need not be the same as $c$ as in the first problem. The distinction between the two problems considered in this paper is analogous to the distinction between minimum equivalent graph in \cite{Moyles:1969:AFM:321526.321534} and transitive reduction in \cite{AGU_transitive_reduction}, in the setting of unweighted directed graph simplification.

\subsection{Main contributions and previous works}

The main contributions of this paper are as follows:
\begin{enumerate}
\item We derive a sufficient condition under which the maximum index set of redundant relations problem has a unique solution and is polynomial-time solvable. In addition, we show that in general the maximum index set of redundant relations problem is NP-hard.
\item We show that the maximum index set of redundant relations problem can be decomposed into a finite number of independent subproblems, one of which being polynomial-time solvable and the rest NP-hard.
\item Based on the decomposition, we provide a parameterization of all solutions to the maximum index set of redundant relations problem.
\item We provide a complete parameterization of all solutions to the equivalent reduction problem. The parameterization suggests a procedure that can find any solution to this problem in polynomial time.
\end{enumerate}

In essence, the results in this paper are generalizations of those in \cite{Moyles:1969:AFM:321526.321534,AGU_transitive_reduction}. The generalization is in the sense that the results in \cite{Moyles:1969:AFM:321526.321534,AGU_transitive_reduction} pertain \emph{unweighted} directed graphs, while the results in this paper pertain \emph{weighted} directed graphs (the relation between weighted directed graphs and precedence inequality systems will be explained in the sequel). When $c_{ij} = 0$ for all $(i,j) \in E$, it can be shown that (\ref{eqn:redundant_sys_set}) is equivalent to the condition that unweighted directed graph $(V,E)$, with $V$ and $E$ being the node set and edge set respectively, has the same reachability as $(V,E \setminus R)$. That is, there is a walk from $i \in V$ to $j \in V$ in $(V,E)$ if and only if there is a walk from $i$ to $j$ in $(V,E \setminus R)$. Thus, the minimum equivalent graph problem for \emph{unweighted} directed graphs \cite{Moyles:1969:AFM:321526.321534} can be reduced to the maximum index set of redundant relations problem with $c_{ij} = 0$ for all $(i,j) \in E$. An implication of the reduction is that the results in this paper can be specialized to obtain those in \cite{Moyles:1969:AFM:321526.321534}. In addition, we establish that even in the generalized setup the complexity and decomposition results in this paper are analogous to those in \cite{Moyles:1969:AFM:321526.321534}. However, the results in this paper and in \cite{Moyles:1969:AFM:321526.321534} are not exactly the same -- there is a difference in the equivalence classes (in the node set) defining the decompositions. Similarly, an instance of the transitive reduction problem (studied in  \cite{AGU_transitive_reduction}) for unweighted directed graph $(V,E)$ can be solved as an instance of equivalent reduction problem with $(V,E,c(E))$ with $c_{ij} = 0$ for all $(i,j) \in E$. In addition, the complexity result and the parameterization of the set of all equivalent reductions provided in this paper are analogous to those of \cite{AGU_transitive_reduction}.

In comparison with methods to simplify general sets of linear equalities and inequalities (e.g., \cite{telgen1983identifying,greenberg1996consistency}), this paper proposes more specialized and time-efficient algorithms for the more restrictive setting of precedence relation inequalities in (\ref{eqn:prec_sys}).

\subsection{Application motivations}
The precedence relation system in (\ref{eqn:prec_sys}) arises in applications such as machine scheduling (e.g., \cite{Finta1996247,Balas:1995:OPD:208828.208841,Chekuri199929,Brucker199977,Wikum199487}), chemical process planning (e.g., \cite{chemical}), smart grid (e.g., \cite{seaport_scheduling_CDC2011,Sou_ECC2013}), parallel computing (e.g., \cite{1639388}) and flexible manufacture systems (e.g., \cite{Pinilla2001341}). In particular, in \cite{Brucker199977,Wikum199487} constraints in (\ref{eqn:prec_sys}) are referred to as positive and negative time-lag constraints and generalized precedence constraints, respectively. Scheduling problems with (\ref{eqn:prec_sys}) are analyzed in \cite{Brucker199977,Wikum199487} and the subsequent literature.

It will be shown, for instance, that the equivalent reduction problem can be solved in $O({|V|}^3)$ time. Hence, algorithms that require more than $O({|V|}^3)$ time for problems involving precedence constraints in (\ref{eqn:prec_sys}) can potentially benefit from the simplification results in this paper. For example, \cite{SSJ_DP2014_part2} considered a nonconvex resource allocation problem where the decision variables $x_i$'s are the start times of tasks to be scheduled. In addition, $x_i$'s are precedence-constrained as in (\ref{eqn:prec_sys}). It was shown in \cite{SSJ_DP2014_part2} that the computation effort for solving the resource allocation problem using dynamic programming increases exponentially with the cardinality of the minimum feedback vertex set of the undirected version of graph $(V,E)$. By solving the precedence relation system simplification problems in this paper, and replacing $(V,E)$ with an equivalent $(V,E')$ with $|E'| < |E|$, it is possible to reduce the computation effort for solving the resource allocation problem in \cite{SSJ_DP2014_part2}.

In \cite{Coffman_Graham} Coffman and Graham derive an algorithm to solve a machine scheduling problem with a special case of (\ref{eqn:prec_sys}), where $x_i$ is interpreted as the start time of task $i$ in a $n$-task scheduling problem and $c_{ij} = -p_i$ with $p_i$ being the given processing time for task $i$. The algorithm in \cite{Coffman_Graham} needs to remove all redundant constraints in the specialized version of (\ref{eqn:prec_sys}). This can be achieved by applying transitive reduction in \cite{AGU_transitive_reduction} to an appropriately constructed unweighted directed acyclic graph. However, in \cite{Brucker199977,Wikum199487} machine scheduling problems with (\ref{eqn:prec_sys}) in its full generality are considered. We are not aware of any solution algorithm for the generalized machine scheduling problem that is similar to the one by Coffman and Graham in \cite{Coffman_Graham} using transitive reduction. However, it seems plausible that if a Coffman-Graham-like algorithm is to be developed for the general machine scheduling problem, the simplification results presented in this paper for (\ref{eqn:prec_sys}) in its full generality would be needed.

\subsection{Organization}
The rest of the paper is organized as follows: Section~\ref{sec:graph_interpretation} defines the precedence graph associated with (\ref{eqn:prec_sys}) and lists some results that are useful in the sequel. Section~\ref{sec:graph_interpretation} also establishes the equivalence between the algebraic notion of index set of redundant relations with a precedence graph-based concept to be defined as redundant edge set. Then, Section~\ref{sec:maximum_redundant_edge_set} develops the complexity and decomposition results for the problem of finding the maximum redundant edge set. After that, Section~\ref{sec:equiv_reduction} is dedicated to the problem of equivalent reduction. It characterizes the set of all equivalent reductions of any precedence relation system, and establishes that any equivalent reduction can be found in polynomial-time. Section~\ref{sec:conclusion} concludes the paper.

\section{Graph interpretation of the main problems} \label{sec:graph_interpretation}
Section~\ref{subsec:precedence_graph} defines precedence graph and establishes some properties necessary for the subsequent analyses. Then, Section~\ref{subsec:reformulation} describes a reformulation of the maximum index set of redundant relations problem into an equivalent graph-based problem involving the precedence graph of (\ref{eqn:prec_sys}). The graph-based formulation will be studied in detail in Section~\ref{sec:maximum_redundant_edge_set}.

\subsection{Precedence graph description and supplementary results} \label{subsec:precedence_graph}
In this paper, we make extensive use of precedence graph to describe a precedence relation system. In fact, we use these two concepts interchangeably. A precedence graph $(V,E,c(E))$ is an edge weighted directed graph with node set, edge set and vector of edge weights being $V$, $E$ and $c(E)$ respectively. In addition, a precedence graph corresponds to a precedence relation system sharing the same triple $(V,E,c(E))$. The correspondence is as follows:
\begin{table}[h]
\begin{center}
\begin{tabular}{|c|c|c|}
\hline
& precedence graph & precedence relation system \\
\hline
$V$ & node set & index set of variables \\
\hline
$E$ & edge set & index set of inequalities \\
\hline
$c_{ij}$ for $(i,j) \in E$ & edge weight & $x_i - x_j \le c_{ij}$ \\
\hline
\end{tabular}
\end{center}
\end{table}

The following standard graph notions are defined in order to make the subsequent discussions more precise. A walk from $u \in V$ and $v \in V$ is defined as $(u = i_0, i_1, \ldots, i_m = v)$ where $(i_k,i_{k+1}) \in E$ for $k = 0,\ldots,m-1$ and the traversed nodes $i_0,\ldots,i_m$ are not necessarily distinct. A closed walk is a walk $(i_0,i_1,\ldots,i_m)$ with the additional requirement that $i_0 = i_m$. A (simple) path from $u$ to $v$ is a walk with the additional requirement that all traversed nodes are distinct. A cycle is a closed walk where $i_0 = i_m$ and all other traversed nodes are distinct. The weight of a walk (e.g., path, cycle) is the sum of the weights of all traversed edges, with the edge weight added as many times as an edge is traversed. This paper allows degenerate path $(u)$ for $u \in V$ (i.e., a single node), and the weight associated with $(u)$ is zero.

The following symbols will be used throughout the paper. For any two nodes $u$ and $v$, the symbol $p_{uv}$ is used to denote a path from $u$ to $v$. The weight of a path $p_{uv}$ is denoted $c_{p_{uv}}$. Similarly, the symbol $w_{uv}$ is used to denote a walk from $u$ to $v$ with the corresponding walk weight denoted $c_{w_{uv}}$. The symbol $u \leadsto v$ is used to substitute the phrase ``from $u$ to $v$''. The symbol $u \rightarrow v$ is used to denote an edge from $u$ to $v$.

Because of the associated precedence inequalities, a precedence graph has the following property:
\begin{lemma} \label{thm:no_negative_weight_cycle}
Let $G = (V,E,c(E))$ be a precedence graph. If the corresponding precedence inequality system has at least one solution, then the weights of all closed walks (e.g., cycles) in $G$ are nonnegative.
\end{lemma}
\begin{proof}
Let $(i_0,i_1,\ldots,i_m)$ with $i_0 = i_m$ denote a closed walk in $G$. By the statement assumption, there exists $x \in \real^n$ satisfying the inequalities corresponding to the edges in the closed walk:
\begin{equation} \label{eqn:cycle_constraints}
\begin{array}{ccc}
x_{i_0}  - x_{i_1} & \le & c_{i_0 i_1} \\
x_{i_1}  - x_{i_2} & \le & c_{i_1 i_2} \\
& \vdots & \\
x_{i_{m-1}}  - x_{i_{m}} & \le & c_{i_{m-1} i_m}.
\end{array}
\end{equation}
Summing up all inequalities, with the fact that $i_0 = i_m$, leads to the desired inequality that $0 \le c_{i_0 i_1} + c_{i_1 i_2} + \ldots + c_{i_{m-1} i_m}$.
\end{proof}

Certain assumptions on the graphs considered in this paper are made:
\begin{itemize}
\item We call a precedence graph feasible if its corresponding precedence relation system has at least one solution. In this paper, all except one precedence graphs are assumed to be feasible (the only exception is in the last part of the proof of Theorem~\ref{thm:equiv_reduction}). This is due to the fact that all given precedence graphs can be assumed to be feasible, and there is no need to consider derived precedence graphs that are not feasible except in the only exception mentioned above. Hence, if there is no mentioning of the feasibility of a precedence graph, Lemma~\ref{thm:no_negative_weight_cycle} is assumed to be applicable to this graph.
\item It is assumed that in all graphs there is at most one (directed) edge from one node to another. In other words, no parallel edges are allowed. This assumption is obvious for precedence graphs: if for any pair $(i,j)$ multiple precedence relations hold with $x_i - x_j \le c_{ij}^k$ for $k = 1,2,\ldots,m$, then $x_i - x_j \le \min_k \{c_{ij}^k\}$ summarizes the same relations. Since all other graphs (e.g., those in statement assumptions) in fact represent precedence graphs, the no-parallel-edge assumption is imposed on these graphs as well.
\item It is assumed that there is no self-loop of the form $(i,i)$ in all graphs. Again, this assumption is obvious for precedence graphs: if the precedence graph is feasible then a self-loop means $x_i - x_i = 0 \le c_{ii}$. This inequality on $c_{ii}$ can be removed without affecting the rest of the precedence relation system. It will also be obvious that for the only exception precedence graph whose feasibility cannot be taken for granted, there is no need to include self-loop in it.
\end{itemize}

In summary, the following are the standing assumptions of this paper:

\begin{assumption} \label{asm:graph_assumptions} \hspace{0cm}
\begin{enumerate}[label=\ref{asm:graph_assumptions}.\alph*]
\item \label{asm:feasibility} With only one exception in the proof of Theorem~\ref{thm:equiv_reduction}, all precedence graphs are feasible.
\item \label{asm:no_negative_cycle} With only one exception in the proof of Theorem~\ref{thm:equiv_reduction}, no precedence graph contains any negative weight closed walks.
\item \label{asm:parallel_edge} There is no parallel edges between nodes in any graph.
\item \label{asm:self_loop} There is no self-loop in any graph.
\end{enumerate}
\end{assumption}
\begin{remark}
If a graph contains a negative weight cycle then it contains a negative weight closed walk. Conversely, if a graph contains a negative weight closed walk, then it must contain a negative weight cycle because a closed walk can be decomposed into a finite number of cycles, with the walk weight being the sum of the weights of the cycles (see Appendix~\ref{app:walk_decomp} for further details). Therefore, Assumption~\ref{asm:no_negative_cycle} is in fact equivalent to the assumption that the no precedence graph contains any negative weight cycles. In addition, the no-negative-weight-cycle and no-negative-weight-closed-walk assumptions will be used interchangeably in this paper for convenience.
\end{remark}

The following statement is important in the subsequent developments. For instance, the statement establishes the equivalence between the algebraic conditions in (\ref{eqn:redundant_sys_set}) and a graph theoretic condition in the corresponding precedence graph. A similar algebraic/graph equivalence can be established for the case of Definition~\ref{def:er_equiv}, with the aid of this same statement.

\begin{lemma} \label{thm:implication}
Let $G = (V,E,c(E))$ be a precedence graph, and its corresponding precedence inequality system be
\begin{equation} \label{eqn:Aij}
x_i - x_j \le c_{ij}, \quad (i,j) \in E.
\end{equation}
Let $u,v \in V$, $u \neq v$, $b_{uv} \in \real$ be given. Then the following two conditions are equivalent:
\begin{enumerate}[label=\ref{thm:implication}.\alph*]
\item\label{lem:implication} Whenever $x \in \real^n$ satisfies (\ref{eqn:Aij}), $x$ satisfies the inequality $x_u - x_v \le b_{uv}$.
\item\label{lem:rp} There exists a path $p_{uv}$, $u \leadsto v$ in $G$ such that the weight of $p_{uv}$, denoted $c_{p_{uv}}$, satisfies $c_{p_{uv}} \le b_{uv}$.
\end{enumerate}
\end{lemma}

\begin{proof}
Since $u \neq v$, we can define $a_{uv} \in \real^n$ such that
\begin{displaymath}
a_{uv}(k) = \left\{ \begin{array}{rl} 1 & \quad k = u \\ -1 & \quad k = v \\ 0 & \quad \text{otherwise} \end{array} \right.\ .
\end{displaymath}
With (standing) Assumption~\ref{asm:self_loop}, if $(i,j) \in E$ then $i \neq j$. Hence, for each $(i,j) \in E$ we can define the vector $a_{ij} \in \real^n$ such that
\begin{displaymath} 
a_{ij}(k) = \left\{ \begin{array}{rl} 1 & \quad k = i \\ -1 & \quad k = j \\ 0 & \quad \text{otherwise} \end{array} \right.\
\end{displaymath}
That is, $a_{ij}$ defines the column of the incidence matrix of $G$ for edge $(i,j)$. With $a_{ij}$, the inequality $x_i - x_j \le c_{ij}$ can be written as ${a_{ij}}^T x \le c_{ij}$. Furthermore, let $J_P$ denote the optimal objective value of the following linear program
\begin{equation} \label{opt:J_P}
\begin{array}{ccl}
J_P := & \underset{x \in \real^n}{\max} & {a_{uv}}^T x \vspace{2mm} \\
& \text{subject to} & {a_{ij}}^T x \le c_{ij}, \quad \forall (i,j) \in E \vspace{2mm}.
\end{array}
\end{equation}
Then, condition~\ref{lem:implication} holds if and only if $J_P \le b_{uv}$. The linear programming dual of (\ref{opt:J_P}) and its optimal objective value can be written as
\begin{equation} \label{opt:J_D}
\begin{array}{ccl}
J_D := & \underset{y_{ij} \in \real}{\min} & \sum\limits_{(i,j) \in E} y_{ij} c_{ij} \vspace{2mm} \\
& \text{subject to} & \sum\limits_{(i,j) \in E} y_{ij} a_{ij} = a_{uv} \vspace{2mm} \\
& & y_{ij} \ge 0, \quad \forall (i,j) \in E.
\end{array}
\end{equation}
Due to Assumption~\ref{asm:feasibility}, (\ref{eqn:Aij}) has at least one solution. Hence, the feasible set of (\ref{opt:J_P}) is nonempty. Thus, by linear programming duality (e.g., \cite{BT97}) it holds that $J_D = J_P$ with the convention that $J_D = \infty$ whenever (\ref{opt:J_D}) is infeasible. Hence,
\begin{equation} \label{eqn:impli_eq_cond}
\text{condition~\ref{lem:implication} holds} \iff J_P \le b_{uv} \iff J_D \le b_{uv}.
\end{equation}
Note that $a_{uv}$ contains only $-1$, 1 or 0. In addition, $a_{ij}$ for $(i,j) \in E$ in the constraint of (\ref{opt:J_D}) are columns of an incidence matrix (of $G$), which is totally unimodular (e.g., \cite{Schrijver_CCO}). Thus, a standard combinatorial optimization argument (e.g., \cite{Schrijver_CCO}) implies that the relations in (\ref{eqn:impli_eq_cond}) can be extended to
\begin{equation} \label{eqn:impli_eq_cond1}
\text{condition~\ref{lem:implication} holds} \iff J_P \le b_{uv} \iff J_D \le b_{uv} \iff J_B \le b_{uv},
\end{equation}
where $J_B$ is the optimal objective value of the following $0-1$ binary linear integer problem
\begin{equation} \label{opt:J_B}
\begin{array}{ccl}
J_B := & \underset{y_{ij} \in \real}{\min} & \sum\limits_{(i,j) \in E} y_{ij} c_{ij} \vspace{2mm} \\
& \text{subject to} & \sum\limits_{(i,j) \in E} y_{ij} a_{ij} = a_{uv} \vspace{2mm} \\
& & y_{ij} \in \{0,1\}, \quad \forall (i,j) \in E.
\end{array}
\end{equation}
That is, (\ref{opt:J_B}) is almost the same as (\ref{opt:J_D}) except that the decision variables in (\ref{opt:J_B}) are restricted to 0 or 1. Next, we establish that $J_B \le b_{uv}$ is equivalent to condition~\ref{lem:rp}. That is,
\begin{equation} \label{eqn:J_B_puv}
J_B \le b_{uv} \iff \text{$\exists$ path $p_{uv} : u \leadsto v$ in $G$ such that $c_{p_{uv}} \le b_{uv}$}.
\end{equation}
One side of the implication is easy to establish: if there is a path $p_{uv} : u \leadsto v$ in $G$ such that $c_{p_{uv}} \le b_{uv}$, then by assigning $y_{ij} = 1$ if and only if $(i,j)$ is part of the path we obtain a feasible solution of (\ref{opt:J_B}) with objective value being $c_{p_{uv}} \le b_{uv}$. Hence, $J_B \le b_{uv}$. Conversely, if $J_B \le b_{uv}$ then the following procedure can be employed to retrieve a walk $p_{uv} : u \leadsto v$ in $G$ such that $c_{p_{uv}} \le b_{uv}$:\begin{itemize}
\item Initialize $W \leftarrow \emptyset$, and $Y \leftarrow \{(i,j) \in E \mid y_{ij}^\star = 1 \}$ with $y^\star$ being an optimal solution to (\ref{opt:J_B})
\item While $Y \neq \emptyset$, do
\begin{enumerate}
\item If $(u,s) \in Y$ for some $s$ then set $(i,j) \leftarrow (u,s)$, otherwise set $(i,j)$ as an arbitrary edge in $Y$
\item $Y \leftarrow Y \setminus \{(i,j)\}$
\item $w \leftarrow (i,j)$
\item While $(j,t) \in Y$ for some $t$, do
\begin{enumerate}
\item $Y \leftarrow Y \setminus \{(j,t)\}$
\item $w \leftarrow w+t$ (meaning that $w$ is appended in the end by $t$)
\item $j \leftarrow t$
\end{enumerate}
\item End (of second while)
\item $W \leftarrow W \cup \{w\}$
\end{enumerate}
\item End (of first while)
\end{itemize}
To analyze the procedure the following nonnegative degree counters are needed: $d^+(i)$ is the out-degree at node $i$ in the graph $(V,Y)$, where $Y$ is being updated in the procedure. Similarly, $d^-(i)$ is the in-degree at node $i$. Initially, the degree counters satisfy
\begin{equation} \label{eqn:deg_initial}
d^+(u) = d^-(u) + 1, \quad d^+(v) = d^-(v) - 1, \quad d^+(i) = d^-(i), \; i \neq u,v,
\end{equation}
because for each $k \in \{1,2,\ldots,n\}$ the constraint in (\ref{opt:J_B}) can be written as
\begin{displaymath}
\sum\limits_{(i,j) \in E} \!\!\!\! y_{ij}^\star a_{ij}(k) = \!\!\!\! \sum\limits_{j \mid (k,j) \in E} \!\!\!\! y_{kj}^\star - \!\!\!\! \sum\limits_{j \mid (j,k) \in E} \!\!\!\! y_{jk}^\star = a_{uv}(k) = \left\{ \begin{array}{rl} 1 & \quad k = u \\ -1 & \quad k = v \\ 0 & \quad k \neq u, v \end{array} \right.\ .
\end{displaymath}

Now we first analyze some general properties of the procedure. Notice that the procedure should terminates in finite number of passes of the first while-loop because within each pass at least one edge is removed from $Y$ and $Y$ is finite. Also, with a similar argument we can establish that the second while-loop terminates for each pass of the first while-loop. At the end of a pass of the first while-loop, a walk $w = (i = i_0, i_1, \ldots, i_m)$ is retrieved. There can be two possibilities with $i_m$: either (I) $i_m = i_0$ (i.e., $w$ is a closed walk) or (II) $i_m \neq i_0$. Let $d_{\text{fore}}^+(i)$ and $d_{\text{fore}}^-(i)$ denote, respectively, the out-degree and in-degree at node $i$ before the pass of the first while-loop. Similarly, let $d_{\text{aft}}^+(i)$ and $d_{\text{aft}}^-(i)$ denote the analogous degrees after the pass of the first while-loop. Then, for case (I) with $i_0 = i_m$ the degree adjustments (after a pass of the first while-loop) are as follows:
\begin{equation} \label{eqn:deg_adj_cw}
\begin{array}{ccl}
d_{\text{aft}}^+(i_q) & = & d_{\text{fore}}^+(i_q) - r_q, \quad \forall \; q \in \{0,1,\ldots,m\},\vspace{2mm} \\
d_{\text{aft}}^-(i_q) & = & d_{\text{fore}}^-(i_q) - r_q, \quad \forall \; q \in \{0,1,\ldots,m\}, \vspace{2mm} \\
d_{\text{aft}}^+(i_m) & = & 0.
\end{array}
\end{equation}
In (\ref{eqn:deg_adj_cw}), $r_q$ are some nonnegative numbers satisfying $d_{\text{aft}}^+(i_q) \ge 0$ and $d_{\text{aft}}^-(i_q) \ge 0$. In other words, the difference between the out-degree and in-degree for each node remains the same in case (I) (when $w$ is a closed walk). For case (II) with $i_0 \neq i_m$ the degree adjustments are as follows:
\begin{equation} \label{eqn:deg_adj_w}
\begin{array}{rcl}
d_{\text{aft}}^+(i_0) & = & d_{\text{fore}}^+(i_0) - r_0, \vspace{2mm} \\
d_{\text{aft}}^-(i_0) & = & d_{\text{fore}}^-(i_0) - r_0 + 1, \vspace{2mm} \\
d_{\text{aft}}^+(i_q) & = & d_{\text{fore}}^+(i_q) - r_q, \quad \forall \; q \in \{1,2,\ldots,m-1\},\vspace{2mm} \\
d_{\text{aft}}^-(i_q) & = & d_{\text{fore}}^-(i_q) - r_q, \quad \forall \; q \in \{1,2,\ldots,m-1\}, \vspace{2mm} \\
d_{\text{aft}}^+(i_m) & = & d_{\text{fore}}^+(i_m) - r_m + 1, \vspace{2mm} \\
d_{\text{aft}}^-(i_m) & = & d_{\text{fore}}^-(i_m) - r_m, \vspace{2mm} \\
d_{\text{aft}}^+(i_m) & = & 0.
\end{array}
\end{equation}
From (\ref{eqn:deg_adj_w}), the in-degree counter of the end node $i_m$ must satisfy
\begin{equation} \label{eqn:deg_adj_w_im}
d_{\text{aft}}^-(i_m) = d_{\text{fore}}^-(i_m) - d_{\text{fore}}^+(i_m) - 1.
\end{equation}

Now we analyze the first pass of the first while-loop. With $J_B \le b_{uv}$ assumed (hence feasibility), the constraint in (\ref{opt:J_B}) requires that $(u,s) \in Y$ for some $s$. Hence, $i_0 = u$ for the walk $w$ in the first pass. In addition, $w$ cannot be a closed walk because (\ref{eqn:deg_initial}) (with degree counters interpreted as those before the pass) and (\ref{eqn:deg_adj_cw}) together imply that $d_{\text{aft}}^-(u) = d_{\text{fore}}^-(u) - d_{\text{fore}}^+(u) = -1$ which is impossible. Hence, case (II) must hold. Then, (\ref{eqn:deg_initial}) (with degree counters interpreted as those before the pass) and (\ref{eqn:deg_adj_w_im}) together imply that 
\begin{equation} \label{eqn:im_v}
i_m = v, \quad d_{\text{aft}}^-(v) = 0,
\end{equation}
since all other choices of $i_m$ would result in negative in-degree after the pass. Therefore, $w$ is in fact a walk from $u$ to $v$, which is denoted as $w_{uv}$ and stored in $W$ in step 6. Furthermore, (\ref{eqn:deg_adj_w}) and (\ref{eqn:im_v}) together suggest that  after the first pass of the first while-loop the degree counters (in particular those of $u$) are updated to
\begin{equation} \label{eqn:degree_after1}
d^+(g) = d^-(g), \quad \forall \; g \in V, \quad d^+(v) = d^-(v) = 0.
\end{equation}
In summary, $v$ is no longer incident to any edge in $Y$ and the out-degree and in-degree for each node become equal (including the case of $u$), when the first pass of the first while-loop is finished. 

Next, we analyze the subsequent passes of the first while-loop. (\ref{eqn:degree_after1}) (with degree counters interpreted as those before the second pass) and (\ref{eqn:deg_adj_w_im}) implies that if the walk in the second pass is not a closed walk then $v$ must be the end node of this walk. Therefore, the walk from the second pass must be a close walk denoted $w_2$ because (\ref{eqn:degree_after1}) specifies that $v$ is no longer incident to any edge in $Y$ at this stage. In addition, by (\ref{eqn:deg_adj_cw}) the in/out degree difference for each node remains zero after the second pass. This pattern continues for all subsequent passes, generating closed walks $w_3, \ldots, w_N$ until $Y$ becomes an empty set. Upon termination of the first while-loop, $W$ contains $w_{uv}$ (i.e., the walk from $u$ to $v$) and closed walks $w_2, w_3, \ldots, w_N$. In addition, it can be seen that
\begin{displaymath}
c_{w_{uv}} + c_{w_2} + c_{w_3} + \ldots + c_{w_N} = J_B \le b_{uv}.
\end{displaymath}
This suggests that $c_{w_{uv}} \le b_{uv}$ since Assumption~\ref{asm:no_negative_cycle} states that $c_{w_q} \ge 0$ for $q = 2,3,\ldots,N$. The walk $w_{uv}$ may not be a path from $u$ to $v$. However, as explained in Appendix~\ref{app:walk_decomp} $w_{uv}$ can be further decomposed into a path $p_{uv} : u \leadsto v$ and a finite number of cycles, which have nonnegative weights because of Assumption~\ref{asm:no_negative_cycle} (this argument will in fact be formalized in Lemma~\ref{thm:walk2path} to be described). Hence, $p_{uv} \le w_{uv} \le b_{uv}$. In conclusion, when $J_B \le b_{uv}$, we can indeed obtain a path $p_{uv} : u \leadsto v$ with weight $p_{uv} \le b_{uv}$. This establishes (\ref{eqn:J_B_puv}). Finally, combining (\ref{eqn:impli_eq_cond1}) and (\ref{eqn:J_B_puv}) yields the desired statement that
\begin{displaymath}
\text{condition~\ref{lem:implication} holds} \iff J_B \le b_{uv} \iff \text{condition~\ref{lem:rp} holds}.
\end{displaymath}
\end{proof}

\subsection{Graph representation of maximum index set of redundant relations problem} \label{subsec:reformulation}
For the analysis in the sequel, the notion of index set of redundant relations defined in (\ref{eqn:redundant_sys_set}) will be reinterpreted as an equivalent and graph based concept of \emph{redundant edge set} to be defined in Definition~\ref{def:redundant_edge_set}. Consequently, the problem of finding the maximum index set of redundant relations in (\ref{eqn:prec_sys}) can be reformulated as the problem of finding the \emph{maximum redundant edge set} in the corresponding precedence graph. A consequence of the reformulation, as will be clear, is that the algebraic problem of finding the maximum index set of redundant relations can be solved using readily available graph-based algorithms (e.g., Floyd-Warshall shortest path algorithm).
\begin{definition}[Redundant edge set] \label{def:redundant_edge_set}
For an edge weighted directed graph $G = (V,E,c(E))$, an edge subset $R \subseteq E$ is called a redundant edge set of $G$ if either $R = \emptyset$ or when $R \neq \emptyset$, it holds that
\begin{equation} \label{eqn:redundant_edge_set}
\begin{array}{l}
\textrm{for every $(u,v) \in R$ with weight $c_{uv}$, there exists a path $p_{uv}$} \vspace{1.5mm} \\
 \textrm{in $(V, E \setminus R, c(E \setminus R))$ from $u$ to $v$ such that $c_{p_{uv}} \le c_{uv}$}.
\end{array}
\end{equation}
\end{definition}

\begin{remark}
For an edge weighted directed graph without negative weight cycles (e.g., a precedence graph), the definition of redundant edge set can be relaxed. Specifically, the existence of a path $p_{uv}$ in condition (\ref{eqn:redundant_edge_set}) can be replaced with the existence of a walk $w_{uv}$ from $u$ to $v$ in $(V, E \setminus R, c(E \setminus R))$ such that $c_{w_{uv}} \le c_{uv}$. The relaxation is justified by the following statement.
\end{remark}

\begin{lemma} \label{thm:walk2path}
Let $G = (V,E,c(E))$ be an edge weighted directed graph without negative weight cycles (e.g., a precedence graph). For any two nodes $u, v \in V$ with $u \neq v$. If $w_{uv}: (u = i_0, i_1, \ldots, i_m = v)$ is a walk such that $c_{w_{uv}} \le c_{uv}$. Then there exists a path $p_{uv} : (i_{k(1)}, i_{k(2)}, \ldots, i_{k(q)})$ where $k(1), k(2), \ldots, k(q)$ is a subsequent of $0,1,\ldots,m$ such that $k(1) = 0$ and $k(q) = m$. In addition, the weight of $p_{uv}$, denoted $c_{p_{uv}}$, satisfies $c_{p_{uv}} \le c_{w_{uv}} \le c_{uv}$.
\end{lemma}
\begin{proof}
According to Appendix \ref{app:walk_decomp}, the walk $w_{uv}$ can be decomposed into a path $p_{uv}$ of the form described in the statement and a finite number of cycles. In addition, the weight of the walk $w_{uv}$ is the sum of the weight of the path $p_{uv}$ and those of the cycles (if any). Since the weight of all cycles are nonnegative as assumed in the statement, it holds that $c_{p_{uv}} \le c_{w_{uv}}$.
\end{proof}


Now we establish the equivalence between the notions of index set of redundant relations and redundant edge set.
\begin{lemma} \label{thm:redundant_edge_set}
Let $G = (V,E,c(E))$ denote a precedence relation system and its precedence graph. Then $R \subseteq E$ is an index set of redundant relations of $G$ (Definition~\ref{def:redundant_relations}) if and only if it is a redundant edge set of $G$ (Definition~\ref{def:redundant_edge_set}).
\end{lemma}

\begin{proof}
The statement holds trivially if $R = \emptyset$. Therefore, the rest of the proof assumes that $R \neq \emptyset$. By Definition~\ref{def:redundant_relations}, $R$ is an index set of redundant relations if and only if
\begin{equation} \label{eqn:pf_red_edge_set_imp}
\begin{array}{cl}
& \text{$x \in \real^n$ satisfies $x_u - x_v \le c_{uv}, \quad \forall (u,v) \in (E \setminus R)$} \vspace{2mm} \\ \implies & x_u - x_v \le c_{uv}, \quad \forall (u,v) \in R.
\end{array}
\end{equation}
$(V, E \setminus R, c(E \setminus R))$, being a subgraph of $G$, is a precedence graph (with Assumption~\ref{asm:graph_assumptions} satisfied). By applying Lemma~\ref{thm:implication} with $(V, E \setminus R, c(E \setminus R))$ for $b_{uv} = c_{uv}$ for each $(u,v) \in R$, the condition in (\ref{eqn:pf_red_edge_set_imp}) is equivalent to
\begin{displaymath}
\begin{array}{l}
\text{for each $(u,v) \in R$ there exists a path $p_{uv} : u \leadsto v$} \vspace{2mm} \\
\text{in $(V, E \setminus R, c(E \setminus R))$ such that $c_{p_{uv}} \le c_{uv}$}.
\end{array}
\end{displaymath}
This is the same as the condition that $R$ is a redundant edge set, according to Definition~\ref{def:redundant_edge_set}.

\end{proof}

In view of Lemma~\ref{thm:redundant_edge_set}, the problem of finding the maximum index set of redundant relations in (\ref{eqn:prec_sys}) can be posed as the problem of finding the \emph{maximum (cardinality) redundant edge set} in the precedence graph associated with (\ref{eqn:prec_sys}), with the definition of redundant edge set given in (\ref{eqn:redundant_edge_set}).

\section{Maximum redundant edge set problem} \label{sec:maximum_redundant_edge_set}
This section presents the results on the maximum redundant edge set problem. In Section~\ref{subsec:complexity} we will establish that the maximum redundant edge set problem is polynomial-time solvable if the precedence graph does not have any zero-weight cycle. On the other hand, we will show that in general the problem is NP-hard. Section~\ref{subsec:decomposition} specifies that the maximum redundant edge set problem can be decomposed into several subproblems: one subproblem is polynomial-time solvable and the other subproblems are NP-hard. Section~\ref{subsec:computation} discusses some computation issues pertaining the decomposition result in Section~\ref{subsec:decomposition}.

\subsection{Complexity of maximum redundant edge set problem} \label{subsec:complexity}
This subsection answers the following basic and complexity related questions regarding the maximum redundant edge set problem.

\begin{enumerate}
\item When is there a nonempty redundant edge set?
\item If there exists at least one redundant edge set, when is the maximum redundant edge set unique and how to compute it?
\item When is it computationally intractable to solve the maximum redundant edge set problem?
\end{enumerate}
The answer to the first question lies in the following concept:
\begin{definition}[Redundant edge] \label{def:redundant_edge}
For an edge weighted directed graph $G = (V,E,c(E))$, an edge $(i,j) \in E$ is called a redundant edge of $G$ if the singleton $\{(i,j)\}$ is a redundant edge set of $G$ (defined in (\ref{eqn:redundant_edge_set})).
\end{definition}
\begin{remark}
In general, an arbitrary set of redundant edges need not be a redundant edge set defined in (\ref{eqn:redundant_edge_set}) in Definition~\ref{def:redundant_edge}, because the ``replacement path'' of one redundant edge might contain another redundant edge.
\end{remark}
\begin{lemma} \label{thm:nonempty_redundant_edge_set}
Let $G = (V,E,c(E))$ be an edge weighted directed graph. Then, there exists a nonempty redundant edge set in $G$ if and only if there exists a redundant edge in $G$.
\end{lemma}
\begin{proof}
The sufficiency part is due to the fact that a redundant edge leads to a single-member redundant edge set. For the necessity part, note that by definition in (\ref{eqn:redundant_edge_set}) any subset of a redundant edge set is also a redundant edge set. Hence, nonexistence of redundant edge implies nonexistence of single-member and in general any redundant edge sets.
\end{proof}

The answer to the second question is provided by the following statement (Lemma~\ref{thm:union}) and its corollary (Theorem~\ref{thm:max_redundant_edge_set}), whose preliminary version appeared in \cite{SSJ_DP2014_part1}:
\begin{lemma} \label{thm:union}
Let $G = (V,E,c(E))$ be an edge weighted directed graph. Assume that all cycles in $G$ have positive weights. Then, if $R$ and $R^\prime$ are redundant edge sets of $G$ satisfying (\ref{eqn:redundant_edge_set}), $R \cup R^\prime$ is also a redundant edge set of $G$.
\end{lemma}
\begin{proof}
To verify that $R \cup R^\prime$ is indeed a redundant edge set, it is sufficient to verify that for each $(u,v) \in R \cup R^\prime$ with weight $c_{uv}$, it holds that
\begin{equation} \label{eqn:redundant_edge_cond}
\begin{array}{l}
\text{$\exists$ a path $p_{uv}$ in $(V, E \setminus (R \cup R^\prime), c(E \setminus (R \cup R^\prime)))$ from $u$ to $v$} \vspace{1mm} \\
\text{such that $c_{p_{uv}} \le c_{uv}$.}
\end{array}
\end{equation}
In the subsequent parts of the proof, the following shorthand is used: the sentence ``a walk $w$ is in $\hat{E} \subseteq E$'' means that $w$ is in the subgraph $(V, \hat{E}, c(\hat{E}))$. Without loss of generality, assume that $(u,v) \in R$ (otherwise we exchange the roles of $R$ and $R'$). Then by (\ref{eqn:redundant_edge_set}) there exists a walk (which is in fact a path) $w_1 :u \leadsto v$ in $E \setminus R$ such that $c_{w_1} \le c_{uv}$. If $w_1$ is also in $E \setminus R^\prime$ then $w_1$ is in $E \setminus (R \cup R^\prime)$, and hence $(u,v)$ satisfies (\ref{eqn:redundant_edge_cond}). If, on the other hand, $w_1$ involves edges belonging to $R^\prime$, then by (\ref{eqn:redundant_edge_set}) each of these edges can be replaced by a corresponding path in $E \setminus R^\prime$. Also by (\ref{eqn:redundant_edge_set}), the weights of the replacement paths are no more than the weights of the corresponding edges. This results in another walk $w_2 : u \leadsto v$ in $E \setminus R^\prime$, with possible edges in $R$. The walk $w_2$ has strictly more nodes (and edges) than $w_1$, and $c_{w_2} \le c_{w_1}$. If $w_2$ is in $E \setminus R$ then $(u,v)$ satisfies (\ref{eqn:redundant_edge_cond}), as argued above. Otherwise, the process of finding replacement walks $w_3, w_4, \ldots$ with increasing number of nodes and nonincreasing weights would continue. Next, we show by contradiction that the replacement-walk-finding process would terminate in a finite number of iterations. It is noted that any walk from $u$ to $v$ (with $u \neq v$) can be decomposed into one path from $u$ to $v$ and a finite number of cycles (see Appendix~\ref{app:walk_decomp} for a proof). In addition, in a finite graph the numbers of possible paths and cycles are finite. Therefore, by the pigeonhole principle, if the replacement-walk-finding process does not terminate at some iteration it will generate a walk $\tilde{w}$ traversing a cycle, and one of the previously constructed walks, denoted $\hat{w}$, is exactly the same as $\tilde{w}$ except that the cycle is not traversed. The construction of the walks specifies that $c_{\tilde{w}} \le c_{\hat{w}}$, and this implies that the cycle has nonpositive weight. This contradicts the assumption that all cycles in $G$ have positive weights. Therefore, the replacement-walk-finding process terminates in a finite number of iterations. Consequently, a walk $w^\star$ in $E \setminus (R \cup R')$ from $u$ to $v$ with $c_{w^\star} \le c_{uv}$ is resulted. Then, by Lemma~\ref{thm:walk2path} the desired path $p_{uv}$ can be constructed from $w^\star$ to satisfy (\ref{eqn:redundant_edge_cond}). Applying the same proof to all members of $R \cup R^\prime$ completes the proof.
\end{proof}
\begin{theorem} \label{thm:max_redundant_edge_set}
Let $G = (V,E,c(E))$ be an edge weighted directed graph. Assume that all cycles in $G$ have positive weights. Then the maximum redundant edge set is unique. In addition, if there is a redundant edge, then the maximum redundant edge set is the set of all redundant edges, which can be computed in $O({|V|}^3)$ time.
\end{theorem}
Before the proof is given, the procedure to compute all redundant edges when $G$ does not have zero or negative weight cycle is described first.
\begin{algorithm}[Finding all redundant edges in graph $G = (V,E,c(E))$ without zero or negative weight cycles] \label{alg:redundant_edges} \hspace{0cm}
\begin{enumerate}
\item Solve the all-pair shortest path problem for all source/destination pairs in $G$. Let $d_{ij}$ denote the shortest path distance from $i$ to $j$.
\item An edge $(i,j) \in E$ is declared a redundant edge if and only if 
\begin{equation} \label{eqn:redundant_edge_criterion}
\min\limits_{(i,k) \in E, \; k \neq i, \; k \neq j} \big\{ c_{ik} + d_{kj} \big\} \le c_{ij}.
\end{equation}
\end{enumerate}
\end{algorithm}

\begin{lemma} \label{thm:computing_redundant_edges}
When $G = (V,E,c(E))$ does not have any zero or negative weight cycles, Algorithm~\ref{alg:redundant_edges} correctly computes all redundant edges in $O({|V|}^3)$ time.
\end{lemma}
\begin{proof}
According to Definition~\ref{def:redundant_edge}, edge $(i,j)$ is a redundant edge if and only if there exists a path $p_{ij} : i \leadsto j$ in $G$ such that $(i,j)$ is not part of $p_{ij}$ and $c_{p_{ij}} \le c_{ij}$. If (\ref{eqn:redundant_edge_criterion}) does not hold, then except possibly $(i,j)$ there is no path in $G$ from $i$ to $j$ with weight less than or equal to $c_{ij}$. Hence, $(i,j)$ cannot be a redundant edge. On the other hand, if (\ref{eqn:redundant_edge_criterion}) holds, then there exists a walk $w_{ikj}: i \rightarrow k \leadsto j$ such that $c_{w_{ikj}} \le c_{ij}$. If $(i,j)$ is part of $w_{ikj}$ then the walk is of the form $i \rightarrow k \leadsto i \rightarrow j \leadsto j$, and its weight is $c_{w_{ikj}} > c_{ij}$ because the closed walks $i \rightarrow k \leadsto i$ and $j \leadsto j$ can be decomposed into sequences of cycles and the cycles in $G$ are positively weighted as assumed in the statement. This is a contradiction because $c_{w_{ikj}} \le c_{ij}$ and $c_{w_{ikj}} > c_{ij}$ cannot be both true. Hence, $(i,j)$ is not part of $w_{ikj}$. From Lemma~\ref{thm:walk2path} the desired path $p_{ij}: i \rightarrow k \leadsto j$ can be found to certify that $(i,j)$ is indeed a redundant edge.

Next we argue for the computation requirement. Since $G$ does not have any negative weight cycle. The all-pair shortest path problem can be solved using, for instance, the Floyd-Warshall algorithm in $O({|V|}^3)$ time (e.g., \cite{Cormen:2009:IAT:1614191}). The third step requires $O(|E| |V|) = O({|V|}^3)$ computation cost.
\end{proof}

\begin{remark}
In Lemma~\ref{thm:computing_redundant_edges} the assumption of no zero or negative weight cycles cannot be removed. See Figure~\ref{fig:zero_weight_cycle} for a consequence of when the assumption is not satisfied.
\begin{figure}[!tbh]
\begin{center}
 \includegraphics[width=60mm]{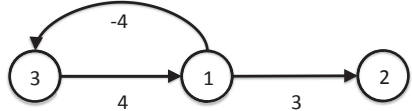}
\end{center}
 \caption{A 3-node example of a graph with a zero weight cycle $(1,3,1)$. In this example, $c_{12} = 3$, $c_{13} = -4$ and $d_{32} = 7$. In (\ref{eqn:redundant_edge_criterion}) for edge $(1,2)$, $c_{13} + d_{32} = 3 \le c_{12}$. However, $(1,2)$ is not a redundant edge.}
 \label{fig:zero_weight_cycle}
\end{figure}
\end{remark}

\begin{proof}[Proof of Theorem~\ref{thm:max_redundant_edge_set}]
Suppose $R$ and $R^\prime$ are two different different maximum redundant edge sets. Then, $|R \cup R^\prime| > |R|$. Further, by Lemma~\ref{thm:union} $R \cup R^\prime$ is a redundant edge set which has more edges than $R$, contradicting the assumption that $R$ is a maximum redundant edge set. Thus, the maximum redundant edge set is unique. It is denoted as $R^\star$.

Next, let $U$ denote the set of all redundant edges. Under the additional assumption that $U \neq \emptyset$, Lemma~\ref{thm:nonempty_redundant_edge_set} specifies that $R^\star \neq \emptyset$. Then, each $(u,v) \in R^\star$ is a redundant edge because the singleton $\{(u,v)\} \subseteq R^\star$ is a redundant edge set. Therefore, $R^\star \subseteq U$. On the other hand, Lemma~\ref{thm:union} states that $U$, which is the union of single-member redundant edge sets, is also a redundant edge set. Thus, $U \subseteq R^\star$. Finally, the claim about $O({|V|}^3)$ computation time follows from Lemma~\ref{thm:computing_redundant_edges}.
\end{proof}

The positive weight cycle assumption in Lemma~\ref{thm:union} is necessary. The presence of zero or negative weight cycles can indeed results in situations where the union of two redundant edge sets is not a redundant edge set, and as a result the maximum redundant edge set is not unique. For a counterexample, consider the graph in Figure~\ref{fig:non_uniqueness}.
\begin{figure}[!tbh]
\begin{center}
 \includegraphics[width=40mm]{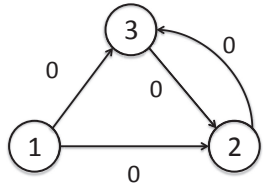}
\end{center}
 \caption{A 3-node example of a graph with zero weight cycles. There are two different maximum redundant edge sets, namely $\{(1,2)\}$ and $\{(1,3)\}$. However, the union $\{(1,2), (1,3)\}$ is not a redundant edge set according to (\ref{eqn:redundant_edge_set}). The cycle $(2,3,2)$ has zero weight.}
 \label{fig:non_uniqueness}
\end{figure}

In fact, the general maximum redundant edge set problem (without the assumption in Lemma~\ref{thm:union}) is NP-hard. This can be shown by a reduction from the NP-hard minimum equivalent graph problem studied in \cite{Moyles:1969:AFM:321526.321534}. The proof also establishes the connection that the maximum redundant edge set problem is a generalization of the minimum equivalent graph problem.

\begin{theorem} \label{thm:NP_hard}
Let $G$ be an edge weighted directed graph. The problem of finding the maximum redundant edge set of $G$ is NP-hard.
\end{theorem}
\begin{proof}
First, the minimum equivalent graph problem in \cite{Moyles:1969:AFM:321526.321534} is summarized, in the context of our discussion. Consider a directed graph $(V,E)$, and let $\hat{R} \subseteq E$. We mention that $(V,E)$ and $(V,E \setminus \hat{R})$ have the same reachability if the following condition is satisfied: there is a walk (hence a path) from $i \in V$ to $j \in V$ in $(V,E)$ if and only if there is a walk from $i$ to $j$ in $(V,E \setminus \hat{R})$. Then, an instance of the minimum equivalent graph problem (associated with $(V,E)$) seeks the maximum cardinality $R^\star \subseteq E$ such that $(V,E)$ and $(V,E \setminus R^\star)$ have the same reachability.

Next, we show that every instance of the minimum equivalent graph problem can be reduced into an instance of the maximum redundant edge set problem. To begin, we claim that
\begin{subequations} \label{eqn:reachability_rp}
\begin{align}
& \;\; \text{$(V,E)$ and $(V,E \setminus \hat{R})$ have the same reachability}. \label{eqn:reachability_rp_a} \\
\iff & \;\; \text{$\forall (i,j) \in \hat{R}$, there exists a walk from $i$ to $j$ in $(V, E \setminus \hat{R})$}. \label{eqn:reachability_rp_b}
\end{align}
\end{subequations}
Since a walk in $(V,E \setminus \hat{R})$ is a walk in $(V,E)$, condition (\ref{eqn:reachability_rp_a}) is the same as the condition that 
\begin{equation} \label{eqn:reachability_rp_a_eq}
\text{$\exists$ walk $w_{ij} : i \leadsto j$ in $(V,E)$} \implies \text{$\exists$ walk $w_{ij}^r : i \leadsto j$ in $(V, E \setminus \hat{R})$}.
\end{equation}
If (\ref{eqn:reachability_rp_b}) holds, then in (\ref{eqn:reachability_rp_a_eq}) for every edge $(u,v) \in \hat{R}$ that is part of $w_{ij}$ there is a walk $w_{uv}^r : u \leadsto v$ in $(V, E \setminus \hat{R})$. Hence, (\ref{eqn:reachability_rp_a_eq}) holds. On the other hand, suppose (\ref{eqn:reachability_rp_b}) does not hold, and let $(s,t) \in \hat{R}$ be an edge such that there is no walk from $s$ to $t$ in $(V,E \setminus \hat{R})$. Then, with $i = s$, $j = t$ and $w_{ij} = (s,t)$ as a counterexample it can be seen that (\ref{eqn:reachability_rp_a_eq}) does not hold. Therefore, (\ref{eqn:reachability_rp_a}), (\ref{eqn:reachability_rp_b}) and (\ref{eqn:reachability_rp_a_eq}) are all equivalent. Further, we define $(V,E,c^0(E))$ where $c_{ij}^0 = 0$ for all $(i,j) \in E$ (in fact, the following argument would hold as long as $c_{ij}^0 = \alpha$ for $\alpha \le 0$). Since any walk in $(V,E)$ is a zero weight walk in $(V,E,c^0(E))$ and vice versa, it can be seen that (\ref{eqn:reachability_rp_b}) is equivalent to the condition that $\hat{R}$ is a redundant edge set in $(V,E,c^0(E))$. This suggests that an instance of minimum equivalent graph problem with $(V,E)$ is equivalent to the instance of maximum redundant edge set problem with $(V,E,c^0(E))$, with the two problem instances having the same optimal solutions. 

To complete the complexity argument, note that by definition of $c^0(E)$, $(V,E)$ contains a cycle if and only if $(V,E,c^0(E))$ contains a zero weight cycle. Thus, if there would be a polynomial time algorithm which can solve all instances of the maximum redundant edge set problem including those with zero or negative weight cycles, then the minimum equivalent graph problem could be solved in polynomial time as well. However, since in general the minimum equivalent graph problem is NP-hard (e.g., \cite{Garey:1990:CIG:574848}), we establish that in general the maximum redundant edge set problem is NP-hard as well.
\end{proof}

In summary, if $G$ is a precedence graph then by standing Assumption~\ref{asm:no_negative_cycle} $G$ does not have any negative weight cycles. If in addition $G$ does not have any zero-weight cycles, then Theorem~\ref{thm:max_redundant_edge_set} states that the maximum redundant edge set of $G$ is unique, and it is the set of all redundant edges (if the set is nonempty). In addition, finding all redundant edges using Algorithm~\ref{alg:redundant_edges} requires $O({|V|}^3)$ time. On the other hand, in the more general case where $G$ is allowed to have zero-weight cycles, the example in Figure~\ref{fig:non_uniqueness} indicates that the maximum redundant edge set need not be unique. In addition, Theorem~\ref{thm:NP_hard} states that the maximum redundant edge set problem is NP-hard in general.

It turns out that the maximum redundant edge set problem can always be decomposed into a finite number of decoupled subproblems, one of which is solvable in polynomial time and all other are NP-hard. This decomposition, which will be detailed in Section~\ref{subsec:decomposition}, is analogous to that in \cite{Moyles:1969:AFM:321526.321534} for the minimum equivalent graph problem for unweighted directed graphs.

\subsection{Decomposition of maximum redundant edge set problem} \label{subsec:decomposition}
In this subsection, we first introduce an equivalence class partitioning of the node set of a precedence graph, and define an auxiliary graph induced by the equivalence class partitioning called condensation. Next, we present some properties of the equivalence classes and the condensation. After that, we establish the fact that the maximum redundant edge set problem can be decomposed into $K+1$ independent subproblems, where $K$ is the number of equivalence classes. The main result will be summarized in Theorem~\ref{thm:decomposition}. 

In any edge weighted directed graph, we define an equivalence relation on the node set as follows:
\begin{definition}[Equivalence relation $\sim$] \label{def:eq_re}
Let $G = (V,E,c(E))$ denote an edge weighted directed graph. For any pair $i \in V$, $j \in V$, we denote $i \sim j$ if either (a) $i = j$, or (b) there exists a zero-weight closed walk in $G$ traversing both $i$ and $j$.
\end{definition}
\begin{remark}
To verify that $\sim$ is indeed an equivalence relation, it suffices to note that if $i \sim j$ and $j \sim k$ then there exist two zero-weight closed walks $i \leadsto j \leadsto i$ and $j \leadsto k \leadsto j$. Consequently, a zero-weight closed walk $i \leadsto j \leadsto k \leadsto j \leadsto i$ exists, and this implies that $i \sim k$.
\end{remark}

The relation $\sim$ defines equivalence classes in the node set. For convenience, we will define some notations associated with the equivalence classes. However, before these notations are defined, the notion of the minimum walk weight should be defined first.

\begin{definition}[Minimum walk weight] \label{def:min_walk_weight}
Let $G = (V,E,c(E))$ be an edge weighted directed graph without negative weight closed walks. For $i \in V$, $j \in V$, define $d_{ij}$ to be the minimum weight of the walk among all walks in $G$ which goes from $i$ to $j$. Note that $d_{ii} = 0$ for all $i \in V$, and this is attained by the single-node path $(i)$ since $G$ does not have any negative weight closed walks.
\end{definition}

\begin{definition}[Equivalence classes induced by equivalence relation $\sim$] \label{def:eq_class_notations}
Let $G = (V,E,c(E))$ denote an edge weighted directed graph without negative weight closed walks. Let relation $\sim$ be defined in Definition~\ref{def:eq_re}. In addition, let $d_{ij}$, the minimum walk weight in $G$, be defined in Definition~\ref{def:min_walk_weight}. We define the following:
\begin{enumerate}[label=\ref{def:eq_class_notations}.\alph*]
\item\label{def:eq_K} $K$ denotes the number of equivalence classes in $V$ defined by relation $\sim$.
\item\label{def:eq_vk} For $k \in \{1,2,\ldots,K\}$, $[v_k] \subseteq V$ denotes the equivalence class containing $v_k$, where $v_k$ is the (arbitrarily) designated representing node for equivalence class containing $v_k$. 
\item\label{def:eq_Ek} For $k \in \{1,2,\ldots,K\}$, we denote
\begin{itemize}
\item $E_k := \{(i,j) \in E \mid i \in [v_k], j \in [v_k]\}$. That is, $E_k$ denotes the set of edges connecting two nodes inside an equivalence class $[v_k]$.
\item Let $E_k^r \subseteq E_k$ be defined as
\begin{equation} \label{def:EkR}
E_k^r := \{(i,j) \in E_k \mid c_{ij} > d_{ij} \}.
\end{equation}
That is, $E_k^r \subseteq E_k$ is a subset of $E_k$ where each member edge has an edge weight strictly greater than the corresponding minimum walk weight. As it will become apparent in the sequel, edges in $E_k^r$ are redundant (i.e., can always be included in any maximum redundant edge set). This motivates the use of superscript ``r'' in (\ref{def:EkR}).
\item Let $E_k^c \subseteq E_k$ be defined as
\begin{equation} \label{def:Ekc}
E_k^c := \{(i,j) \in E_k \mid c_{ij} = d_{ij} \}.
\end{equation}
That is, $E_k^c = E_k \setminus E_k^r$ since $c_{ij} \ge d_{ij}$ for all $(i,j) \in E$. This motivates the superscript ``c'', since $E_k^c$ is the complement of $E_k^r$.
\end{itemize}
\item\label{def:eq_Eij} For $i \in \{1,2,\ldots,K\}, j \in \{1,2,\ldots,K\}$, $i \neq j$ we denote
\begin{itemize}
\item $E_{ij} := \{(u,v) \in E \mid u \in [v_i], \; v \in [v_j] \}$. That is, $E_{ij}$ denotes the set of edges from a node in $[v_i]$ to another node in $[v_j]$ with $[v_i] \neq [v_j]$ (because of the assumption that $i \neq j$).
\item Let $E_{ij}^c \subseteq E_{ij}$ be defined as
\begin{equation} \label{def:Eijc}
E_{ij}^c := \{(u,v) \in E_{ij} \mid (u,v) \in \underset{(s,t) \in E_{ij}}{\text{argmin}} \; d_{v_i s} + c_{st} + d_{t v_j} \}.
\end{equation}
While in (\ref{def:Eijc}) the definition of $E_{ij}^c$ assumes a designation of the representing nodes $v_i$ and $v_j$, it turns out that $E_{ij}^c$ is in fact independent of the choice of the designation. This will be argued in Remark~\ref{rmk:decomposition}.
\item For $E_{ij} \neq \emptyset$ (hence $E_{ij}^c \neq \emptyset$), we (arbitrarily) designate a particular edge $(v_{ij}^i, v_{ij}^j) \in E_{ij}^c$ (with $v_{ij}^i \in [v_i]$, $v_{ij}^j \in [v_j]$) as the ``representing'' edge for $E_{ij}$. 
\end{itemize}
\item\label{def:eq_E0} Collecting all inter-equivalence class edges, we denote
\begin{equation} \label{eqn:E0}
\begin{aligned}
E_0 & :=  \mathop{\cup}\limits_{1 \le i \neq j \le K} E_{ij} \\
E_0^c & :=  \mathop{\cup}\limits_{1 \le i \neq j \le K} E_{ij}^c \\
E_0^d & :=  \mathop{\cup}\limits_{1 \le i \neq j \le K} \{(v_{ij}^i, v_{ij}^j)\}
\end{aligned},
\end{equation}
where $1 \le i \neq j \le K$ is shorthand for $\{(i,j) \mid 1 \le i,j \le K, i \neq j \}$. It holds that $E_0^d \subseteq E_0^c \subseteq E_0$.
\end{enumerate}
\end{definition}

Figure~\ref{fig:eq_classes} shows an illustration of the equivalence classes defined by relation $\sim$.
\begin{figure}[!tbh]
\begin{center}
 \includegraphics[width=60mm]{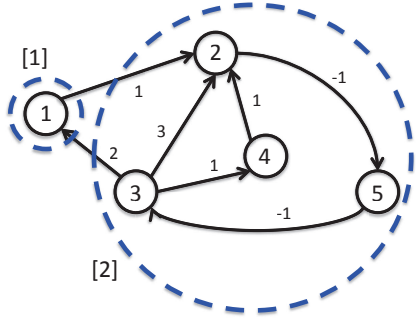}
 \caption{An example precedence graph with five nodes. The number on each edge denotes the weight of the corresponding edge. Node 1 defines an equivalence class $[1]$ by itself because it is not on any zero weight closed walk. $\{2,3,4,5\}$ is the other equivalence class $[2]$ because of the zero-weight cycle (i.e., closed walk) $(3,4,2,5,3)$. In constrast, for the simplification of unweighted directed graphs, references \cite{Moyles:1969:AFM:321526.321534,AGU_transitive_reduction} consider the partitioning of $V$ into a different type equivalence classes which are the strongly connected components of the graph. For this example graph, there is only one strongly connected component which is the set of all nodes. This is different from the two-part equivalence class partitioning induced by the $\sim$ relation considered in this paper. The intra-equivalence class edge sets are $E_1 = \emptyset$, $E_2 = \{(2,5), (3,2), (3,4), (4,2), (5,3)\}$ and $E_2^r = \{(3,2)\}$. The inter-equivalence class edge sets are $E_{12} = \{(1,2)\}$ and $E_{21} = \{(3,1)\}$. The representing edges for $E_{12}$ and $E_{21}$ are, respectively, $(v_{12}^1, v_{12}^2) = (1,2)$ and $(v_{21}^2, v_{21}^1) = (3,1)$. In this example, $E_0 = E_0^c = E_0^d = \{(1,2), (3,1)\}$.
 }
 \label{fig:eq_classes}
\end{center}
\end{figure}
Analogous to the condensation of a unweighted graph, we define the condensation of an edge weighted directed graph as follows:

\begin{definition}[Condensation] \label{def:condensation}
Let $G = (V,E,c(E))$ denote an edge weighted directed graph without negative weight closed walks. Let other involved symbols be defined in Definitions~\ref{def:eq_re}, \ref{def:min_walk_weight}, \ref{def:eq_class_notations} in the context of $G$. We define the condensation of $G$, denoted $\tilde{G} := (\tilde{V}, \tilde{E}_0, \tilde{c}(\tilde{E}_0))$ as follows:
\begin{itemize}
\item The set of all nodes of $\tilde{G}$ is $\tilde{V} := \{v_1, v_2, \ldots, v_K\}$.
\item In $\tilde{G}$, there is an edge $(v_i, v_j) \in \tilde{E}_0$ with $i \neq j$ if and only if in $G$ the set $E_{ij} \neq \emptyset$ (see Definition~\ref{def:eq_Eij} for $E_{ij}$).
\item For any $(v_i, v_j) \in \tilde{E}_0$, the edge weight $\tilde{c}_{v_i v_j}$ is defined as
\begin{equation} \label{eqn:tilde_c_ij}
\tilde{c}_{v_i v_j} := \min\limits_{(u,v) \in E_{ij}} \;\; d_{v_i u} + c_{u v} + d_{v v_j},
\end{equation}
where $d_{v_i u}$ and $d_{v v_j}$ are defined in Definition~\ref{def:min_walk_weight}. 
\item According to the definition of $v_{ij}^i$ and $v_{ij}^j$ in Definition~\ref{def:eq_Eij}, it holds that
\begin{equation} \label{eqn:viji_vijj}
\begin{array}{l}
\tilde{c}_{v_i v_j} = d_{v_i v_{ij}^i} + c_{v_{ij}^i v_{ij}^j} + d_{v_{ij}^j v_j}.
\end{array}
\end{equation}
\item The vector of edge weights is denoted $\tilde{c}(\tilde{E}_0)$.
\end{itemize}
\end{definition}

\begin{remark}
For a graph $G$, different designations of the representing nodes in the equivalence classes (e.g., $v_1, v_2, \ldots, v_K$) can result in different condensations $\tilde{G}$. However, certain properties of $\tilde{G}$ vital to the main results in this paper are independent of the designation. See Remarks~\ref{rmk:decomposition} and \ref{rmk:equiv_reduction} for details.
\end{remark}
The condensation of the graph in Figure~\ref{fig:eq_classes} is illustrated in Figure~\ref{fig:condensation_example}.
\begin{figure}[!tbh]
\begin{center}
 \includegraphics[width=36mm]{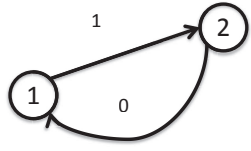}
 \caption{The condensation of the graph in Figure~\ref{fig:eq_classes}. There are two edges in the condensation: $(v_1,v_2) = (1,2)$ and $(v_2,v_1) = (2,1)$ (i.e., $\tilde{E}_0 = \{(1,2), (2,1)\}$). Their weights are $\tilde{c}_{v_1 v_2} = d_{v_1 v_{12}^1} + c_{v_{12}^1 v_{12}^2} + d_{v_{12}^2 v_2} = 0 + 1 + 0 = 1$, and $\tilde{c}_{v_2 v_1} = d_{v_2 v_{21}^2} + c_{v_{21}^2 v_{21}^1} + d_{v_{21}^1 v_1} = -2 + 2 + 0 = 0$.}
 \label{fig:condensation_example}
\end{center}
\end{figure}

The equivalence relation $\sim$ (Definition~\ref{def:eq_re}), the equivalence classes (Definition~\ref{def:eq_class_notations}) and the condensation (Definition~\ref{def:condensation}) satisfy certain properties that will be useful in the proof of the main results in this subsection. Lemma~\ref{thm:sim_preservation} and Lemma~\ref{thm:eq_class_preservation} specify that relation $\sim$ and the corresponding equivalence class partitioning are preserved even if a subset of edges is removed from the graph, as long as the removed edges form a redundant edge set.

\begin{lemma} \label{thm:sim_preservation}
Let $G = (V,E,c(E))$ denote an edge weighted directed graph without negative weight closed walks, and let $R$ be any redundant edge set of $G$ (Definition~\ref{def:redundant_edge_set}). Denote $G^c = (V,E \setminus R, c(E \setminus R))$. Let $i \in V$ and $j \in V$. Then $i \sim j$ in $G$ if and only if $i \sim j$ in $G^c$.
\end{lemma}
\begin{proof}
The proof considers only the case when $i \neq j$, since the statement is true by Definition~\ref{def:eq_re} when $i = j$. If $i \sim j$ in $G^c$ then $i$ and $j$ are on a zero-weight closed walk in $G^c$. The same closed walk is also in $G$ because $(E \setminus R) \subseteq E$. Hence, $i \sim j$ in $G$. On the other hand, if $i \sim j$ in $G$, then there exists a zero-weight closed walk $w$ in $G$ traversing $i$ and $j$. Part of $w$ can be edges in $R$. However, for each $(u,v) \in R$ such that $(u,v)$ is part of $w$, a replacement path $p_{uv}$ exists in $G^c$ such that $c_{p_{uv}} \le c_{uv}$. By substituting edges in $R$ which are part of $w$ with the corresponding replacement paths in $G^c$, it is possible to construct another closed walk $\hat{w}$ (in $G^c$) traversing $i$ and $j$ satisfying $c_{\hat{w}} \le c_w = 0$. The no-negative-weight-cycle assumption in the statement excludes the case where $c_{\hat{w}} < 0$. Hence, $\hat{w}$ is a zero-weight closed walk in $G^c$ traversing $i$ and $j$. In other words, $i \sim j$ in $G^c$.
\end{proof}

\begin{lemma} \label{thm:eq_class_preservation}
Let $G = (V,E,c(E))$ an edge weighted directed graph without negative weight closed walks. The equivalence class partitioning defined by relation $\sim$ is the same for all subgraphs $(V, E \setminus R, c(E \setminus R))$, where $R \subseteq E$ is any redundant edge set in $G$.
\end{lemma}
\begin{proof}
This is a direct consequence of Lemma~\ref{thm:sim_preservation}.
\end{proof}

The following preliminary statements are also useful in the proof of the main results. Lemma~\ref{thm:eq_class} establishes some properties regarding the minimum walk weights (and the corresponding walks) between nodes in an equivalence class. Lemma~\ref{thm:tildeG_positive_cycle} specifies that the cycles in the condensation are always positively weighted, even if the graph from which the condensation is derived can have zero-weight cycles.

\begin{lemma} \label{thm:eq_class}
Let $G = (V,E,c(E))$ be an edge weighted directed graph without negative weight closed walks. For $i,j \in V$, let $d_{ij}$ be the minimum walk weight from $i$ to $j$ among all walks in $G$ (i.e., Definition~\ref{def:min_walk_weight}). Let $V^\prime \subseteq V$, $E^\prime \subseteq E$. Let $U \subseteq V^\prime$ be an equivalence class in $(V^\prime, E^\prime, c(E^\prime))$ defined by relation $\sim$. Define $E_U := \{(i,j) \in E^\prime \mid i \in U, j \in U \}$, and $G_U = (U,E_U,c(E_U))$. For all $i \in U$, $j \in U$, the following statements hold:
\begin{enumerate}[label=\ref{thm:eq_class}.\alph*]
\item\label{lem:wij} There exists a walk $w_{ij}$ in $G_U$ (which is a subgraph of $G^\prime$ which in turn is a subgraph of $G$) attaining the minimum weight $d_{ij}$ among all walks in $G$ which goes from $i$ to $j$. Similiarly, there exists a walk $w_{ji}$ in $G_U$ attaining the minimum weight $d_{ji}$ among all walks in $G$ which goes from $j$ to $i$. In case $i = j$, $d_{ii} = 0$ is attained by the degenerate walk containing the single node $i$.
\item\label{lem:dij_dji} $d_{ij} = -d_{ji}$.
\item\label{lem:dis_dsj} If $s \in U$, then $d_{ij} = d_{is} + d_{sj}$.
\end{enumerate}
\end{lemma}
\begin{proof}
First, \ref{lem:wij} and \ref{lem:dij_dji} are shown together. If $i = j$, then $d_{ij} = -d_{ji} = 0$. Thus, we only consider the case when $i \neq j$. By Definition~\ref{def:eq_re}, $i \in U$ and $j \in U$ means that there exists a zero-weight closed walk $w_{ii}: i \leadsto j \leadsto i$ in $G_U$. This closed walk can be decomposed into two walks in $G_U$: $w_{ij}: i \leadsto j$ and $w_{ji}: j \leadsto i$ (there can be multiple ways to decompose). Let $c_{w_{ii}} = 0$, $c_{w_{ij}}$ and $c_{w_{ji}}$ denote the weights of the walks respectively. Then the decomposition of $w_{ii}$ implies that $c_{w_{ij}} + c_{w_{ji}} = c_{w_{ii}} = 0$. Hence, $c_{w_{ij}} = - c_{w_{ji}}$. Next, we show that indeed $c_{w_{ij}} = d_{ij}$ and $c_{w_{ji}} = d_{ji}$. First, note that $c_{w_{ij}} \ge d_{ij}$ by Definition~\ref{def:min_walk_weight}. Suppose $c_{w_{ij}} > d_{ij}$, and hence there exists another walk $\hat{w}_{ij}$ in $G$ going from $i$ and $j$ such that $c_{\hat{w}_{ij}} < c_{w_{ij}}$. Then, concatenating $\hat{w}_{ij}$ and $w_{ji}$ leads to a closed walk with weight $c_{\hat{w}_{ij}} + c_{w_{ji}} = c_{\hat{w}_{ij}} - c_{w_{ij}} < 0$. This violates the no negative weight closed walk assumption in the statement. Thus, $c_{w_{ij}} = d_{ij}$. With a symmetric argument, it can be shown that $c_{w_{ji}} = d_{ji}$.

For \ref{lem:dis_dsj}, if $s = i$, $s = j$ or $i = j$ then the equality trivially holds because $d_{ii} = d_{jj} = 0$, and $d_{is} = -d_{si}$ by \ref{lem:dij_dji}. Thus, we consider only the case where $s$, $i$, $j$ are all distinct. Since $i,j,s \in U$, by \ref{lem:wij} there exist walks $w_{is}: i \leadsto s$ and $w_{sj}: s \leadsto j$ with weights $d_{is}$ and $d_{sj}$, respectively. Thus, concatenating $w_{is}$ and $w_{sj}$ yields a walk $i \leadsto s \leadsto j$ with weight $d_{is} + d_{sj}$. Consequently, by Definition~\ref{def:min_walk_weight}, $d_{ij} \le d_{is} + d_{sj}$. Next, we show that $d_{ij} < d_{is} + d_{sj}$ is impossible. Assume, on the contrary, that 
\begin{equation} \label{eqn:d_isj}
d_{ij} < d_{is} + d_{sj}
\end{equation}
holds. Let $w_{ij}$ be a minimum weight walk in $G$ attaining weight $d_{ij}$ by Definition~\ref{def:min_walk_weight}. In addition, since $i,j,s \in U$, by \ref{lem:wij} and \ref{lem:dij_dji} there exist walks in $G_U$ (and hence in $G$) $w_{js}: j \leadsto s$ and $w_{si}: s \leadsto i$ with weights $d_{js} = -d_{sj}$ and $d_{si} = -d_{is}$ respectively. Consequently, by concatenating $w_{ij}$, $w_{js}$ and $w_{si}$ we obtain a closed walk (in $G$) $i \leadsto j \leadsto s \leadsto i$ with weight $c_{w_{ij}} + c_{w_{js}} + c_{w_{si}} = d_{ij} - d_{sj} - d_{is} < 0$ according to (\ref{eqn:d_isj}). This violates the no negative weight closed walk assumption in the statement. Therefore, (\ref{eqn:d_isj}) does not hold, and $d_{ij} = d_{is} + d_{sj}$ as desired.

\end{proof}

\begin{lemma} \label{thm:tildeG_positive_cycle}
Let $G = (V,E,c(E))$ be an edge weighted directed graph without negative weight closed walks. Let $\tilde{G} := (\tilde{V}, \tilde{E}_0, \tilde{c}(\tilde{E}_0))$ be the condensation of $G$ defined in Definition~\ref{def:condensation}. Then, the weights of all cycles in $\tilde{G}$ are positive.
\end{lemma}
\begin{proof}
Let $\tilde{w}$ denote a cycle in $\tilde{G}$ as $(v_{i_0}, v_{i_1}, \ldots, v_{i_m} = v_{i_0})$. By definition of cycle, $m \ge 2$ and at least one $v_{i_q}$ for $q > 0$ is different from $v_{i_0}$. The weight of the cycle is
\begin{equation} \label{eqn:cg_cycle0}
\tilde{c}_{\tilde{w}} = \tilde{c}_{v_{i_0} v_{i_1}} + \tilde{c}_{v_{i_1} v_{i_2}} + \ldots + \tilde{c}_{v_{i_{m-1}} v_{i_m}}.
\end{equation}
By (\ref{eqn:viji_vijj}), the edge weights are
\begin{equation} \label{eqn:cg_cycle1}
\tilde{c}_{v_{i_k} v_{i_{k+1}}} = d_{v_{i_k} v_{i_k i_{k+1}}^{i_k}} + c_{v_{i_k i_{k+1}}^{i_k} v_{i_k i_{k+1}}^{i_{k+1}}} + d_{v_{i_k i_{k+1}}^{i_{k+1}} v_{i_{k+1}}}, \quad \forall k \in \{0,\ldots,m-1\}.
\end{equation}
With (\ref{eqn:cg_cycle1}), the expression in (\ref{eqn:cg_cycle0}) can be rewritten as
\begin{equation} \label{eqn:cg_cycle2}
\begin{array}{rcl}
\tilde{c}_{\tilde{w}} & = & d_{v_{i_0} v_{i_0 i_1}^{i_0}} + c_{v_{i_0 i_1}^{i_0} v_{i_0 i_1}^{i_1}} + d_{v_{i_0 i_1}^{i_1} v_{i_1}}
 + d_{v_{i_1} v_{i_1 i_2}^{i_1}} + c_{v_{i_1 i_2}^{i_1} v_{i_1 i_2}^{i_2}} + d_{v_{i_1 i_2}^{i_2} v_{i_2}} \vspace{2mm} \\
& & + \ldots + d_{v_{i_{m-1}} v_{i_{m-1} i_m}^{i_{m-1}}} + c_{v_{i_{m-1} i_m}^{i_{m-1}} v_{i_{m-1} i_m}^{i_m}} + d_{v_{i_{m-1} i_m}^{i_m} v_{i_m}}.
\end{array}
\end{equation}
The right-hand side of (\ref{eqn:cg_cycle2}) is the weight of a closed walk in $G$ of the form
\begin{equation} \label{eqn:cg_cycle3}
v_{i_0} \leadsto v_{i_0 i_1}^{i_0} \rightarrow v_{i_0 i_1}^{i_1} \leadsto v_{i_1} \leadsto v_{i_1 i_2}^{i_1} \rightarrow v_{i_1 i_2}^{i_2} \leadsto \dots \rightarrow v_{i_{m-1} i_0}^{i_0} \leadsto v_{i_0},
\end{equation}
where the existence of the intra equivalence class walks and the inter equivalence class edges is guaranteed by  Lemma~\ref{lem:wij} and Definition~\ref{def:condensation}, respectively. $\tilde{c}_{\tilde{w}} < 0$ is impossible because of the statement assumption. In addition, if $\tilde{c}_{\tilde{w}} = 0$ then the closed walk (in $G$) in (\ref{eqn:cg_cycle3}) would have zero weight. Consequently, $v_{i_0}, v_{i_1}, \ldots$, which are the representing nodes of different equivalence classes, would be all traversed by one zero-weight closed walk. This is a contradiction. Therefore, $\tilde{c}_{\tilde{w}} > 0$ as desired.
\end{proof}

Now we begin to analyze and characterize the maximum redundant edge set problem for an edge weighted directed graph $G = (V,E,c(E))$. The main result in this subsection is concerned with a decomposition of the set of decision variables (i.e., $E$), induced by the equivalence class partitioning in Definition~\ref{def:eq_class_notations}:
\begin{equation} \label{eqn:E_decomp}
E = E_0 \cup (E_1^r \cup E_2^r \cup \ldots \cup E_K^r) \cup (E_1^c \cup E_2^c \cup \ldots \cup E_K^c).
\end{equation}
In the following, Lemma~\ref{thm:res_union} and Lemma~\ref{thm:cij_gr_dij} are first introduced as components of the proof of subsequent lemmas. After that, Lemma~\ref{thm:EkR} states that $(E_1^r \cup E_2^r \cup \ldots \cup E_K^r)$ should always be included in any maximum redundant edge set of $G$. Lemma~\ref{thm:Eijc_R} delivers a similar but less straightforward result. Apart from other properties to be discussed, Lemma~\ref{thm:Eijc_R} states that at most one member for each $E_{ij}$ (recall that $E_0 = \mathop{\cup}\limits_{i,j} E_{ij}$) can be excluded from any maximum redundant edge set.

\begin{lemma} \label{thm:res_union}
Let $G = (V,E,c(E))$ be an edge weighted directed graph without negative weight closed walks, and let $R$ be a redundant edge set of $G$. Suppose $A \subseteq E$ satisfies the property that for each $(i,j) \in A$ with edge weight $c_{ij}$, the following condition holds:
\begin{equation} \label{eqn:resu_rp_ij}
\text{$\exists \; w_{ij} : i \leadsto j$ in $\Big(V,E \setminus (R \cup A), c\big(E \setminus (R \cup A)\big)\Big)$ such that $c_{w_{ij}} \le c_{ij}$}.
\end{equation}
Then, $R \cup A$ is also a redundant edge set of $G$.
\end{lemma}
\begin{proof}
If $R = \emptyset$ then the statement is true because (\ref{eqn:resu_rp_ij}) is a restatement of Definition~\ref{def:redundant_edge_set} for $A$. Similarly, the statement holds trivially when $A = \emptyset$. Hence, for the rest of the proof we assume $R \neq \emptyset$, $A \neq \emptyset$. Since $R$ is a redundant edge set of $G$, by Definition~\ref{def:redundant_edge_set} for each $(u,v) \in R$ there exists a path $\hat{p}_{uv} : u \leadsto v$ in $(V, E \setminus R, c(E \setminus R))$ satisfying $c_{\hat{p}_{uv}} \le c_{uv}$. The path $\hat{p}_{uv}$ might include as parts the edges in $A$. However, we can substitute each $(i,j) \in A$ that is part of $\hat{p}_{uv}$ with the corresponding replacement walk $w_{ij}$ in (\ref{eqn:resu_rp_ij}). The outcome is a walk $w_{uv} : u \leadsto v$ in $(V, E \setminus (R \cup A), c(E \setminus (R \cup A)))$ such that $c_{w_{uv}} \le c_{\hat{p}_{uv}} \le c_{uv}$. Furthermore, by applying Lemma~\ref{thm:walk2path} with $w_{uv}$, we establish that for all $(u,v) \in R$,
\begin{equation} \label{eqn:resu_rp_uv}
\text{$\exists \; p_{uv} : u \leadsto v$ in $\Big(V,E \setminus (R \cup A), c\big(E \setminus (R \cup A)\big)\Big)$ s.t.~$c_{p_{uv}} \le c_{uv}$}.
\end{equation}
Combing (\ref{eqn:resu_rp_ij}) (again, with an application of Lemma~\ref{thm:walk2path}) and (\ref{eqn:resu_rp_uv}) yields the desired statement that $R \cup A$ is a redundant edge set of $G$.
\end{proof}

\begin{lemma} \label{thm:cij_gr_dij}
Let $G = (V,E,c(E))$ be an edge weighted directed graph without negative weight closed walks. For $u,v \in V$, let $d_{uv}$ be the minimum walk weight from $u$ to $v$ defined in Definition~\ref{def:min_walk_weight} for $G$. Let $R$ be any redundant edge set of $G$. If $(i,j) \in E$ with edge weight satisfying $c_{ij} > d_{ij}$, then $R \cup \{(i,j)\}$ is a redundant edge set of $G$.
\end{lemma}
\begin{proof}
$c_{ij} > d_{ij}$ implies that there exists a walk $w_{ij} : i \leadsto j$ in $G$ with weight $c_{w_{ij}} < c_{ij}$. The walk $w_{ij}$ might contain edges in $R$. However, by Definition~\ref{def:redundant_edge_set} each edge of $w_{ij}$ that is in $R$ can be replaced by another path in $(V,E \setminus R, c(E \setminus R)$ with no greater weight. Hence, there exists a walk $\hat{w}_{ij}$ in $(V,E \setminus R, c(E \setminus R)$ such that $c_{\hat{w}_{ij}} < c_{ij}$. If $(i,j)$ is part of $\hat{w}_{ij}$, then $\hat{w}_{ij}$ is of the form $i \leadsto i \rightarrow j \leadsto j$. The fact that $c_{\hat{w}_{ij}} < c_{ij}$ implies that at least one of the closed walks $i \leadsto i$ and $j \leadsto j$ must have negative weight. This contradicts the statement assumption. Hence, $(i,j)$ cannot be part of $\hat{w}_{ij}$. Therefore, it holds that
\begin{equation} \label{eqn:cij_gr_dij}
\text{$\exists \; \hat{w}_{ij} : i \leadsto j$ in $\Big(V,E \setminus (R \cup \{(i,j)\}), c\big(E \setminus (R \cup \{(i,j)\})\big)\Big)$ s.t.~$c_{\hat{w}_{ij}} \le c_{ij}$}.
\end{equation}
(\ref{eqn:cij_gr_dij}) implies that applying Lemma~\ref{thm:res_union} with $A = \{(i,j)\}$ yields the desired statement.
\end{proof}

\begin{lemma} \label{thm:EkR}
Let $G = (V,E,c(E))$ be an edge weighted directed graph without negative weight closed walks. Let $(i,j) \in E_k^r$ for some $k \in \{1,2,\ldots,K\}$, where $E_k^r$ is defined in (\ref{def:EkR}) in Definition~\ref{def:eq_Ek}, in the context of $G$. Then for any redundant edge set (of $G$) denoted $R$, the union $R \cup \{(i,j)\}$ is also a redundant edge set of $G$. Consequently, $R \cup E_1^r \cup E_2^r \cup \ldots \cup E_K^r$ is also a redundant edge set of $G$.
\end{lemma}
\begin{proof}
This is a direct consequence of Lemma~\ref{thm:cij_gr_dij} because by (\ref{def:EkR}) in Definition~\ref{def:eq_Ek} $c_{ij} > d_{ij}$ for $(i,j) \in E_k^r$, with $d_{ij}$ being the minimum walk weight from $i$ to $j$ (among all walks in $G$) defined in Definition~\ref{def:min_walk_weight}.
\end{proof}

\begin{lemma} \label{thm:Eijc_R}
Let $G = (V,E,c(E))$ be an edge weighted directed graph without negative weight closed walks. For any $k \in \{1,2,\ldots,K\}$, $q \in \{1,2,\ldots,K\}$, $k \neq q$, let $E_{kq}$ and $E_{kq}^c$ be defined in Definition~\ref{def:eq_Eij} such that $E_{kq} \neq \emptyset$. If $R$ is a redundant edge set of $G$, then for any $(i,j) \in E_{kq}^c$, the set $(R \cup E_{kq}) \setminus \{(i,j)\}$ is a redundant edge set of $G$.
\end{lemma}
\begin{proof}
Since $(R \cup E_{kq}) \setminus \{(i,j)\} = (R \setminus \{(i,j)\}) \cup (E_{kq} \setminus \{(i,j)\}) := R^\prime \cup E_{kq}^\prime$, we will show that $R^\prime \cup E_{kq}^\prime$ is a redundant edge set according to Definition~\ref{def:redundant_edge_set}. First, it is claimed that for each $(u,v) \in E'_{kq}$, there exists a replacement walk $w_{uv} : u \leadsto i \rightarrow j \leadsto v$ satisfying
\begin{subequations}
\begin{align}
&\text{$w_{uv}$ in $(V,((E_k \cup E_q) \setminus R') \cup \{(i,j)\},c(((E_k \cup E_q) \setminus R') \cup \{(i,j)\}))$}, \label{eqn:Eijc_R_wuv0} \\
&\text{the weight of $w_{uv}$, denoted $c_{w_{uv}}$, satisfies $c_{w_{uv}} \le c_{uv}$}. \label{eqn:Eijc_weight}
\end{align}
\end{subequations}
The argument for the existence of $w_{uv}$ and (\ref{eqn:Eijc_R_wuv0}) is as follows:  associated with $E_{kq}$ let $[v_k]$ and $[v_q]$ be the equivalence classes in $G$, as defined in Definition~\ref{def:eq_vk}. Since $R^\prime \subseteq R$ and $R$ is a redundant edge set of $G$, $R'$ is also a redundant edge set of $G$. Hence, $[v_k]$ and $[v_q]$ remain equivalence classes in $(V, E \setminus R', c(E \setminus R'))$. Consequently, Lemma~\ref{lem:wij} implies that there is a walk $u \leadsto i$ in $([v_k], E_k \setminus R', c(E_k \setminus R'))$. In addition, the weight of the walk is $d_{ui}$, the minimum walk weight $u \leadsto i$ in $G$ in Definition~\ref{def:min_walk_weight}. Similarly, Lemma~\ref{lem:wij} implies that there is a walk $j \leadsto v$ in $([v_q], E_q \setminus R', c(E_q \setminus R'))$ with walk weight $d_{j v}$. Combining the walks $u \leadsto i$, $j \leadsto v$ and the edge $(i,j)$, we conclude that $w_{uv}$ exists and (\ref{eqn:Eijc_R_wuv0}) is satisfied. To show (\ref{eqn:Eijc_weight}), first note that 
\begin{equation} \label{eqn:E_ij_redundant}
c_{w_{uv}} = d_{u i} + c_{ij} + d_{jv},
\end{equation}
where $c_{ij}$ is the weight of edge $(i,j)$. By Lemma~\ref{lem:dis_dsj}, it holds that $d_{ui} = d_{u v_k} + d_{v_k i}$ since $u,i,v_k \in [v_k]$. Similarly, it holds that $d_{jv} = d_{j v_q} + d_{v_q v}$. Hence, (\ref{eqn:E_ij_redundant}) can be rewritten as
\begin{displaymath}
\begin{array}{rcl}
c_{w_{uv}} & = & d_{u v_k} + \overbrace{d_{v_k i} + c_{ij} + d_{j v_q}}^{\text{$= \tilde{c}_{v_k v_q} \le d_{v_k u} + c_{uv} + d_{v v_q}$ by (\ref{def:Eijc}) in Definition~\ref{def:eq_Eij}}} + d_{v_q v} \vspace{2mm} \\
& \le & \underbrace{d_{u v_k} + d_{v_k u}}_{\text{$ = 0$ by Lemma~\ref{lem:dij_dji}}} + c_{uv} + \underbrace{d_{v v_q} + d_{v_q v}}_{\text{$ = 0$ by Lemma~\ref{lem:dij_dji}}} \vspace{2mm} \\
& = & c_{uv}.
\end{array}
\end{displaymath}
Therefore, a replacement walk $w_{uv}: u \leadsto i \rightarrow j \leadsto v$ satisfying (\ref{eqn:Eijc_R_wuv0}) and (\ref{eqn:Eijc_weight}) exists. Since it holds that
\begin{equation} \label{eqn:Eijc_R1}
(i,j) \notin R^\prime \implies R^\prime \setminus (i,j) = R^\prime, \;\; \text{and} \;\; E_{kq}^\prime \cap (E_k \cup E_q \cup \{(i,j)\}) = \emptyset,
\end{equation}
The set $((E_k \cup E_q) \setminus R^\prime) \cup \{(i,j)\}$ can be rewritten as
\begin{displaymath}
\begin{array}{cl}
& ((E_k \cup E_q) \setminus R^\prime) \cup \{(i,j)\} \vspace{1mm} \\
= & (E_k \cup E_q \cup \{(i,j)\}) \setminus (R^\prime \setminus \{(i,j)\}) \vspace{1mm} \\
\overset{(\ref{eqn:Eijc_R1})}{=} & (E_k \cup E_q \cup \{(i,j)\}) \setminus (R^\prime \cup E_{kq}^\prime) \vspace{1mm} \\
\subseteq & E \setminus (R^\prime \cup E_{kq}^\prime).
\end{array}
\end{displaymath}
This implies that (\ref{eqn:Eijc_R_wuv0}) can be modified to state that 
\begin{equation} \label{eqn:Eijc_R_wuv}
\text{$w_{uv}$ is in $(V, E \setminus (R^\prime \cup E_{kq}^\prime), c(E \setminus (R^\prime \cup E_{kq}^\prime)))$}.
\end{equation}
Finally, (\ref{eqn:Eijc_R_wuv}) and (\ref{eqn:Eijc_weight}) imply that (\ref{eqn:resu_rp_ij}) in Lemma~\ref{thm:res_union} holds with $R = R'$ and $A = E'_{kq}$. Hence, Lemma~\ref{thm:res_union} guarantees that $R' \cup E'_{kq}$ is a redundant edge set of $G$.

\end{proof}

The implication of Lemma \ref{thm:EkR} and Lemma \ref{thm:Eijc_R} is as follows: corresponding to the edge set decomposition in (\ref{eqn:E_decomp}), any maximum redundant edge set must be a member of
\begin{equation} \label{eqn:mres_decomp_general}
\mathcal{R} = \left\{R_0 \cup \left(\mathop{\cup}\limits_{k = 1}^K E_k^r \cup R_k\right) \right\}
= \left\{\left(\mathop{\cup}\limits_{1 \le i \neq j \le K} R_{ij}\right) \cup \left(\mathop{\cup}\limits_{k = 1}^K E_k^r \cup R_k\right) \right\}
\end{equation}
where
\begin{equation} \label{eqn:mres_dg_properties}
\begin{aligned}
& \text{For $k \ge 1$, \; $E_k^r$ is defined by (\ref{def:EkR}) in Definition~\ref{def:eq_Ek}},\\
& R_0 \subseteq E_0, \quad R_0 = \mathop{\cup}\limits_{1 \le i \neq j \le K} \; R_{ij}, \\
& R_{ij} \subseteq E_{ij}, \quad |E_{ij} \setminus R_{ij}| \le 1, \quad (E_{ij} \setminus R_{ij}) \subseteq E_{ij}^c, \\
& R_k \subseteq E_k^c, \quad k \in \{1,2,\ldots,K\}.
\end{aligned}
\end{equation}
In (\ref{eqn:mres_decomp_general}) the inclusion of $E_k^r$ for $1 \le k \le K$ is due to Lemma~\ref{thm:EkR}. For now, $R_{ij}$ and $R_k$ are not fully specified, and the partial characterization of $R_{ij}$ in (\ref{eqn:mres_dg_properties}) is due to Lemma~\ref{thm:Eijc_R}. Furthermore, Lemma~\ref{thm:Eijc_R} suggests that, instead of searching over $\mathcal{R}$ for a maximum redundant edge set, it is without loss of generality to search over the following restricted set
\begin{equation} \label{eqn:mres_decomp}
\begin{array}{ccl}
\mathcal{R}^d & := & \left\{\Big(\mathop{\cup}\limits_{1 \le i \neq j \le K} \big((E_{ij} \setminus \{(v_{ij}^i, v_{ij}^j)\}) \cup R_{ij}^d\big)\Big) \cup (\mathop{\cup}\limits_{k = 1}^K E_k^r \cup R_k)\right\} \vspace{2mm} \\
& = & \left\{\big((E_0 \setminus E_0^d) \cup R_0^d \big)\cup (\mathop{\cup}\limits_{k = 1}^K E_k^r \cup R_k) \right\},
\end{array}
\end{equation}
where
\begin{displaymath}
R_{ij}^d \subseteq \{(v_{ij}^i, v_{ij}^j)\}, \quad R_0^d = \mathop{\cup}\limits_{1 \le i \neq j \le K} R_{ij}^d, \quad R_0^d \subseteq E_0^d, \quad R_k \subseteq E_k^c,
\end{displaymath}
and we note that $E_0^d$ (i.e., the collection of all edges $(v_{ij}^i, v_{ij}^j)$) is defined in (\ref{eqn:E0}). The restriction from (\ref{eqn:mres_decomp_general}) to (\ref{eqn:mres_decomp}) amounts to the following specializations
\begin{displaymath}
\begin{array}{ccl}
\text{$R_{ij}$ in (\ref{eqn:mres_decomp_general})} & \rightarrow & \text{$(E_{ij} \setminus \{(v_{ij}^i, v_{ij}^j)\}) \cup R_{ij}^d$ in (\ref{eqn:mres_decomp})}, \vspace{1mm} \\
\text{$R_0$ in (\ref{eqn:mres_decomp_general})} & \rightarrow & \text{$(E_0 \setminus E_0^d) \cup R_0^d$ in (\ref{eqn:mres_decomp})}.
\end{array}
\end{displaymath}
The restriction is justified as follows: suppose $R^\star \in \mathcal{R}$ in (\ref{eqn:mres_decomp_general}),
\begin{displaymath}
R^\star = \left\{\left(\mathop{\cup}\limits_{1 \le i \neq j \le K} R_{ij}^\star \right) \cup \left(\mathop{\cup}\limits_{k = 1}^K E_k^r \cup R_k^\star \right) \right\}
\end{displaymath}
is a maximum redundant edge set, with appropriate choices (to be discussed in the sequel) of $R_{ij}^\star$ and $R_k^\star$ satisfying (\ref{eqn:mres_dg_properties}). Define $R_{ij}^{d \star}$ by
\begin{displaymath}
R_{ij}^{d \star} := \begin{cases} \{(v_{ij}^i, v_{ij}^j)\}, & \text{if $R_{ij}^\star = E_{ij}^c$} \\ \emptyset, & \text{if $R_{ij}^\star \neq E_{ij}^c$} \end{cases},
\end{displaymath}
and define $R^{d \star}$ by
\begin{displaymath}
R^{d \star} = \left\{\Big(\mathop{\cup}\limits_{1 \le i \neq j \le K} \big((E_{ij} \setminus \{(v_{ij}^i, v_{ij}^j)\}) \cup R_{ij}^{d \star}\big)\Big) \cup (\mathop{\cup}\limits_{k = 1}^K E_k^r \cup R_k^\star)\right\}.
\end{displaymath}
By construction, $R^{d \star} \in \mathcal{R}^d$ in (\ref{eqn:mres_decomp}) and $|R^{d \star}| = |R^\star|$. In addition, since $R^\star$ is a maximum redundant edge set, Lemma~\ref{thm:Eijc_R} states that $R^{d \star}$ is also a redundant edge set (and $|R^{d \star}| = |R^\star|$). Hence, $R^{d \star}$ is a maximum redundant edge set. Conversely, define the set-valued function $F : \mathcal{R}^d \mapsto 2^\mathcal{R}$,
\begin{equation} \label{eqn:Er_star_exp}
F(X) = \left\{ R \subseteq \mathcal{R} \; \vline \; 
\begin{array}{l}
R = (\mathop{\cup}\limits_{1 \le i \neq j \le K} R_{ij}) \cup (\mathop{\cup}\limits_{k = 1}^K (E_k^r \cup R_k)), \vspace{1mm} \\
X =  (\mathop{\cup}\limits_{1 \le i \neq j \le K} ((E_{ij} \setminus \{(v_{ij}^i, v_{ij}^j)\}) \cup R_{ij}^{d})) \vspace{1mm} \\
\quad \quad \cup (\mathop{\cup}\limits_{k = 1}^K (E_k^r \cup R_k)), \vspace{1.5mm} \\
R_{ij} := \begin{cases} E_{ij}, & \text{if $R_{ij}^d \neq \emptyset$} \\ E_{ij} \setminus \{(g,h)\}, \;\; (g,h) \in E_{ij}^c, & \text{if $R_{ij}^d = \emptyset$} \end{cases}
\end{array} \right\}.
\end{equation}
Then, if $R^{d \star} \in \mathcal{R}^d$ in (\ref{eqn:mres_decomp}) is a maximum redundant edge set, by Lemma~\ref{thm:Eijc_R} the set $F(R^{d \star}) \subseteq \mathcal{R}$ is the set of all maximum redundant edge sets of the form in (\ref{eqn:mres_decomp_general}) that share the same $R_k^\star$'s as in $R^{d \star}$.

Next, we focus on the specialized maximum redundant edge sets in $\mathcal{R}^d$ in (\ref{eqn:mres_decomp}). Let $R^d$ denote any member of $\mathcal{R}^d$. The second expression in (\ref{eqn:mres_decomp}) reveals that the components of $R^d$ which are not fully specified are $R_0^d \subseteq E_0^d$ and $R_k \subseteq E_k^c$ for $k = 1,2,\ldots,K$. First, we examine the conditions on these components under which $R^d$ is a redundant edge set (of graph $G = (V,E,c(E))$) according to Definition~\ref{def:redundant_edge_set}. Since the fixed components in (\ref{eqn:mres_decomp}) (i.e., the edges guaranteed to be included in $R^d$) can be written as
\begin{equation} \label{eqn:Rd_eps}
\begin{aligned}
(\mathop{\cup}\limits_{k = 1}^K E_k^r) \cup (E_0 \setminus E_0^d) = \; & (\mathop{\cup}\limits_{k = 1}^K (E_k \setminus E_k^c)) \cup (E_0 \setminus E_0^d) \\
= \; & E \setminus (E_0^d \cup (\mathop{\cup}\limits_{k = 1}^K E_k^c)) \\
:= \; & E \setminus \mathcal{E},
\end{aligned}
\end{equation}
two necessary conditions for $R^d$ to be a redundant edge set are
\begin{equation} \label{eqn:R0d_red_cond}
\begin{array}{l}
\forall (u,v) \in R_0^d, \; \exists \; p_{uv} : u \leadsto v \; \text{in} \vspace{1mm} \\
 \left(V, \mathcal{E} \setminus \big(R_0^d \cup (\mathop{\cup}\limits_{k = 1}^K R_k)\big), c\Big(\mathcal{E} \setminus \big(R_0^d \cup (\mathop{\cup}\limits_{k = 1}^K R_k)\big)\Big)\right), \; \text{with} \; c_{p_{uv}} \le c_{uv},
\end{array}
\end{equation}
and
\begin{equation} \label{eqn:Rk_red_cond}
\begin{array}{l}
\forall \; k \in \{1,2,\ldots,K\}, \; \; \forall (u,v) \in R_k, \; \exists \; p_{uv} : u \leadsto v \; \text{in} \vspace{1mm} \\
  \left(V, \mathcal{E} \setminus \big(R_0^d \cup (\mathop{\cup}\limits_{k = 1}^K R_k)\big), c\Big(\mathcal{E} \setminus \big(R_0^d \cup (\mathop{\cup}\limits_{k = 1}^K R_k)\big)\Big)\right), \; \text{with} \; c_{p_{uv}} \le c_{uv}.
\end{array}
\end{equation}
Conversely, if (\ref{eqn:R0d_red_cond}) and (\ref{eqn:Rk_red_cond}) are satisfied then Lemma~\ref{thm:res_union} can be applied to show that $R^d$ is indeed a redundant edge set of $G$. Lemma~\ref{thm:res_union} is applied in the following settings:
\begin{displaymath}
R \leftarrow (\mathop{\cup}_{k = 1}^K E_k^r) \cup (E_0 \setminus E_0^d) = E \setminus \mathcal{E}, \;\;
A \leftarrow R_0^d \cup (\mathop{\cup}_{k = 1}^K R_k), \;\;
(\ref{eqn:resu_rp_ij}) \leftarrow (\ref{eqn:R0d_red_cond}), (\ref{eqn:Rk_red_cond}),
\end{displaymath}
and the fact that $(\mathop{\cup}_{k = 1}^K E_k^r) \cup (E_0 \setminus E_0^d)$ is a redundant edge set of $G$ is due to Lemma~\ref{thm:EkR} and Lemma~\ref{thm:Eijc_R}. Therefore, $R^d$ is a redundant edge set of $G$ if and only if (\ref{eqn:R0d_red_cond}) and (\ref{eqn:Rk_red_cond}) are satisfied. The following two statements, Lemma~\ref{thm:Ek} and Lemma~\ref{thm:E0}, specify that the conditions in (\ref{eqn:R0d_red_cond}) and (\ref{eqn:Rk_red_cond}) are in fact equivalent to $K+1$ decoupled conditions, one for each set of $R_0^d, R_1, R_2, \ldots, R_K$. These two statements are first described. Then, their consequences are discussed.

\begin{lemma} \label{thm:Ek}
Let $G = (V,E,c(E))$ be an edge weighted directed graph without negative weight closed walks, and let $(i,j) \in E_k^c$ for some $k \in \{1,2,\ldots,K\}$ (see (\ref{def:Ekc}) in Definition~\ref{def:eq_Ek} for $E_k^c$). Let the minimum walk weight $d_{ij}$ (of $G$) be defined in Definition~\ref{def:min_walk_weight}. If $w_{ij} : i \leadsto j$ is a walk in $G$ such that $c_{w_{ij}} = d_{ij}$ (note that $c_{w_{ij}} \ge d_{ij}$ must hold), then all nodes traversed by $w_{ij}$ are in the equivalence class $[v_k]$ (corresponding to $E_k^c$).
\end{lemma}
\begin{proof}
By definition $(i,j) \in E_k^c$ means that $i,j \in [v_k]$. Hence, Lemma~\ref{lem:wij} and \ref{lem:dij_dji} states that there exists a walk $w_{ji}: j \leadsto i$ (in $G$) such that $c_{w_{ji}} = -d_{ij} = -c_{ij}$ (the last equality is due to the fact that $(i,j) \in E_k^c$). Let $w_{ij}$ be the walk described in the statement (with weight $c_{w_{ij}} = d_{ij} = c_{ij}$). If $w_{ij}$ traverses a node $t \notin [v_k]$, then by concatenating $w_{ij}$ and $w_{ji}$ we obtain a closed walk $w_{tt}: i \leadsto t \leadsto j \leadsto i$. The weight of $w_{tt}$ is $c_{w_{tt}} = c_{w_{ij}} + c_{w_{ji}} = c_{ij} - c_{ij} = 0$. Therefore, the assumption that $t \notin [v_k]$ leads to the contradictory conclusion that $t \sim i$ and hence $t \in [v_k]$. Consequently, the walk $w_{ij}$ cannot traverse any node $t \notin [v_k]$.
\end{proof}

\begin{lemma} \label{thm:E0}
Let $G = (V,E,c(E))$ be an edge weighted directed graph without negative weight closed walks. In addition, let the following be assumed in the context of $G$:
\begin{enumerate}
\item $[v_1], [v_2], \ldots, [v_K]$ denote the equivalence classes in $V$ induced by relation $\sim$ (see Definition~\ref{def:eq_vk}).
\item $E_0^d$, $E_k^c$ for $k = 1,2,\ldots,K$ are defined in Definition~\ref{def:eq_E0} and \ref{def:eq_Ek}, respectively. Let $\mathcal{E} := E_0^d \cup (\cup_{q = 1}^K E_q^c)$.
\item For $k \in \{1,\ldots,K\}$, let $R_k \subseteq E_k^c$ be given and assume that $[v_1], \ldots, [v_K]$ remain equivalence classes in $(V, \mathcal{E} \setminus (\mathop{\cup}_{q = 1}^K R_q),c(\mathcal{E} \setminus (\mathop{\cup}_{q = 1}^K R_q)))$.
\item $\tilde{G} = (\tilde{V}, \tilde{E}_0, \tilde{c}(\tilde{E}_0))$ is the condensation of $G$, in accordance with the designation of representing nodes in $[v_1], \ldots, [v_K]$. $\tilde{E}_0$ is the set of all edges in $\tilde{G}$ (see Definition~\ref{def:condensation}).
\item Let $R_0^d \subseteq E_0^d$ be given, and let $\tilde{R}_0 \subseteq \tilde{E}_0$ be defined to correspond to $R_0^d$ in the sense that $(v_i, v_j) \in \tilde{R}_0$ if and only if $(v_{ij}^i, v_{ij}^j) \in R_0^d$.
\end{enumerate}

Then, the following two statements are equivalent:
\begin{description}
\item[(a)] For each $(v_{ij}^i,v_{ij}^j) \in R_0^d$, there is a (replacement) walk $w_{v_{ij}^i v_{ij}^j}$ in graph $\big(V, \mathcal{E} \setminus (R_0^d \cup (\cup_{k = 1}^K R_k)), c(\mathcal{E} \setminus (R_0^d \cup (\cup_{k = 1}^K R_k)))\big)$ satisfying $c_{v_{ij}^i v_{ij}^j} \ge c_{w_{v_{ij}^i v_{ij}^j}}$.
\item[(b)] For each $(v_i,v_j) \in \tilde{R}_0$, there is a (replacement) walk $\tilde{w}_{v_i v_j}$ in graph $\big(\tilde{V}, \tilde{E}_0 \setminus \tilde{R}_0, \tilde{c}(\tilde{E}_0 \setminus \tilde{R}_0)\big)$ satisfying $\tilde{c}_{v_i v_j} \ge \tilde{c}_{\tilde{w}_{v_i v_j}}$.
\end{description}
\end{lemma}
\begin{proof}
For convenience, we denote 
\begin{displaymath}
\begin{aligned}
\mathcal{G}^c &:= (V, \mathcal{E} \setminus (R_0^d \cup (\cup_{k = 1}^K R_k)), c(\mathcal{E} \setminus (R_0^d \cup (\cup_{k = 1}^K R_k)))), \\
\tilde{G}^c &:= (\tilde{V}, \tilde{E}_0 \setminus \tilde{E}_r^0, \tilde{c}(\tilde{E}_0 \setminus \tilde{E}_r^0)).
\end{aligned}
\end{displaymath}
Because of the definition of $E_0^d$, in $\mathcal{G}^c$ every walk $w_{v_{ij}^i v_{ij}^j}$ from $v_{ij}^i$ to $v_{ij}^j$ is of the form
\begin{equation} \label{eqn:E_0_wij}
\underbrace{v_{ij}^i \leadsto v_{k(1) k(2)}^{k(1)}}_{\text{in $\mathcal{G}_{k(1)}^c$}} \rightarrow \underbrace{v_{k(1) k(2)}^{k(2)} \leadsto v_{k(2) k(3)}^{k(2)}}_{\text{in $\mathcal{G}_{k(2)}^c$}} \rightarrow \dots \rightarrow \underbrace{v_{k(m-1) k(m)}^{k(m)} \leadsto v_{ij}^j}_{\text{in $\mathcal{G}_{k(m)}^c$}},
\end{equation}
where $k(1) = i$, $k(m) = j$, $k(2),\ldots,k(m-1) \in \{1,2,\ldots,K\}$ are indices of the intermediate equivalence classes in the order traversed by $w_{v_{ij}^i v_{ij}^j}$. In addition, $\mathcal{G}_{k(q)}^c = ([v_{k(q)}], E_{k(q)}^c \setminus R_{k(q)}, c(E_{k(q)}^c \setminus R_{k(q)}))$ for $q = 1,2,\ldots,m$. Due to statement assumption 3, $[v_{k(q)}]$'s remain equivalence classes in $\mathcal{G}^c$. Thus,  by Lemma~\ref{lem:wij} for any two nodes $s$ and $t$ in $[v_{k(q)}]$ there exists at least one walk $w_{st}$ from $s$ to $t$ in $\mathcal{G}_{k(q)}^c$. Since $R_k \subseteq E_k^c$ and $E_{k(q)}^c \cap E_k^c = \emptyset$ as long as $k(q) \neq k$, it holds that $R_k \cap E_{k(q)}^c = \emptyset$ as long as $k(q) \neq k$. In addition, since $R_0^d \subseteq E_0^d$ and $E_0^d \cap E_k^c = \emptyset$ for $k \ge 1$, $R_0^d \cap E_k^c = \emptyset$ for $k \ge 1$. Therefore, it holds that
\begin{displaymath}
\underbrace{(E_{k(q)}^c \setminus R_{k(q)})}_{\text{edge set of $\mathcal{G}_{k(q)}^c$}} = (E_{k(q)}^c \setminus (R_0^d \cup (\cup_{k = 1}^K R_k))) \subseteq \underbrace{(\mathcal{E} \setminus (R_0^d \cup (\cup_{k = 1}^K R^k)))}_{\text{edge set of $\mathcal{G}^c$}}.
\end{displaymath} 
Thus, since $w_{st}$ is in $\mathcal{G}_{k(q)}^c$ it is also in $\mathcal{G}^c$. Therefore, a walk of the form in (\ref{eqn:E_0_wij}) exists in $\mathcal{G}^c$ if and only if all edges $(v_{k(q)k(q+1)}^{k(q)}, v_{k(q)k(q+1)}^{k(q+1)})$ exist in $E_0^d \setminus R_0^d$ for $q = 1,2,\ldots,m-1$, since these edges can only be contained in $E_0^d$ or $R_0^d$. By Definition~\ref{def:condensation} and statement assumption 5, these edges exist if and only if the edges $(v_{k(q)}, v_{k(q+1)})$ exist in $\tilde{E}_0 \setminus \tilde{R}_0$ for all $q$. Further, if a walk $w_{v_{ij}^i v_{ij}^j}$ of the form (\ref{eqn:E_0_wij}) exists in $\mathcal{G}^c$ then the corresponding walk $\tilde{w}_{v_i v_j}$ in $\tilde{G}^c$ is of the form
\begin{equation} \label{eqn:E_0_wij_tilde}
(v_i = ) v_{k(1)} \rightarrow v_{k(2)} \rightarrow \ldots \rightarrow v_{k(m)} ( = v_j).
\end{equation}
Conversely, if the walk $\tilde{w}_{v_i v_j}$ in (\ref{eqn:E_0_wij_tilde}) exists in $\tilde{G}^c$, then in $\mathcal{G}^c$ at least one walk $w_{v_{ij}^i v_{ij}^j}$ of the form (\ref{eqn:E_0_wij}) exists (the possible multiplicity of the walks is due to the possibilities of multiple walks within $\mathcal{G}_{k(q)}^c$).

Next, we establish the equivalence between the walk weight inequalities (i.e., $c_{v_{ij}^i v_{ij}^j} \ge c_{w_{v_{ij}^i v_{ij}^j}}$ and $\tilde{c}_{v_i v_j} \ge \tilde{c}_{\tilde{w}_{v_i v_j}}$). We consider only the cases when $w_{v_{ij}^i v_{ij}^j}$ is restricted to the choices where the walks in $\mathcal{G}_{k(q)}^c$ have minimum weights. By Lemma~\ref{lem:wij}, these minimum weights are $d_{v_{ij}^i \; v_{k(1)k(2)}^{k(1)}}$, $d_{v_{k(m-1)k(m)}^{k(m)} \; v_{ij}^j}$, and $d_{v_{k(q-1)k(q)}^{k(q)} \; v_{k(q)k(q+1)}^{k(q)}}$ for $q = 2,3,\ldots,m-1$ respectively. Therefore, $c_{v_{ij}^i v_{ij}^j} \ge c_{w_{v_{ij}^i v_{ij}^j}}$ if and only if
\begin{equation} \label{eqn:E_0_ineq}
\begin{array}{rcl}
c_{v_{ij}^i v_{ij}^j} & \ge & d_{v_{ij}^i \; v_{k(1)k(2)}^{k(1)}} + d_{v_{k(m-1)k(m)}^{k(m)} \; v_{ij}^j} \vspace{2mm} \\
& & + \sum\limits_{q = 2}^{m-1} d_{v_{k(q-1)k(q)}^{k(q)} \; v_{k(q)k(q+1)}^{k(q)}} + \sum\limits_{q = 1}^{m-1} c_{v_{k(q)k(q+1)}^{k(q)} \; v_{k(q)k(q+1)}^{k(q+1)}}.
\end{array}
\end{equation}
By Lemma~\ref{lem:dij_dji} and \ref{lem:dis_dsj}, in (\ref{eqn:E_0_ineq}) it holds that
\begin{equation} \label{eqn:E_0_ineq1}
\begin{array}{l}
d_{v_{ij}^i v_{k(1)k(2)}^{k(1)}} = d_{v_{ij}^i v_{k(1)}} + d_{v_{k(1)} v_{k(1)k(2)}^{k(1)}} = - d_{v_{k(1)} v_{ij}^i} + d_{v_{k(1)} v_{k(1)k(2)}^{k(1)}}, \vspace{2mm} \\
d_{v_{k(m-1)k(m)}^{k(m)} v_{ij}^j} = d_{v_{k(m-1)k(m)}^{k(m)} v_{k(m)}} + d_{v_{k(m)} v_{ij}^j} = d_{v_{k(m-1)k(m)}^{k(m)} v_{k(m)}} - d_{v_{v_{ij}^j k(m)}}, \vspace{2mm} \\
d_{v_{k(q-1)k(q)}^{k(q)} \; v_{k(q)k(q+1)}^{k(q)}} = d_{v_{k(q-1)k(q)}^{k(q)} v_{k(q)}} + d_{v_{k(q)} v_{k(q)k(q+1)}^{k(q)}}, \; \forall q = 2,\ldots,m-1.
\end{array}
\end{equation}
Therefore, with (\ref{eqn:E_0_ineq1}) the inequality in (\ref{eqn:E_0_ineq}) can be rearranged into
\begin{equation} \label{eqn:E_0_ineq2}
\begin{array}{rcl}
d_{v_{k(1)} v_{ij}^i} + c_{v_{ij}^i v_{ij}^j} + d_{v_{ij}^j v_{k(m)}} & \ge & \sum\limits_{q = 1}^{m-1} \Big(d_{v_{k(q)} v_{k(q)k(q+1)}^{k(q)}} + \vspace{2mm} \\
& & c_{v_{k(q)k(q+1)}^{k(q)} v_{k(q)k(q+1)}^{k(q+1)}} + d_{v_{k(q)k(q+1)}^{k(q+1)} v_{k(q+1)}} \Big).
\end{array}
\end{equation}
Since $i = k(1)$ and $j = k(m)$, by (\ref{eqn:viji_vijj}) in Definition~\ref{def:condensation}, (\ref{eqn:E_0_ineq2}) can be rewritten as
\begin{equation} \label{eqn:E_0_ineq3}
\tilde{c}_{v_i v_j} = d_{v_{k(1)} v_{ij}^i} + c_{v_{ij}^i v_{ij}^j} + d_{v_{ij}^j v_{k(m)}} \ge \sum\limits_{q = 1}^{m-1} \tilde{c}_{v_{k(q)} v_{k(q+1)}} = \tilde{c}_{\tilde{w}_{v_i v_j}},
\end{equation}
where the last equality in (\ref{eqn:E_0_ineq3}) is due to (\ref{eqn:E_0_wij_tilde}). Therefore,
\begin{displaymath}
c_{v_{ij}^i v_{ij}^j} \ge c_{w_{v_{ij}^i v_{ij}^j}} \iff \tilde{c}_{v_i v_j} \ge \tilde{c}_{\tilde{w}_{v_i v_j}}.
\end{displaymath}

Now we establish the equivalence between (a) and (b) in the statement. If $R_0^d$ satisfies (a), then for each $(v_{ij}^i, v_{ij}^j) \in R_0^d$ there exists a walk $w_{v_{ij}^i v_{ij}^j}$ of the form in (\ref{eqn:E_0_wij}) satisfying $c_{v_{ij}^i v_{ij}^j} \ge c_{w_{v_{ij}^i v_{ij}^j}}$. Additionally, we can assume in $w_{v_{ij}^i v_{ij}^j}$ all walks in $\mathcal{G}_{k(q)}^c$ for all $q$ have the least possible weights. Therefore, by the previous parts of the proof, the corresponding edge $(v_i, v_j)$ and walk $\tilde{w}_{v_i v_j}$ exists in $\tilde{G}^c$, and they satisfy the inequality $\tilde{c}_{v_i v_j} \ge \tilde{c}_{\tilde{w}_{v_i v_j}}$. Thus, $\tilde{R}_0 \subseteq \tilde{E}_0$, as defined in the statement, satisfies (b). This shows that (a) implies (b). The argument for (b) implying (a) can be shown in a similar fashion.
\end{proof}

Now we analyze the consequences of Lemma~\ref{thm:Ek} and Lemma~\ref{thm:E0}. Due to Lemma~\ref{thm:Ek}, in condition (\ref{eqn:Rk_red_cond}) $\mathcal{E}$ can be replaced with $E_k^c$. That is, (\ref{eqn:Rk_red_cond}) becomes
\begin{equation} \label{eqn:Erk_rpp1}
\begin{aligned}
&\text{$\forall k \in \{1,2,\ldots,K\}$, \; $\forall (u,v) \in R_k$, \; $\exists \; p_{uv}: u \leadsto v$} \\
&\text{in $([v_k], E_k^c \setminus (R_0^d \cup (\mathop{\cup}\limits_{k = 1}^K R_k)), c( E_k^c \setminus (R_0^d \cup (\mathop{\cup}\limits_{k = 1}^K R_k))))$, $c_{p_{uv}} \le c_{uv}$}.
\end{aligned}
\end{equation}
In addition, since $E_k^c \cap R_q = \emptyset$ as long as $q \neq k$ and $E_k^c \cap R_0^d = \emptyset$, it holds that $E_k^c \setminus (R_0^d \cup (\cup_{k = 1}^K R_q)) = E_k^c \setminus R_k$. Therefore, (\ref{eqn:Erk_rpp1}) can be further simplified to establish the following observations:
\begin{subequations} \label{eqn:Ekr_redundant}
\begin{align}
&\text{$R_k$ is part of redundant edge set $R^d$, for $1 \le k \le K$} \\
\iff & \text{$\forall (u,v) \in R_k$, $\exists \; p_{uv}$ in $([v_k], E_k^c \setminus R_k, c(E_k^c \setminus R_k))$, $c_{p_{uv}} \le c_{uv}$} \label{eqn:Ekr_redundant_b} \\
\iff &\text{$R_k$ is a redundant edge set in $([v_k], E_k^c, c(E_k^c))$}.
\end{align}
\end{subequations}
One of the consequences of (\ref{eqn:Ekr_redundant}) is that, for $k = 1,2,\ldots,K$, whether or not $R_k$ is part of a redundant edge set $R^d$ does not depend on the choices of $R_q$ for $q \in \{1,\ldots,K\} \setminus \{k\}$ or the choice of $R_0^d$. In addition, (\ref{eqn:Ekr_redundant}) can be used to establish a similar independence result for $R_0^d$. To begin, notice that by Lemma~\ref{thm:EkR} and Lemma~\ref{thm:Eijc_R}, $(\mathop{\cup}\limits_{k = 1}^K E_k^r) \cup (E_0 \setminus E_0^d)$ is a redundant edge set in $G$. In addition, by (\ref{eqn:Rd_eps}) $\mathcal{E} := E_0^d \cup (\mathop{\cup}\limits_{k = 1}^K E_k^c) = E \setminus \big((\mathop{\cup}\limits_{k = 1}^K E_k^r) \cup (E_0 \setminus E_0^d)\big)$. Hence, Lemma~\ref{thm:eq_class_preservation} states that
\begin{equation} \label{eqn:Ekr_eq_class}
\text{$[v_1],[v_2],\ldots,[v_K]$ are equivalence classes in the graph $(V, \mathcal{E}, c(\mathcal{E}))$}.
\end{equation}
Secondly, in order for $R^d$ to be a redundant edge set (in $G$), (\ref{eqn:Ekr_redundant}) must hold. Then, by (\ref{eqn:Ekr_redundant_b}), $\cup_{k = 1}^K R_k$ is a redundant edge set in $(V, \mathcal{E}, c(\mathcal{E}))$. Consequently, when applied to $(V, \mathcal{E}, c(\mathcal{E}))$, which is a subgraph of $G$ without negative weight closed walks, Lemma~\ref{thm:eq_class_preservation} implies that
\begin{equation} \label{eqn:Ekr_eq_class1}
\text{$[v_1],\ldots,[v_K]$ are equivalence classes in $\big(V, \mathcal{E} \setminus (\mathop{\cup}\limits_{k = 1}^K R_k), c(\mathcal{E} \setminus (\mathop{\cup}\limits_{k = 1}^K R_k)))\big)$}.
\end{equation}
(\ref{eqn:Ekr_eq_class1}) implies that statement assumption 3 for Lemma~\ref{thm:E0} is satisfied with the $\mathcal{E}$ and $R_k$ for $k = 1,2,\ldots,K$. Therefore, with $\tilde{R}_0$ defined in statement assumption 5 in Lemma~\ref{thm:E0}, the lemma specifies that
\begin{align} \label{eqn:E0r_redundant}
& \text{$R_0^d$ is part of a redundant edge set $R^d$} \nonumber \\
\iff & \text{condition (\ref{eqn:R0d_red_cond})} \nonumber \\
\iff & \text{$\tilde{R}_0$ is a redundant edge set in $\tilde{G}$, the condensation of $G$}.
\end{align}
According to Definition~\ref{def:condensation}, $\tilde{G}$ is independent of $R_q$ for $q \in \{1,2,\ldots,K\}$. Hence, (\ref{eqn:E0r_redundant}) establishes the desired property that whether or not $R_0^d$ is part of a redundant edge set is independent of the choices of $R_k$ for $k \in \{1,2,\ldots,K\}$. In conclusion, for $R^d = \big((E_0 \setminus E_0^d) \cup R_0^d \big)\cup (\mathop{\cup}\limits_{k = 1}^K E_k^r \cup R_k)$ (i.e., (\ref{eqn:mres_decomp})) to be a redundant edge set of $G$, it is necessary and sufficient for $R_0^d$ to satisfy (\ref{eqn:E0r_redundant}) and each of $R_k$ (for $k = 1,2,\ldots,K$) to satisfy its individual version of (\ref{eqn:Ekr_redundant}).

The decoupling of the redundant edge set membership requirements in (\ref{eqn:Ekr_redundant}) and (\ref{eqn:E0r_redundant}) suggests that the maximum redundant edge set problem can be decoupled into $K+1$ independent maximum redundant edge set subproblems on the graphs $([v_k], E_k^c, c(E_k^c))$ in (\ref{eqn:Ekr_redundant}) and on the condensation $\tilde{G}$ in (\ref{eqn:E0r_redundant}), respectively. The following statement, whose proof has already been discussed, summarizes the main decomposition results which have been discussed so far:

\begin{theorem} \label{thm:decomposition}
Let $G = (V,E,c(E))$ be an edge weighted directed graph without negative weight closed walks. Let the following be defined in the context of $G$:
\begin{itemize}
\item $K$ is the number of equivalence classes induced by relation $\sim$ in Definition~\ref{def:eq_K}.
\item For $k \in \{1,2,\ldots,K\}$, $[v_k]$ denotes the equivalence class defined in Definition~\ref{def:eq_vk}.
\item For $k \in \{1,2,\ldots,K\}$, $E_k^r$ and $E_k^c$ are defined in Definition~\ref{def:eq_Ek}.
\item For $1 \le i \neq j \le K$, $E_{ij}$ and $E_{ij}^c$ are defined in Definition~\ref{def:eq_Eij}.
\item $\tilde{G}$ is the condensation of $G$ defined in Definition~\ref{def:condensation}.
\end{itemize}
Then, every maximum redundant edge set of $G$ can be parameterized by
\begin{equation} \label{eqn:mres_parameterization}
R_0^\star \cup \left(\mathop{\cup}\limits_{k = 1}^K (E_k^r \cup R_k^\star) \right),
\end{equation}
where for $k \in \{1,2,\ldots,K\}$, $R_k^\star$ is a maximum redundant edge set of the subgraph $([v_k], E_k^c, c(E_k^c))$. In addition, $R_0^\star$ is parameterized by
\begin{equation} \label{eqn:R0_star}
R_0^\star = \mathop{\cup}\limits_{1 \le i \neq j \le K} R_{ij}^\star,
\end{equation}
where
\begin{equation} \label{eqn:Rij_star}
R_{ij}^\star = \begin{cases} E_{ij}, & \text{if} \; (v_i, v_j) \in \tilde{R}_0^\star \\ E_{ij} \setminus \{(g,h)\} \quad \text{for some $(g,h) \in E_{ij}^c$}, & \text{if} \; (v_i, v_j) \notin \tilde{R}_0^\star \end{cases},
\end{equation}
and $\tilde{R}_0^\star$ is the maximum redundant edge set of $\tilde{G}$, the condensation of $G$.
\end{theorem}

\begin{remark}
The expression in (\ref{eqn:Rij_star}) is jointly due to Lemma~\ref{thm:E0}, (\ref{eqn:mres_decomp}) and (\ref{eqn:Er_star_exp}).
\end{remark}

\begin{remark} \label{rmk:decomposition}
While it appears that the statement of Theorem~\ref{thm:decomposition} (e.g., $E_{ij}^c$) depends on the choices of the representing nodes $v_1, \ldots, v_K$ in the equivalence classes, Theorem~\ref{thm:decomposition} in fact holds irrespective of these choices. In particular, it can be shown (in Appendix~\ref{app:decompose_well_defined}) that
\begin{itemize}
\item The definition of $E_{ij}^c$ is independent of the choices of $v_1, \ldots, v_K$.
\item With different choices of $v_1, \ldots, v_K$, it is possible to define different condensations of $G$ with different representing nodes and different edge weights $\tilde{c}$. However, for all $1 \le i \neq j \le K$, $\big| \tilde{R}_0^\star \cap ([v_i] \times [v_j]) \big|$ (which can only be 0 or 1) is independent of the choices of $v_1, \ldots, v_K$.
\end{itemize}
\end{remark}
As a result of Theorem~\ref{thm:decomposition}, the graph in Figure~\ref{fig:eq_classes} with its maximum redundant edge set removed is illustrated in Figure~\ref{fig:MRES_example}.
\begin{figure}[!tbh]
\begin{center}
 \includegraphics[width=60mm]{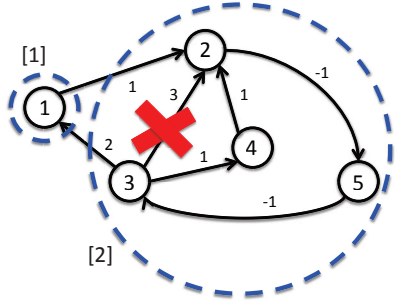}
 \caption{As a result of Theorem~\ref{thm:decomposition}, the maximum redundant edge set of the graph in Figure~\ref{fig:eq_classes} is $R^\star = \{(3,2)\}$ (in general the maximum redundant edge set need not be unique). In the parameterization in (\ref{eqn:mres_parameterization}), $K = 2$, $E_1^r = \emptyset$ and $E_r^2 = \{(3,2)\}$. $R_1^\star = \emptyset$, and by inspection $R_2^\star = \emptyset$ since the subgraph $([2],E_2^c = \{(2,5), (5,3), (3,4), (4,2)\}, c(E_2^c))$ is a zero-weight cycle. From Figure~\ref{fig:condensation_example}, the maximum redundant edge set of the condensation $\tilde{G}$ is $\tilde{R}_0^\star = \emptyset$. Hence, from (\ref{eqn:Rij_star}) $R_{12}^\star = R_{21}^\star = \emptyset$ (since, for instance, $E_{12} = E_{12}^c = \{(1,2)\}$).
 }
 \label{fig:MRES_example}
\end{center}
\end{figure}

\subsection{Computation for maximum redundant edge set} \label{subsec:computation}
To compute the quantities in the statement of Theorem~\ref{thm:decomposition}, the first step  is the identification of the equivalence classes $[v_1], [v_2], \ldots, [v_k]$ defined by relation $\sim$. The following statement is useful in the identification:
\begin{lemma} \label{thm:eq_class_id}
Let $G = (V,E,c(E))$ be an edge weighted directed graph without negative weight closed walks. For $i,j \in V$, let $d_{ij}$ be the minimum walk weight defined in Definition~\ref{def:min_walk_weight}. Then, $i \sim j$ if and only if $d_{ij} + d_{ji} = 0$.
\end{lemma}
\begin{proof}
If $i = j$, then by Definition~\ref{def:eq_re} $i \sim j$ and by Definition~\ref{def:min_walk_weight} $d_{ij} + d_{ji} = d_{ii} + d_{ii} = 0$. Hence the statement holds trivially when $i = j$. Next, we consider the case when $i \neq j$. If $i \sim j$, then Lemma~\ref{lem:dij_dji} specifies that $d_{ij} + d_{ji} = d_{ij} - d_{ij} = 0$. Conversely, suppose $d_{ij} + d_{ji} = 0$. By Definition~\ref{def:min_walk_weight}, associated with $d_{ij}$ and $d_{ji}$ there exist walks $w_{ij} : i \leadsto j$ and $w_{ji} : j \leadsto i$ with weights $d_{ij}$ and $d_{ji}$ respectively. Concatenating $w_{ij}$ and $w_{ji}$ results in a zero weight closed walk traversing $i$ and $j$, and hence $i \sim j$.
\end{proof}
Based on Lemma~\ref{thm:eq_class_id}, we identify the equivalence classes as follows:
\begin{algorithm}[Identification of equivalence classes of graph $G = (V,E,c(E))$ without negative weight closed walks] \label{alg:eq_class_id} \hspace{0cm}
\begin{enumerate}
\item Solve the all-pair shortest path problem for all source/destination pairs in $G$. Let $d_{ij}$ denote the shortest path distance from $i$ to $j$.
\item For each pair of $1 \le i \neq j \le n$, declare $i \sim j$ if and only if $d_{ij} + d_{ji} = 0$. Build an undirected graph $(V,E_\sim)$ such that edge $\{i,j\} \in E_\sim$ if and only if $i \sim j$ and $i \neq j$.
\item The equivalence classes defined by relation $\sim$ are the connected components of $(V,E_\sim)$.
\end{enumerate}
\end{algorithm}
The first step of Algorithm~\ref{alg:eq_class_id} can be computed using Floyd-Warshall algorithm in $O({|V|}^3)$ time, because $G$ does not have any negative weight closed walks. The second step requires $O({|V|}^2)$ time. The third step requires $O(|V| + |E_{\sim}|) = O({|V|}^2)$ time (e.g., \cite{Cormen:2009:IAT:1614191}). Hence, Algorithm~\ref{alg:eq_class_id} requires $O({|V|}^3)$ time. Once the equivalence classes have been identified, the computation involved in Definition~\ref{def:eq_class_notations} and Definition~\ref{def:condensation} requires $O(|E|)$ time and $O(K^2)$ time respectively.

Theorem~\ref{thm:decomposition} decomposes the maximum redundant edge set problem into $K+1$ independent subproblems. Because of Lemma~\ref{thm:tildeG_positive_cycle} and Theorem~\ref{thm:max_redundant_edge_set}, solving the maximum redundant edge set subproblem on the condensation $\tilde{G}$ (for $R_0^\star$) requires only polynomial-time (i.e., $O(K^3)$ with $K \le |V|$). On the other hand, to solve for $R_k^\star$ for $k = 1,2,\ldots,K$ in subgraphs $([v_k], E_k^c, c(E_k^c))$ is NP-hard. The argument is similar to the proof of Theorem~\ref{thm:NP_hard}: the minimum equivalent graph problem \cite{Moyles:1969:AFM:321526.321534}, even for strongly connected graphs, is NP-hard. Additionally, the maximum redundant edge set problem for graphs consisting only of one equivalence class generalizes the former problem. Hence, it is NP-hard to compute $R_k^\star$. On the other hand, it turns out that the subproblem for finding a maximum redundant edge set in $([v_k], E_k^c, c(E_k^c))$ can be solved as the (NP-hard) minimum equivalent graph problem for undirected graph $([v_k], E_k^c)$ using available (exact or inexact) algorithms (e.g., \cite{Moyles:1969:AFM:321526.321534, Hsu:1975:AFM:321864.321866, doi:10.1137/S0097539793256685}). The following statement provides the rationale:
\begin{lemma} \label{thm:weight2unweight}
Let $G = (V,E,c(E))$ be an edge weighted directed graph without negative weight closed walks. For $k \in \{1,2,\ldots,K\}$, let $[v_k]$ be the equivalence class defined by relation $\sim$ (Definition~\ref{def:eq_vk}), and $E_k^c$ be defined in Definition~\ref{def:eq_Ek}. Let $R_k \subseteq E_k^c$ be given. Then, the following two statements are equivalent:
\begin{enumerate}[label=\ref{thm:weight2unweight}.\alph*]
\item \label{lem:R_w2uw} $R_k$ is a redundant edge set of $([v_k], E_k^c, c(E_k^c))$.
\item \label{lem:reach_w2uw} $([v_k], E_k^c)$ and $([v_k], E_k^c \setminus R_k)$ have the same reachability (i.e., there is a walk $i \leadsto j$ in $([v_k], E_k^c)$ if and only if there is a walk $i \leadsto j$ in $([v_k], E_k^c \setminus R_k)$).
\end{enumerate}
\end{lemma}
\begin{proof}
We consider the case for $R_k \neq \emptyset$, since otherwise the statement is trivial. As argued in the proof of Theorem~\ref{thm:NP_hard}, condition~\ref{lem:reach_w2uw} is equivalent to
\begin{equation} \label{eqn:reachability_rp1_b}
\text{$\forall (i,j) \in R_k$, there exists a walk from $i$ to $j$ in $([v_k], E_k^c \setminus R_k)$}.
\end{equation}
Hence, it suffices to argue for the equivalence between \ref{lem:R_w2uw} and (\ref{eqn:reachability_rp1_b}). By Definition~\ref{def:redundant_edge_set}, condition~\ref{lem:R_w2uw} implies (\ref{eqn:reachability_rp1_b}). Conversely, if (\ref{eqn:reachability_rp1_b}) holds then for each $(i,j) \in R_k$ there exists a walk $w_{ij}$ in  $([v_k], E_k^c \setminus R_k, c(E_k^c \setminus R_k))$. Let the walk $w_{ij}$ be of the form $(i = i_0, i_1, \ldots, i_m = j)$. By (\ref{def:Ekc}) and Lemma~\ref{lem:dis_dsj}, the weight of walk $w_{ij}$ is $c_{w_{ij}} = d_{i_0 i_1} + \ldots + d_{i_{m-1} i_m} = d_{i_0 i_m} = d_{ij}$. This is the same as the weight of edge $(i,j) \in E_k^c$ (again, by (\ref{def:Ekc})). Hence, (\ref{eqn:reachability_rp1_b}) implies \ref{lem:R_w2uw}, and the desired equivalence is established.
\end{proof}
\begin{remark}
Lemma~\ref{thm:weight2unweight} states that for subgraph $([v_k], E_k^c, c(E_k^c))$, the maximum redundant edge set problem can be reduced to the minimum equivalent graph problem by ignoring the edge weights. However, this reduction is not guaranteed to be valid in more general cases. See Figure~\ref{fig:weighted_counterex} for an example.
\begin{figure}[!tbh]
\begin{center}
 \includegraphics[width=40mm]{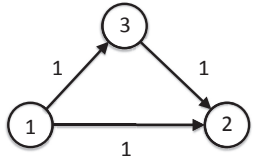}
 \vspace{-0.3cm}
\end{center}
 \caption{The graph illustrated is not one of $([v_k], E_k^c, c(E_k^c))$ for any $k$, since there is no zero weight cycle. When interpreted as a weighted graph, no edge is redundant according to Definition~\ref{def:redundant_edge_set}. However, when the edge weights are ignored, the minimum equivalent graph contains two edges $(1,3)$ and $(3,2)$. In other words, edge $(1,2)$ is redundant in unweighted sense but not redundant in weighted sense.}
 \label{fig:weighted_counterex}
\end{figure}
\end{remark}

The decomposition result (mainly Lemma~\ref{thm:E0}) also leads to some guideline in obtaining approximate solutions to the maximum redundant edge set problem. If $\hat{R}_k$ for $k = 1,2,\ldots,K$ are (not necessarily maximum) redundant edge sets for the subgraphs $([v_k],E_k^c, c(E_k^c))$, then it can be established that
\begin{displaymath}
\text{$[v_1],\ldots,[v_K]$ are equivalence classes in $\big(V, \mathcal{E} \setminus (\mathop{\cup}\limits_{k = 1}^K \hat{R}_k), c(\mathcal{E} \setminus (\mathop{\cup}\limits_{k = 1}^K \hat{R}_k)))\big)$},
\end{displaymath}
where by (\ref{eqn:Rd_eps}) $\mathcal{E} = E_0^d \cup (\mathop{\cup}\limits_{k = 1}^K E_k^c)$. Hence, Lemma~\ref{thm:E0} can be applied to establish that the choices of redundant edge sets in $E_0^d$ are independent of $\hat{R}_k$. Consequently, the largest cardinality redundant edge sets, given the components $\hat{R}_k \in E_k^c$ for $k = 1,2,\ldots,K$, can be parameterized as
\begin{displaymath}
R_0^\star \cup \left(\mathop{\cup}\limits_{k = 1}^K (E_k^r \cup \hat{R}_k) \right),
\end{displaymath}
where $R_0^\star$ is parameterized in (\ref{eqn:R0_star}).

\section{Equivalent reduction of precedence relation systems} \label{sec:equiv_reduction}
This section presents a full parameterization of the set of all equivalent reductions of any precedence relation system. The parameterization will indicate that contrary to the problem of finding the maximum redundant edge set which is NP-hard, every equivalent reduction can be computed in polynomial-time. In the following, we first present some preparatory results in Section~\ref{subsec:er_notations}. Next, the main result on the parameterization of equivalent reduction is described in Section~\ref{subsec:er_parameterization}. Section~\ref{subsec:ve_er} describes a further simplification of equivalent reduction, which is connected to the maximum redundant edge set of the condensation of the original precedence graph.  

Because of the correspondence between a precedence relation system and its precedence graph, in this section we shall extend the notion of equivalent reduction in Definition~\ref{def:equivalent_reduction} to precedence graphs. Given a precedence graph $G$, we call another precedence graph $G'$ equivalent to $G$, with notation $G' \equiv G$, if the precedence relation systems corresponding to $G$ and $G'$ are equivalent (i.e., they have the same solution set). Consequently, an equivalent reduction of a precedence graph $G$ is a precedence graph $G' \equiv G$ such that $G'$ has the minimum possible number of edges.

\subsection{Preparatory results} \label{subsec:er_notations}
The following statements will be used in the proof of the main result to establish some graph-based necessary conditions for the equivalent reductions of a precedence graph.

\begin{lemma} \label{thm:rp_equiv}
Let $G = (V,E,c(E))$ and $G^\prime = (V,E^\prime,c^\prime(E^\prime))$ be two equivalent precedence graphs (i.e., $G \equiv G'$). Then, if $(f,g) \in E$ with weight $c_{fg}$, then there exists a path $p^\prime_{f g}: f \leadsto g$ in $G^\prime$ with weight $c^\prime_{p^\prime_{f g}} \le c_{fg}$.
\end{lemma}
\begin{proof}
Corresponding to $(f,g) \in E$ is an inequality $x_f - x_g \le c_{fg}$ in $G$. Since $G^\prime \equiv G$, whenever $x \in \real^n$ satisfies $G^\prime$, $x$ also satisfies $x_f - x_g \le c_{fg}$. This is condition Lemma~\ref{lem:implication} applied to $G'$. Hence, Lemma~\ref{thm:implication} implies that in $G^\prime$ there exists a path $p^\prime_{fg} : f \leadsto g$ with weight $c^\prime_{p^\prime_{fg}} \le c_{fg}$.
\end{proof}

\begin{remark}
Lemma~\ref{thm:rp_equiv} can in fact be extended (the details omitted) to show that the equivalent reduction problem has the following graph interpretation: given precedence graph $G = (V,E,c(E))$, find (possibly another) precedence graph $G' = (V,E',c'(E'))$ satisfying
\begin{enumerate}
\item For each $(i,j) \in E$ with weight $c_{ij}$ there exists a path $p'_{ij} : i \leadsto j$ in $G'$ with weight $c'_{p'_{ij}} \le c_{ij}$. Conversely, for each $(u,v) \in E'$ with weight $c'_{uv}$ there exists a path $p_{uv} : u \leadsto v$ in $G$ with weight $c_{p_{uv}} \le c'_{uv}$,
\item with respect to the first bullet, $E'$ has minimum cardinality.
\end{enumerate}
From this graph interpretation it is also possible to see that the equivalent reduction problem is a generalization of the transitive reduction problem for unweighted directed graphs in \cite{AGU_transitive_reduction}. That is, an instance of transitive reduction problem is an instance of equivalent reduction problem with edge weights $c_{ij}$ set to zero or a negative constant.
\end{remark}

\begin{lemma} \label{thm:sim_preservation_equiv}
Let $G = (V,E,c(E))$ and $G^\prime = (V,E^\prime,c^\prime(E^\prime))$ be two equivalent precedence graphs (i.e., $G \equiv G'$). Then, for $i,j \in V$, $i \sim j$ in $G$ if and only if $i \sim j$ in $G^\prime$.
\end{lemma}
\begin{proof}
The statement holds trivially when $i = j$. Therefore, only the cases when $i \neq j$ are considered. Suppose $i \sim j$ in $G$. Then there exists a zero-weight closed walk $w_{iji}: (i = i_0, i_1, \ldots, i_k = j, i_{k+1}, \ldots, i_m = i)$ in $G$. Since Lemma~\ref{thm:rp_equiv} and Lemma~\ref{thm:sim_preservation_equiv} have the same statement assumption, Lemma~\ref{thm:rp_equiv} implies that for each $q = 0,1,\ldots,m-1$ there exists a path $p^\prime_{i_q i_{q+1}} : i_q \leadsto i_{q+1}$ in $G^\prime$ whose path weight is less than or equal to the weight of the corresponding edge $(i_q, i_{q+1})$ in $G$. Hence, in $G^\prime$ there is a closed walk $w^\prime_{iji} : i \leadsto j \leadsto i$ with weight less than or equal to that of $w_{iji}$ (which is zero). Since $G^\prime$ is a precedence graph, Assumption~\ref{asm:no_negative_cycle} implies that the weight of $w^\prime_{iji}$ cannot be negative. Hence, $i \sim j$ in $G^\prime$. The above argument can be repeated, with appropriate modifications, to show that $i \sim j$ in $G$, whenever $i \sim j$ in $G^\prime$. Hence, the desired statement is established.
\end{proof}

\begin{remark}
Lemma~\ref{thm:sim_preservation} is in fact a corollary of Lemma~\ref{thm:sim_preservation_equiv}.
\end{remark}

\begin{lemma} \label{thm:E0_equiv}
Let $G = (V,E,c(E))$ be a precedence graph. Let $G^\star = (V,E^\star,c(E^\star))$ with $E^\star \subseteq E$, and $G^\prime = (V,E^\prime,c^\prime(E^\prime))$ with $E^\prime \subseteq V \times V$ and $c^\prime(E^\prime) \in \real^{|E'|}$ be given. Assume that
\begin{itemize}
\item $G \equiv G^\star \equiv G^\prime$.
\item Neither $G^\star$ nor $G^\prime$ has any redundant edge (cf.~Definition~\ref{def:redundant_edge}).
\item For $k \in \{1,2,\ldots,K\}$, $[v_k]$ is the equivalence class induced by relation $\sim$ in $G$, $G^\star$ and $G^\prime$ (the three precedence graphs having the same equivalence class partitioning is a consequence of Lemma~\ref{thm:sim_preservation_equiv} with $G \equiv G^\star \equiv G^\prime$).
\item For $i,j \in V$ and $i \neq j$, let $E^\star_{ij} := E^\star \cap ([v_i] \times [v_j])$, and $E^\prime_{ij} := E^\prime \cap ([v_i] \times [v_j])$. That is, $E^\star_{ij}$ is the subset of $E^\star$ whose member edges go from $[v_i]$ to $[v_j]$, and $E^\prime_{ij}$ is defined analogously.
\end{itemize}
Then, 
\begin{equation} \label{eqn:Eij_equiv}
E^\star_{ij} = \emptyset \iff E^\prime_{ij} = \emptyset, \quad \forall \; 1 \le i \neq j \le K.
\end{equation}
In addition, suppose $(f,g) \in E^\star_{ij}$ with weight $c_{fg}$ and $(u,v) \in E^\prime_{ij}$ with weight $c^\prime_{uv}$. Then,
\begin{equation} \label{eqn:cuv_equiv}
c^\prime_{uv} = d_{uf} + c_{fg} + d_{gv},
\end{equation}
with $d_{uf}$ and $d_{gv}$ being the minimum walk weights in $G$ defined in Definition~\ref{def:min_walk_weight}.
\end{lemma}
\begin{proof}
To show (\ref{eqn:Eij_equiv}) by contradiction, first assume that $E^\star_{ij} = \emptyset$ but $E^\prime_{ij} \neq \emptyset$. Let $(u,v) \in E^\prime_{ij}$ with weight $c_{uv}^\prime$. Since $G^\star \equiv G^\prime$, Lemma~\ref{thm:rp_equiv} guarantees the existence of a path $p_{uv}: u \leadsto v$ in $G^\star$ such that $c_{p_{uv}} \le c^\prime_{uv}$. Further, since $E^\star_{ij} = \emptyset$, $p_{uv}$ must traverse a node $t \notin ([v_i] \cup [v_j])$ and hence it is of the form $p_{uv} : u \leadsto t \leadsto v$. Applying Lemma~\ref{thm:rp_equiv} to each edge in $p_{uv}$ results in a walk $w^\prime_{uv} : u \leadsto t \leadsto v$ in $G^\prime$ such that $c^\prime_{w^\prime_{uv}} \le c_{p_{uv}} \le c^\prime_{uv}$. If $(u,v)$ is not part of $w^\prime_{uv}$, then $(u,v)$ is a redundant edge in $G^\prime$. This is a contradiction. Thus, $(u,v)$ is part of $w^\prime_{uv}$, implying that $w^\prime_{uv}$ is either (a) $u \leadsto u \rightarrow v \leadsto t \leadsto v$ or (b) $u \leadsto t \leadsto u \rightarrow v \leadsto v$. In the case of (a), we distinguish two cases:
\begin{itemize}
\item the weight of the closed walk $v \leadsto t \leadsto v$ is zero ($G'$ cannot have negative weight closed walk because of Assumption~\ref{asm:no_negative_cycle}). Then, $t \sim v$ in $G^\prime$, and this is a contradiction since $t \notin [v_j]$.
\item the weight of the closed walk $v \leadsto t \leadsto v$ is positive. Then, $c^\prime_{w^\prime_{uv}} \le c^\prime_{uv}$ implies that the weight of the closed walk $u \leadsto u$ is negative. This is also a contradiction of Assumption~\ref{asm:no_negative_cycle}.
\end{itemize}
A similar argument can show that case (b) also leads to contradictory conclusions. Therefore, the original assumption that $E^\star_{ij} = \emptyset$ but $E^\prime_{ij} \neq \emptyset$ cannot hold. Further, an analogous argument can show that $E^\prime_{ij} = \emptyset$ but $E^\star_{ij} \neq \emptyset$ cannot hold, and hence (\ref{eqn:Eij_equiv}) is established.

Now we show (\ref{eqn:cuv_equiv}) by contradiction. First, assume that
\begin{equation} \label{eqn:cuv_equiv_gr}
c^\prime_{uv} > d_{uf} + c_{fg} + d_{gv}.
\end{equation}
Because $E^\star \subseteq E$, $G^\star$ is a subgraph of $G$. In addition, since $u,f \in [v_i]$, $g,v \in [v_j]$, $[v_i]$ and $[v_j]$ are equivalence classes of $G^\star$ (as argued in the statement), Lemma~\ref{lem:wij} specifies that there are two walks $u \leadsto f$ and $g \leadsto v$ in $G^\star$ with weights $d_{uf}$ and $d_{gv}$ respectively. Therefore, the right-hand side of (\ref{eqn:cuv_equiv_gr}) is the weight of a walk $w_{uv} : u \leadsto f \rightarrow g \leadsto v$ in $G^\star$. Applying Lemma~\ref{thm:rp_equiv} to each edge of $w_{uv}$ yields a walk $w^\prime_{uv} : u \leadsto f \leadsto g \leadsto v$ in $G^\prime$ such that 
\begin{equation} \label{eqn:cuv_equiv_gr1}
c^\prime_{uv} > d_{uf} + c_{fg} + d_{gv} \ge c^\prime_{w^\prime_{uv}}.
\end{equation}
If $(u,v)$ is not part of $w^\prime_{uv}$ then (\ref{eqn:cuv_equiv_gr1}) implies that $(u,v)$ is a redundant edge in $G^\prime$. This contradicts the statement assumption. On the other hand, if $(u,v)$ is part of $w^\prime_{uv}$, then $w^\prime_{uv}$ is of the form $u \leadsto u \rightarrow v \leadsto v$. (\ref{eqn:cuv_equiv_gr1}) implies that $w^\prime_{uv} \neq (u,v)$ and at least one of the closed walks $u \leadsto u$ and $v \leadsto v$ have negative weight. This contradicts Assumption~\ref{asm:no_negative_cycle}. In conclusion, (\ref{eqn:cuv_equiv_gr}) cannot hold. Next, we assume that
\begin{equation} \label{eqn:cuv_equiv_le}
c^\prime_{uv} < d_{uf} + c_{fg} + d_{gv}.
\end{equation}
By Lemma~\ref{lem:dij_dji}, $d_{uf} = - d_{fu}$ and $d_{gv} = - d_{vg}$. Hence, (\ref{eqn:cuv_equiv_le}) is equivalent to
\begin{equation} \label{eqn:cuv_equiv_le1}
c_{fg} > d_{fu} + c^\prime_{uv} + d_{vg}.
\end{equation}
In addition, Lemma~\ref{thm:rp_equiv} applied to $(u,v) \in E_{ij}^\prime$ yields a walk $w_{uv} : u \leadsto v$ in $G^\star$ such that $c_{w_{uv}} \le c^\prime_{uv}$. This, together with (\ref{eqn:cuv_equiv_le1}), implies
\begin{equation} \label{eqn:cuv_equiv_le2}
c_{fg} > d_{fu} + c_{w_{uv}} + d_{vg}.
\end{equation}
By a similar argument as in (\ref{eqn:cuv_equiv_gr}), Lemma~\ref{lem:wij} guarantees the existence of two walks $f \leadsto u$ and $v \leadsto g$ in $G^\star$ with weights $d_{fu}$ and $d_{vg}$ respectively. Consequently, in $G^\star$ there exists a walk $w_{fg} : f \leadsto g$ such that $c_{w_{fg}} < c_{fg}$. If $(f,g)$ is not part of $w_{fg}$ then (\ref{eqn:cuv_equiv_le2}) implies that $(f,g)$ is a redundant edge in $G^\star$. This contradicts the statement assumption. On the other hand, if $(f,g)$ is part of $w_{fg}$, then $w_{fg}$ is of the form $f \leadsto f \rightarrow g \leadsto g$. (\ref{eqn:cuv_equiv_le2}) implies that $w_{fg} \neq (f,g)$ and at least one of the closed walks $f \leadsto f$ and $g \leadsto g$ have negative weight. This again contradicts Assumption~\ref{asm:no_negative_cycle}. Hence, (\ref{eqn:cuv_equiv_le}) does not hold. Consequently, (\ref{eqn:cuv_equiv}) must hold.
\end{proof}

\subsection{Parameterization of equivalent reduction} \label{subsec:er_parameterization}
Analogous to the maximum redundant edge set problem, an equivalent reduction can be decomposed into $K+1$ components. However, all components of equivalent reduction can be computed in polynomial-time. The following statement summarizes the decomposition result related to equivalent reduction:
\begin{theorem} \label{thm:equiv_reduction}
Let $G = (V,E,c(E))$ be a precedence graph. Let the following be defined in the context of $G$:
\begin{itemize}
\item $K$ is the number of equivalence classes induced by relation $\sim$ in Definition~\ref{def:eq_K}.
\item For $k \in \{1,2,\ldots,K\}$, $[v_k]$ denotes the equivalence class defined in Definition~\ref{def:eq_vk}, with $v_k$ being the representing node of $[v_k]$.
\item For $i, j \in V$, $d_{ij}$ is the minimum walk weight defined in Definition~\ref{def:min_walk_weight}.
\item For $1 \le i \neq j \le K$, $E_{ij}$ and $(v_{ij}^i, v_{ij}^j)$ are defined in Definition~\ref{def:eq_Eij}.
\item $\tilde{G}$ is the condensation of $G$ defined in Definition~\ref{def:condensation}, in accordance with the designation of representing nodes $v_k$'s.
\end{itemize}
Then, every equivalent reduction of $G$ (defined in Definition~\ref{def:equivalent_reduction}) can be parameterized by
\begin{equation} \label{eqn:equiv_reduction_parameterization}
(V,E^{er}, c^{er}(E^{er})).
\end{equation}
In (\ref{eqn:equiv_reduction_parameterization}), the edge set $E^{er}$ is parameterized by
\begin{displaymath}
E^{er} = E_0^{er} \cup (\mathop{\cup}\limits_{k = 1}^K E_k^{er}),
\end{displaymath}
where 
\begin{equation} \label{eqn:Ek_cycle_equiv}
\begin{array}{l}
\text{for $k \in \{1,\ldots,K\}$, $E_k^{er} = \emptyset$ if $|[v_k]| = 1$, otherwise $E_k^{er}$ contains $|[v_k]|$} \vspace{1mm} \\ \text{edges forming a zero weight directed cycle traversing all nodes in $[v_k]$}.
\end{array}
\end{equation}
In addition, $E_0^{er}$ can be decomposed into
\begin{displaymath}
E_0^{er} = \mathop{\cup}\limits_{1 \le i \neq j \le K} E_{ij}^{er},
\end{displaymath}
where
\begin{equation} \label{eqn:Eij_er_edge}
E_{ij}^{er} = \begin{cases} (u,v) \in [v_i] \times [v_j], & \text{if $E_{ij} \neq \emptyset$ and $(v_i, v_j) \notin \tilde{R}_0^\star$} \\ \emptyset, & \text{if $E_{ij} = \emptyset$ or $(v_i, v_j) \in \tilde{R}_0^\star$} \end{cases},
\end{equation}
and $\tilde{R}_0^\star$ is the maximum redundant edge set of $\tilde{G}$, the condensation of $G$. In (\ref{eqn:equiv_reduction_parameterization}), the edge weights $c^{er}(E^{er})$ are defined by
\begin{equation} \label{eqn:Eij_er_weight}
c_{uv}^{er} = \begin{cases} d_{u v_{ij}^i} + c_{v_{ij}^i v_{ij}^j} + d_{v_{ij}^j v}, & \text{if} \; (u,v) \in [v_i] \times [v_j], \; i \neq j, \\ d_{uv}, & \text{if} \; (u,v) \in [v_i] \times [v_i]. \end{cases}
\end{equation}
\end{theorem}

\begin{remark} \label{rmk:equiv_reduction}
It appears that the statement of Theorem~\ref{thm:equiv_reduction} may depend on some arbitrary choices. For instance, $\tilde{R}_0^\star$ may depend on the choices of representing nodes $v_1, v_2, \ldots, v_K$ for the equivalence classes, and the designation of $(v_{ij}^i, v_{ij}^j)$ in (\ref{eqn:Eij_er_edge}) is arbitrary (see Definition~\ref{def:eq_Eij}). However, it turns out that Theorem~\ref{thm:equiv_reduction} is independent of arbitrary choices. In particular, we can show (in Appendix~\ref{app:er_well_defined}) that
\begin{itemize}
\item As in the case of Remark~\ref{rmk:decomposition}, for all $1 \le i \neq j \le K$, $\big| \tilde{R}_0^\star \cap ([v_i] \times [v_j]) \big|$ (which can only be 0 or 1) is independent of the choices of $v_1, \ldots, v_K$.
\item In (\ref{eqn:Eij_er_weight}), $c_{uv}^{er}$ is independent of the designation of $(v_{ij}^i, v_{ij}^j)$ because
\begin{displaymath}
d_{u v_{ij}^i} + c_{v_{ij}^i v_{ij}^j} + d_{v_{ij}^j v} = \underset{(s,t) \in E_{ij}}{\min} \; d_{us} + c_{st} + d_{tv}.
\end{displaymath}
\end{itemize}
\end{remark}

The main difference between the computations in Theorem~\ref{thm:decomposition} and Theorem~\ref{thm:equiv_reduction} is that in an equivalent reduction, the edges in each equivalence class with more than one node form a zero-weight cycle. The formation of these cycles requires only $O(|[v_k]|)$ time for each equivalence class $[v_k]$. Hence, forming all cycles for all equivalent classes requires only $O(|V|)$ time. This difference is analogous to the distinction between the solutions to minimum equivalent graph problem in \cite{Moyles:1969:AFM:321526.321534} and transitive reduction problem in \cite{AGU_transitive_reduction}. Figure~\ref{fig:ER_example} shows two equivalent reductions of the example graph in Figure~\ref{fig:eq_classes}.
\begin{figure}
        \centering
        \begin{subfigure}[b]{0.45\textwidth}
                \includegraphics[width=\textwidth]{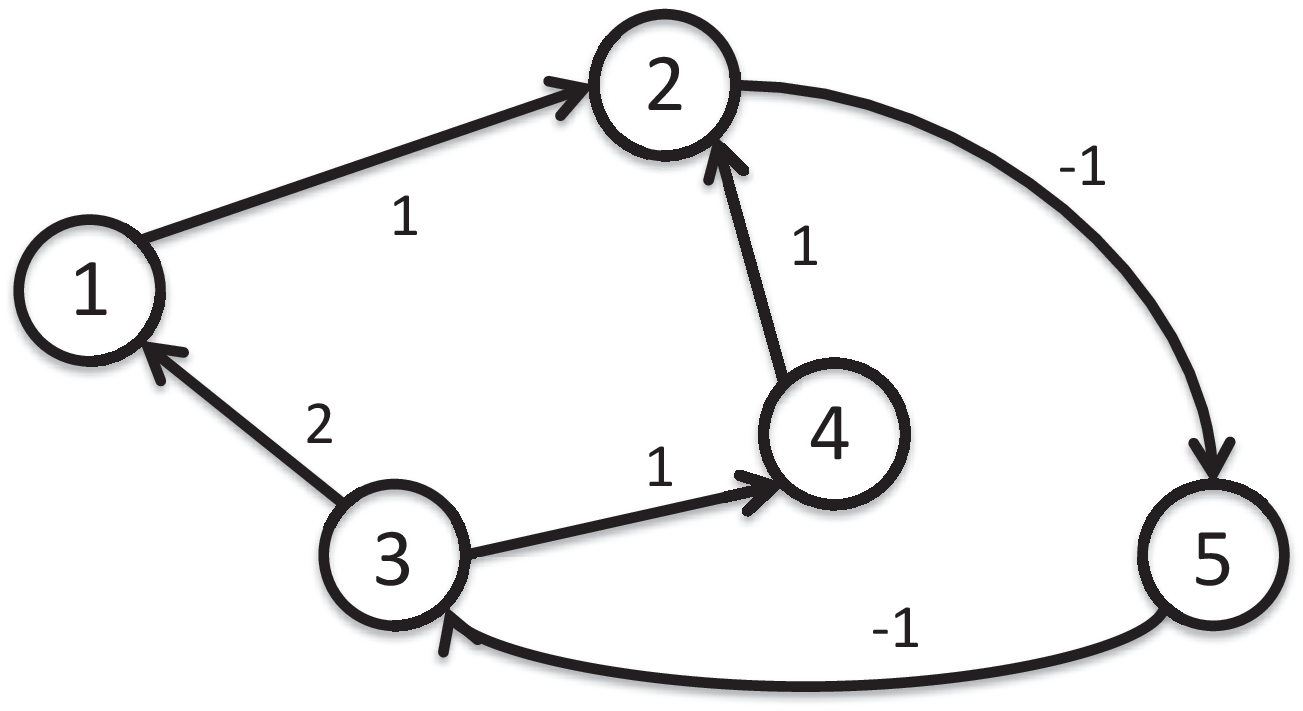}
                \caption{}
                \label{fig:ER_example1}
        \end{subfigure}
         \begin{subfigure}[b]{0.45\textwidth}
                \includegraphics[width=\textwidth]{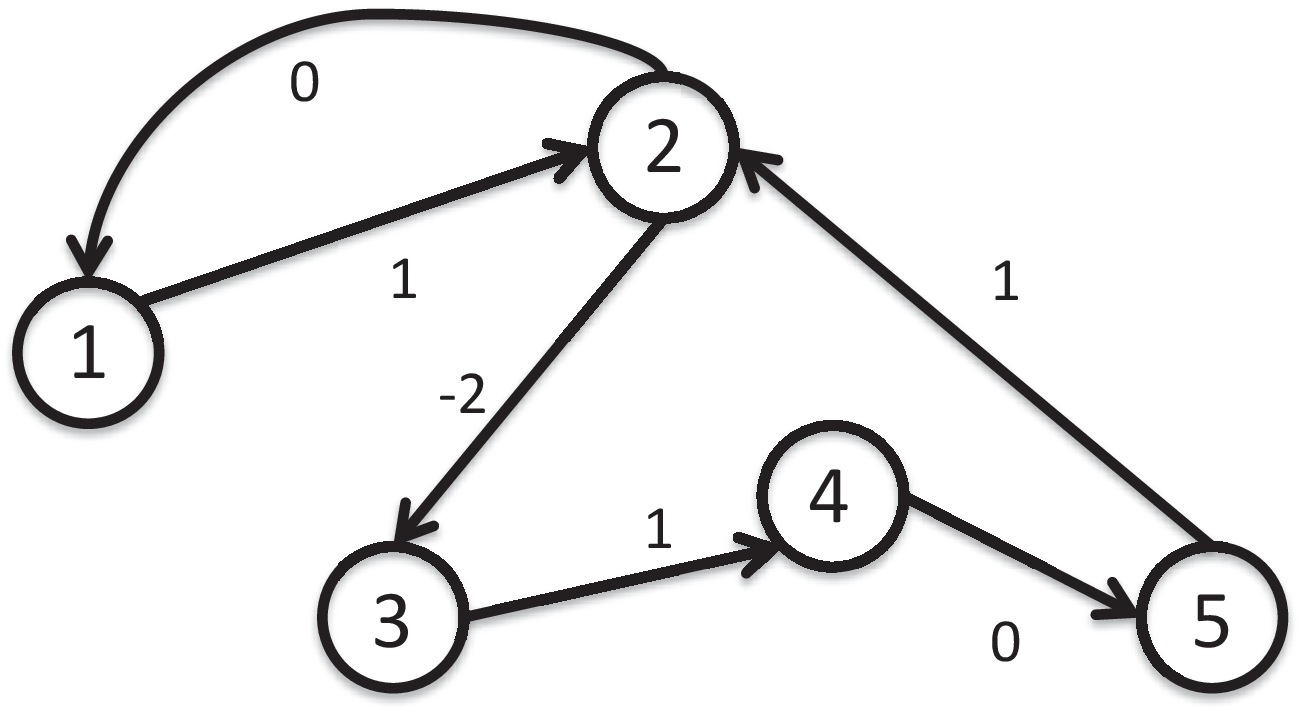}
                \caption{}
                \label{fig:ER_example2}
        \end{subfigure}
        \caption{Two equivalent reductions of the graph in Figure~\ref{fig:eq_classes}. Figure~\ref{fig:ER_example1} is the graph in Figure~\ref{fig:eq_classes} without edge $(3,2)$ since it is not needed to form the zero-weight cycle traversing nodes 2,3,4,5. It can be verified that Figure~\ref{fig:ER_example1} satisfies the parameterization in (\ref{eqn:equiv_reduction_parameterization}). On the other hand, while Figure~\ref{fig:ER_example2} resembles less the graph in Figure~\ref{fig:eq_classes}, it is in fact another equivalent reduction. The cycle $(2,3,4,5,2)$ and the corresponding edge weights satisfy (\ref{eqn:Ek_cycle_equiv}) and (\ref{eqn:Eij_er_weight}). In addition, by (\ref{eqn:Eij_er_edge}) there should be one edge each from $[1]$ to $[2]$ and from $[2]$ to $[1]$, since $E_{12} \neq \emptyset, E_{21} \neq \emptyset$ and $\tilde{R}_0^\star = \emptyset$. The edges $(1,2)$ and $(2,1)$ together with their weights satisfy (\ref{eqn:Eij_er_edge}) and (\ref{eqn:Eij_er_weight}).
        } \label{fig:ER_example}
\end{figure}

\begin{proof}[Proof of Theorem~\ref{thm:equiv_reduction}]
The proof is divided into three parts. In the first part, some necessary conditions of equivalent reduction are listed. In the second part, an alternative characterization of the set of all equivalent reductions is introduced. Finally, in the third part it is established that the alternative characterization is in fact the parameterization provided in the statement of Theorem~\ref{thm:equiv_reduction}. To begin with the first part, let 
\begin{equation} \label{eqn:G_star}
G^\star = (V,E^\star, c(E^\star)) := (V, E \setminus R^\star, c(E \setminus R^\star)),
\end{equation}
with $R^\star$ being a maximum redundant edge set of $G$ parameterized in Theorem~\ref{thm:decomposition}, where in (\ref{eqn:Rij_star}) the edge $(g,h)$ is chosen to be $(v_{ij}^i, v_{ij}^j)$ for $1 \le i \neq j \le K$. Then, because $R^\star$ is a redundant edge set Definition~\ref{def:redundant_edge_set}, Lemma~\ref{thm:redundant_edge_set} and Definition~\ref{def:redundant_relations} imply that $G^\star \equiv G$. In addition, by definition of maximum redundant edge set $G^\star$ does not have any redundant edge. Further, by (\ref{eqn:Rij_star})
\begin{equation} \label{eqn:Eij_star_equiv}
E^\star \cap ([v_i] \times [v_j]) = \begin{cases} (v_{ij}^i, v_{ij}^j), & \text{if $E_{ij} \neq \emptyset$ and $(v_i, v_j) \notin \tilde{R}_0^\star$} \\ \emptyset, & \text{if $E_{ij} = \emptyset$ or $(v_i, v_j) \in \tilde{R}_0^\star$} \end{cases},
\end{equation}
for $1 \le i \neq j \le K$. Now we analyze the properties of equivalent reduction. Let $\mathcal{G}^{ER}$ denote the set of all equivalent reductions of $G$. That is,
\begin{equation} \label{eqn:GER_opt}
\mathcal{G}^{ER} = \underset{G' = (V,E',c'(E')) \equiv G}{\text{argmin}} \; |E'|.
\end{equation}
Let $G'  = (V,E',c'(E')) \in \mathcal{G}^{ER}$. In the following, three properties of $E'$ will be discussed regarding
\begin{itemize}
\item self-loops and parallel edges,
\item edges between two different equivalence classes,
\item edges between two different nodes in an equivalence class.
\end{itemize}
First, we claim that $E'$ does not contain any self-loops or parallel edges. That is,
\begin{equation} \label{eqn:er_no_self-loop}
(i,i) \notin E', \;\; \forall \; i \in V, \quad \text{and all elements of $E'$ are distinct}.
\end{equation}
To see (\ref{eqn:er_no_self-loop}), first note that $G'$ has nonempty solution set because $G' \equiv G$ and $G$ has nonempty solution set because of Assumption~\ref{asm:feasibility}. Hence, the weight $c'_{ii}$ of any self-loop $(i,i)$ must be nonnegative, and the corresponding precedence inequality $x_i - x_i = 0 \le c'_{ij}$ is redundant and should not be present in $G' \in \mathcal{G}^{ER}$. It is also clear that if $G'$ has the minimum number of edges as characterized by (\ref{eqn:GER_opt}), it is impossible to have more than one edge connecting two nodes.

Second, for the edges between equivalent classes, note that $G' \in \mathcal{G}^{ER}$ implies $G' \equiv G$ and $G'$ does not have any redundant edge. Hence, by applying Lemma~\ref{thm:E0_equiv} with $G'$ and $G^\star$, it can be concluded from (\ref{eqn:Eij_equiv}) and (\ref{eqn:Eij_star_equiv}) that
\begin{equation} \label{eqn:Eij_prime_equiv}
E' \cap ([v_i] \times [v_j]) = \begin{cases} (u,v) \in [v_i] \times [v_j], & \text{if $E_{ij} \neq \emptyset$ and $(v_i, v_j) \notin \tilde{R}_0^\star$} \\ \emptyset, & \text{if $E_{ij} = \emptyset$ or $(v_i, v_j) \in \tilde{R}_0^\star$} \end{cases},
\end{equation}
for $1 \le i \neq j \le K$. In addition, by (\ref{eqn:cuv_equiv}) the edge weight in (\ref{eqn:Eij_prime_equiv}) is
\begin{equation} \label{eqn:cuv_prime_equiv}
c_{uv}^\prime = d_{u v_{ij}^i} + c_{v_{ij}^i v_{ij}^j} + d_{v_{ij}^j v}.
\end{equation}
(\ref{eqn:Eij_prime_equiv}) implies that
\begin{equation} \label{eqn:er_E0_edge_count}
\Big| E' \cap (\mathop{\cup}\limits_{1 \le i \neq j \le K} [v_i] \times [v_j]) \Big| = |\tilde{E}_0| - |\tilde{R}_0^\star|,
\end{equation}
where $\tilde{E}_0$ is the edge set of any condensation of $G$. 

Third, we consider the edges in equivalence class $[v_k]$ where $|[v_k]| \ge 2$. Since $G' \equiv G$, Lemma~\ref{thm:sim_preservation_equiv} states that
\begin{equation} \label{eqn:isimj_equiv}
\forall \; i,j \in V, \quad \text{$i \sim j$ in $G' \iff i \sim j$ in $G$}.
\end{equation}
Condition (\ref{eqn:isimj_equiv}) implies that for $k \in \{1,2,\ldots,K\}$, 
\begin{subequations} \label{eqn:isimj_consequences}
\begin{align}
& \text{The graph $H_k := ([v_k], E' \cap ([v_k] \times [v_k]))$ is connected}, \label{eqn:isimj_connected} \\
& |[v_k]| \ge 2 \implies \text{every node in $H_k$ is incident to at least two edges}, \label{eqn:isimj_deg_geq2}\\
& |[v_k]| \ge 2 \implies \text{$H_k$ cannot be a tree}.
\end{align}
\end{subequations}
Thus, $G' \in \mathcal{G}^{ER}$ implies (\ref{eqn:isimj_consequences}), which in turn implies that
\begin{equation} \label{eqn:Ek_er_LB}
|[v_k]| \ge 2 \implies | E' \cap ([v_k] \times [v_k]) | \ge |[v_k]|,
\end{equation}
because if a connected graph has more than one node and it is not a tree, then it has at least as many edges as the number of nodes (e.g., \cite{graph_theory_diestel}). This concludes the first part of the proof (of Theorem~\ref{thm:equiv_reduction}), listing of properties of $G' \in \mathcal{G}^{ER}$.

For the rest of the proof, $G'$ is not assumed to be a member of $\mathcal{G}^{ER}$. We begin the second part of the proof by considering a (to be shown to be) alternative description of $\mathcal{G}^{ER}$. Denote the set
\begin{equation} \label{eqn:Gstar_er}
\mathcal{G}^\star := \left\{ G' = (V,E',c'(E')) \mid \text{$G' \equiv G$, $G'$ satisfies (\ref{eqn:er_no_self-loop}), (\ref{eqn:Eij_prime_equiv}), (\ref{eqn:cuv_prime_equiv}), (\ref{eqn:Ek_er_min})} \right\},
\end{equation}
where condition (\ref{eqn:Ek_er_min}) is defined as
\begin{equation} \label{eqn:Ek_er_min}
|[v_k]| \ge 2 \implies | E' \cap ([v_k] \times [v_k]) | = |[v_k]|.
\end{equation}
In essence, $\mathcal{G}^\star$ in (\ref{eqn:Gstar_er}) is the (to be shown to be nonempty) subset of the feasible set in the optimization problem in (\ref{eqn:GER_opt}), whose elements attain the lower bound of the total number of edges of $G'$ specified by (\ref{eqn:er_no_self-loop}), (\ref{eqn:er_E0_edge_count}) and (\ref{eqn:Ek_er_LB}). In other words, (\ref{eqn:GER_opt}), (\ref{eqn:er_no_self-loop}), (\ref{eqn:er_E0_edge_count}) and (\ref{eqn:Ek_er_LB}) together imply that
\begin{equation} \label{eqn:Gstar_optimal}
\mathcal{G}^\star \neq \emptyset \implies \mathcal{G}^\star = \mathcal{G}^{ER}.
\end{equation}
The set $\mathcal{G}^\star$ in (\ref{eqn:Gstar_er}) can be described in a more convenient form. Its derivation is based on the establishment of two properties of $\mathcal{G}^\star$ to be shown in (\ref{eqn:Ek_cycle_equiv1}) and (\ref{eqn:Ek_weight_equiv1}). As (\ref{eqn:Ek_cycle_equiv1}), it is claimed that if $G' \in \mathcal{G}^\star$ then for all $k \in \{1,2,\ldots,K\}$,
\begin{equation} \label{eqn:Ek_cycle_equiv1}
\begin{array}{rcl}
|[v_k]| = 1 & \implies & E' \cap ([v_k] \times [v_k]) = \emptyset, \vspace{2mm} \\
|[v_k]| \ge 2 & \implies & \text{$E' \cap ([v_k] \times [v_k])$ forms a zero weight directed cycle} \vspace{1mm} \\ & & \text{with $|[v_k]|$ edges traversing all nodes in $[v_k]$}.
\end{array}
\end{equation}
The first implication in (\ref{eqn:Ek_cycle_equiv1}) is a specialization of (\ref{eqn:er_no_self-loop}). For the second implication in (\ref{eqn:Ek_cycle_equiv1}), note that $G' \in \mathcal{G}^\star$ satisfies (\ref{eqn:Ek_er_min}). Hence, restricted to the graph $H_k := ([v_k], E' \cap ([v_k] \times [v_k]))$ the sum of degrees (in-degrees and out-degrees together) of all nodes is $2 \times |[v_k]|$ (e.g., \cite[Theorem 15.3]{Biggs_discrete_math}). If in $H_k$ there is a node with degree greater than two, then there exists at least one other node in $H_k$ with degree less than two. This violates (\ref{eqn:isimj_deg_geq2}), and hence conditions (\ref{eqn:isimj_equiv}) and $G' \equiv G$ are violated as well. Thus, every node in $H_k$ has degree exactly two. This, together with the no self-loop condition in (\ref{eqn:er_no_self-loop}), implies that
\begin{equation} \label{eqn:isimj_deg2}
|[v_k]| \ge 2 \implies \text{every node in $H_k$ is incident to exactly two edges}.
\end{equation}
Note that $G' \equiv G$ since $G' \in \mathcal{G}^\star$. Hence, $G'$ satisfies (\ref{eqn:isimj_equiv}) and (\ref{eqn:isimj_connected}). Then, it is claimed that (\ref{eqn:isimj_equiv}), (\ref{eqn:isimj_connected}) and (\ref{eqn:isimj_deg2}) together imply the second implication in (\ref{eqn:Ek_cycle_equiv1}). (\ref{eqn:isimj_equiv}) and (\ref{eqn:isimj_deg2}) imply that for each node in $[v_k]$ there is one outgoing edge and one incoming edge in $H_k$. This, together with the fact that $H_k$ is a finite graph, suggests that starting from any node $i_0 \in [v_k]$ and following the edges along their directions it is possible to trace a cycle $w = (i_0, i_1, \ldots, i_m = i_0)$ in $H_k$. If there exists $i_{m+1} \in [v_k]$ that is not part of cycle $w$, then by following the outgoing edge of $i_{m+1}$ and subsequent edges along their directions either one of the following is resulted: (i) another cycle $w'$ which is disjointed from $w$ is traced, or (ii) there exists $0 < q \le m$ such that $i_q$ is connected to some $v \in [v_k]$, $v$ not part of $w$. Case (i) contradicts (\ref{eqn:isimj_connected}). On the other hand, case (ii) implies that $i_q$ is connected to $i_{q-1}$, $i_{((q+1) \mod m)}$ and $v$. This suggests that $i_q$ has degree three in $H_k$, and contradicts (\ref{eqn:isimj_deg2}). Thus, both (i) and (ii) lead to contradictory conclusions, and hence $w$ is a directed cycle traversing all nodes in $[v_k]$. Finally, the fact that $w$ is a zero weight cycle is due to (\ref{eqn:isimj_equiv}). Thus, the second implication in (\ref{eqn:Ek_cycle_equiv1}) is established. Furthermore, as a consequence of $G' \in \mathcal{G}^\star$ (in particular $G' \equiv G$ and (\ref{eqn:Ek_cycle_equiv1})), it is claimed that
\begin{equation} \label{eqn:Ek_weight_equiv1}
\text{$(u,v) \in E' \cap ([v_k] \times [v_k])$ for some $k$} \implies c_{uv}^\prime = d_{uv},
\end{equation}
where we note that $d_{uv}$ is the minimum walk weight from $u$ to $v$ in $G$ (Definition~\ref{def:min_walk_weight}). The proof of (\ref{eqn:Ek_weight_equiv1}) is as follows: let $(i_0,i_1,\ldots,i_m = i_0)$ be the directed cycle in $[v_k]$ in (\ref{eqn:Ek_cycle_equiv1}). Then, (\ref{eqn:Ek_weight_equiv1}) holds if it is true that
\begin{equation} \label{eqn:er_proof_cyclew}
c_{i_q i_{q+1}}^\prime = d_{i_q i_{q+1}}, \quad \forall \; q = 0,1,\ldots,m-1.
\end{equation}
The proof of (\ref{eqn:er_proof_cyclew}) is as follows: since $G^\prime \equiv G$, by Lemma~\ref{thm:rp_equiv} for each $q = 0,1,\ldots,m-1$ there exists a walk $w_{i_q i_{q+1}}: i_q \leadsto i_{q+1}$ in $G$ such that $c_{w_{i_q i_{q+1}}} \le c_{i_q i_{q+1}}^\prime$. In addition, since $d_{i_q i_{q+1}} \le c_{w_{i_q i_{q+1}}}$ by Definition~\ref{def:min_walk_weight}, it holds that 
\begin{equation} \label{eqn:er_proof1}
\forall q = 0,\ldots,m-1, \quad d_{i_q i_{q+1}} \le c_{w_{i_q i_{q+1}}} \le c_{i_q i_{q+1}}^\prime \implies c_{i_q i_{q+1}}^\prime - d_{i_q i_{q+1}} \ge 0.
\end{equation}
In addition, by Lemma \ref{lem:dij_dji} and \ref{lem:dis_dsj}
\begin{equation} \label{eqn:er_proof2}
d_{i_0 i_1} + d_{i_1 i_2} + \ldots + d_{i_{m-1} i_m} = d_{i_0 i_{m-1}} + d_{i_{m-1} i_m} = d_{i_0 i_{m-1}} - d_{i_0 i_{m-1}} = 0.
\end{equation}
Since $(i_0,i_1,\ldots,i_m = i_0)$ is a zero weight cycle, it holds that
\begin{equation} \label{eqn:er_proof3}
c_{i_0 i_1}^\prime + c_{i_1 i_2}^\prime + \ldots + c_{i_{m-1} i_m}^\prime = 0.
\end{equation}
Combining (\ref{eqn:er_proof2}) and (\ref{eqn:er_proof3}), it can be concluded that
\begin{equation} \label{eqn:er_proof4}
(c_{i_0 i_1}^\prime - d_{i_0 i_1}) + (c_{i_1 i_2}^\prime - d_{i_1 i_2}) + \ldots + (c_{i_{m-1} i_m}^\prime - d_{i_={m-1} i_m}) = 0.
\end{equation}
Then, (\ref{eqn:er_proof1}) and (\ref{eqn:er_proof4}) together lead to (\ref{eqn:er_proof_cyclew}). In summary,
\begin{displaymath}
G' \in \mathcal{G}^\star \implies \text{$G' \equiv G$, $G'$ satisfies (\ref{eqn:er_no_self-loop}), (\ref{eqn:Eij_prime_equiv}), (\ref{eqn:cuv_prime_equiv}), (\ref{eqn:Ek_cycle_equiv1}), (\ref{eqn:Ek_weight_equiv1})}.
\end{displaymath}
Since (\ref{eqn:Ek_cycle_equiv1}) leads to (\ref{eqn:Ek_er_min}), it holds that
\begin{equation} \label{eqn:Gstar_er1}
\mathcal{G}^\star = \left\{ G' \mid \text{$G' \equiv G$, $G'$ satisfies (\ref{eqn:er_no_self-loop}), (\ref{eqn:Eij_prime_equiv}), (\ref{eqn:cuv_prime_equiv}), (\ref{eqn:Ek_cycle_equiv1}), (\ref{eqn:Ek_weight_equiv1})} \right\}.
\end{equation}
This concludes the second part of the proof (of Theorem~\ref{thm:equiv_reduction}), defining and characterizing $\mathcal{G}^\star$ as a possible alternative description of $\mathcal{G}^{ER}$, the set of all equivalent reductions of $G$.

In the last part of the proof, we establish the connection between $\mathcal{G}^\star$ and the parameterization in the statement of Theorem~\ref{thm:equiv_reduction}. Denote $\mathcal{G}^{er}$ as the set of all precedence graphs satisfying (\ref{eqn:equiv_reduction_parameterization}) to (\ref{eqn:Eij_er_weight}) (i.e., the parameterization in the statement of Theorem~\ref{thm:equiv_reduction}). Then, it can be seen that
\begin{equation} \label{eqn:Ger1}
\mathcal{G}^{er} = \left\{ G' \mid \text{$G'$ satisfies (\ref{eqn:er_no_self-loop}), (\ref{eqn:Eij_prime_equiv}), (\ref{eqn:cuv_prime_equiv}), (\ref{eqn:Ek_cycle_equiv1}), (\ref{eqn:Ek_weight_equiv1})} \right\}.
\end{equation}
That is, the set $\mathcal{G}^{er}$ in (\ref{eqn:Ger1}) includes $\mathcal{G}^\star$ in (\ref{eqn:Gstar_er1}), since the definition of the $\mathcal{G}^{er}$ is the same as $\mathcal{G}^\star$ except that the requirement $G' \equiv G$ is removed. As a side note, $G' \in \mathcal{G}^{er}$ is the only precedence graph considered in this paper where Assumption~\ref{asm:feasibility} (i.e., feasibility) cannot be taken for granted because the condition $G' \equiv G$ has not been shown. In the remaining part of the proof we will show that
\begin{equation} \label{eqn:Gpp_equiv_Gstar}
G' = (V, E', c'(E')) \in \mathcal{G}^{er} \implies G' \equiv G^\star \; (\equiv G),
\end{equation}
where $G^\star$ is defined in the earlier part of the proof in (\ref{eqn:G_star}). If (\ref{eqn:Gpp_equiv_Gstar}) holds, then the desired statement in the theorem is shown because
\begin{displaymath}
\mathcal{G}^{er} = \mathcal{G}^\star \implies \mathcal{G}^\star \neq \emptyset \; (\text{since $\mathcal{G}^{er} \neq \emptyset$}) \overset{(\ref{eqn:Gstar_optimal})}{\implies} \mathcal{G}^{er} = \mathcal{G}^\star = \mathcal{G}^{ER}.
\end{displaymath}
Now we start to show (\ref{eqn:Gpp_equiv_Gstar}), and assume that $G' \in \mathcal{G}^{er}$. We first show one direction of ``$\equiv$'' in (\ref{eqn:Gpp_equiv_Gstar}) by proving
\begin{equation} \label{eqn:Gstar2Gpp}
\text{$x \in \real^n$ satisfies $G^\star$} \implies x_u - x_v \le c'_{uv}, \quad \forall (u,v) \in E'.
\end{equation}
There are two cases in (\ref{eqn:Gstar2Gpp}): (a) $u \neq v, u,v \in [v_k]$ for some $k$, or (b) $u \in [v_k]$ and $v \in [v_q]$ with $k \neq q$. For case (a), (\ref{eqn:Ek_weight_equiv1}) requires that the weight of $(u,v)$ is $c'_{uv} = d_{uv}$. Since $G^\star \equiv G$, Lemma~\ref{thm:sim_preservation_equiv} implies that $G^\star$ and $G$ have the same equivalence class partitioning. Thus, by Lemma~\ref{lem:wij} there exists a walk $u \leadsto v$ in $G^\star$ with weight $d_{uv}$, which is the same as the weight of $(u,v)$ in $G'$. Then, Lemma~\ref{thm:implication} (applied to $G^\star$) establishes (\ref{eqn:Gstar2Gpp}) in case (a). For case (b), by (\ref{eqn:Eij_star_equiv}) and (\ref{eqn:Eij_prime_equiv}) $(u,v) \in E'$ implies that $(v_{kq}^k, v_{kq}^q) \in E^\star$. In addition, (\ref{eqn:cuv_prime_equiv}) implies that $c'_{uv} = d_{u v_{kq}^k} + c_{v_{kq}^k v_{kq}^q} + d_{v_{kq}^q v}$. Since $u,v_{kq}^k \in [v_k]$, $v,v_{kq}^q \in [v_q]$, Lemma~\ref{lem:wij} specifies that there exists two walks in $G^\star$, $u \leadsto v_{kq}^k$ and $v_{kq}^q \leadsto v$, with weights $d_{u v_{kq}^k}$ and $d_{v_{kq}^q v}$ respectively. Hence, in $G^\star$ there exists a walk $u \leadsto v_{kq}^k \rightarrow v_{kq}^q \leadsto v$ with weight equal to $c'_{uv}$. Again, by Lemma~\ref{thm:implication} it can be concluded that (\ref{eqn:Gstar2Gpp}) holds under case (b). In conclusion, (\ref{eqn:Gstar2Gpp}) holds. Next, we show the other direction of ``$\equiv$'' in (\ref{eqn:Gpp_equiv_Gstar}) by establishing 
\begin{equation} \label{eqn:Gpp2Gstar}
\text{$x \in \real^n$ satisfies $G'$} \implies x_i - x_j \le c_{ij}, \quad \forall (i,j) \in E^\star.
\end{equation}
First, we note that (\ref{eqn:Gstar2Gpp}) implies that $G'$ has at least one solution since $G^\star$ has at least one. Therefore, $G'$ also satisfies Assumption~\ref{asm:graph_assumptions} and Lemma~\ref{thm:implication} can be applied to $G'$. There are two cases in (\ref{eqn:Gpp2Gstar}): (c) $i \neq j, i,j \in [v_k]$ for some $k$, or (d) $i \in [v_k]$ and $j \in [v_q]$ with $k \neq q$. For case (c), by (\ref{eqn:Ek_cycle_equiv1}) and (\ref{eqn:Ek_weight_equiv1}) there exists a path $(i = i_0, i_1, \ldots, i_m = j)$ in $G'$, as a part of the cycle traversing nodes in $[v_k]$, with weight $d_{i_0 i_1} + \ldots + d_{i_{m-1} i_m} = d_{ij} \le c_{ij}$ (the equality is due to Lemma~\ref{lem:dis_dsj} and the inequality is due to Definition~\ref{def:min_walk_weight}). Thus, by Lemma~\ref{thm:implication} (\ref{eqn:Gpp2Gstar}) holds in case (c). For case (d), (\ref{eqn:Eij_star_equiv}), (\ref{eqn:Eij_prime_equiv}) and (\ref{eqn:cuv_prime_equiv}) specify that $i = v_{kq}^k$, $j = v_{kq}^q$ and there exists $(u,v) \in E'$ with $u \in [v_k]$, $v \in [v_q]$ such that $c_{ij} = d_{iu} + c'_{uv} + d_{vj}$. Since $i,u \in [v_k]$ and $v,j \in [v_q]$, (\ref{eqn:Ek_cycle_equiv1}) and (\ref{eqn:Ek_weight_equiv1}) imply the existence of two paths $i \leadsto u$ and $v \leadsto j$ in $G'$ with weights $d_{iu}$ and $d_{vj}$ respectively. Hence, in $G'$ there exists a path $i \leadsto u \rightarrow v \leadsto j$ with weight $d_{iu} + c'_{uv} + d_{vj} = c_{ij}$. Consequently, by Lemma~\ref{thm:implication} (\ref{eqn:Gpp2Gstar}) holds in case (d). Therefore, (\ref{eqn:Gpp2Gstar}) holds and consequently (\ref{eqn:Gpp_equiv_Gstar}) holds. This concludes the proof of Theorem~\ref{thm:equiv_reduction}.
\end{proof}

\subsection{Variable elimination in equivalent reduction} \label{subsec:ve_er}
Let $G^{er} = (V, E^{er}, c^{er}(E^{er}))$ denote an equivalent reduction of a precedence graph $G$. By Theorem~\ref{thm:equiv_reduction} (i.e., (\ref{eqn:Ek_cycle_equiv}) and (\ref{eqn:Eij_er_weight})) in $G^{er}$ if $i, j \in [v_k]$ for some $k$ then $x_i - x_j = d_{ij}$ for all $x \in \real^n$ satisfying $G$. In some applications (e.g., optimization), it is beneficial to use these equalities to eliminate all except one variable in each equivalence class (e.g., keeping only the representing node $v_k$). The resulted simplified inequality system can be considered as a ``condensation'' of $G^{er}$ constructed as follows:
\begin{algorithm}[Condensation of equivalent reduction] \label{alg:er_condensation} \hspace{0cm}
\begin{enumerate}
\item Designate representing nodes in the equivalence classes as $v_1, v_2, \ldots, v_K$.
\item Define $\tilde{V} := \{v_1, v_2, \ldots, v_K\}$.
\item Define $\hat{E}_0^{er} := \{(v_i, v_j) \mid E^{er} \cap ([v_i] \times [v_j]) = E_{ij}^{er} \neq \emptyset, \; i \neq j \}$.
\item For each $(v_i, v_j) \in \hat{E}_0^{er}$, define edge weight
\begin{equation} \label{eqn:er_con_edge_weight}
\hat{c}_{v_i v_j}^{er} := d_{v_i u} + c^{er}_{uv} + d_{v v_j}, 
\end{equation}
where $\{(u,v)\} = E^{er} \cap ([v_i] \times [v_j])$.
\item Define the condensation of $G^{er}$ to be $(\tilde{V}, \hat{E}_0^{er}, \hat{c}^{er}(\hat{E}_0^{er}))$.
\end{enumerate}
\end{algorithm}
It turns out that the condensation of an equivalent reduction of $G$ can be interpreted in an alternative way related to the condensation of $G$:

\begin{lemma}
Let $G^{er} = (V, E^{er}, c^{er}(E^{er}))$ be an equivalent reduction of a precedence graph $G$. Let $\hat{E}_0^{er}$ and $\hat{c}^{er}(\hat{E}_0^{er})$, defined in Algorithm~\ref{alg:er_condensation}, describe the condensation of $G^{er}$. Let $\tilde{G} = (\tilde{V}, \tilde{E}_0, \tilde{c}(\tilde{E}_0))$ denote the condensation of $G$, with the same designation of $v_1,v_2,\ldots,v_K$ as in the choice of Algorithm~\ref{alg:er_condensation}. Let $\tilde{R}_0^\star$ be the maximum redundant edge set of $\tilde{G}$. Then,
\begin{equation} \label{eqn:equiv_cons}
(\tilde{V}, \hat{E}_0^{er}, \hat{c}^{er}(\hat{E}_0^{er})) = (\tilde{V}, \tilde{E}_0 \setminus \tilde{R}_0^\star, \tilde{c}(\tilde{E}_0 \setminus \tilde{R}_0^\star)). 
\end{equation}
That is, the condensation of an equivalent reduction of $G$ is the condensation of $G$ with its maximum redundant edge set removed.
\end{lemma}
\begin{proof}
To see (\ref{eqn:equiv_cons}), we compare the following:
\begin{itemize}
\item $\tilde{V}$ is the node set in the left-hand and right-hand sides of (\ref{eqn:equiv_cons}).
\item For the edge sets, by the definition in Algorithm~\ref{alg:er_condensation} and Theorem~\ref{thm:equiv_reduction} (i.e., (\ref{eqn:Eij_er_edge})) it holds that
\begin{displaymath}
(v_i, v_j) \in \hat{E}_0^{er} \iff E_{ij}^{er} \neq \emptyset \iff (v_i, v_j) \in (\tilde{E}_0 \setminus \tilde{R}_0^\star).
\end{displaymath}
Hence, the left-hand and right-hand sides of (\ref{eqn:equiv_cons}) have the same edge sets (i.e., $\hat{E}_0^{er} = \tilde{E}_0 \setminus \tilde{R}_0^\star$).
\item For each $(v_i, v_j) \in \hat{E}_0^{er}$, the edge weight satisfies
\begin{displaymath}
\begin{array}{rcl}
\hat{c}_{v_i v_j}^{er} & \overset{(\ref{eqn:er_con_edge_weight})}{=} & d_{v_i u} + c^{er}_{uv} + d_{v v_j} \vspace{0mm} \\
& \overset{(\ref{eqn:Eij_er_weight})}{=} & d_{v_i u} + d_{u v_{ij}^i} + c_{v_{ij}^i v_{ij}^j} + d_{v_{ij}^j v} + d_{v v_j} \vspace{0mm} \\
& \overset{\text{Lemma~\ref{lem:dis_dsj}}}{=} & d_{v_i v_{ij}^i} + c_{v_{ij}^i v_{ij}^j} + d_{v_{ij}^j v_j} \vspace{0mm} \\
& \overset{(\ref{eqn:viji_vijj})}{=} & \tilde{c}_{v_i v_j}
\end{array}.
\end{displaymath}
\end{itemize}
\end{proof}



\section{Conclusions} \label{sec:conclusion}
The maximum index set of redundant relations problem (i.e., maximum redundant edge set problem) is a generalization of the minimum equivalent graph problem. Similarly, the equivalent reduction problem is a generalization of the transitive reduction problem. Nevertheless, the generalizations are shown to possess analogous computation properties of the respective restrictions. The maximum redundant edge set problem is NP-hard in general, and it is solvable in polynomial time if the graph does not have zero-weight cycles. This is analogous to the statement that the minimum equivalent graph problem is NP-hard in general, and it is solvable in polynomial time if the graph is acyclic. In addition, the decomposition of the maximum redundant edge set problem based on the equivalence classes defined by the ``on-a-zero-weight-closed-walk'' relation is analogous to the decomposition of the minimum equivalent graph problem based on strongly connected components. Further, in the decomposition, the subproblems dealing with the edges between equivalence classes are both solvable in polynomial time. The subproblem within an equivalence class for maximum redundant edge set problem is in fact equivalent to the (NP-hard) minimum equivalent graph problem within the corresponding equivalence class, with the implication that all available exact or inexact algorithms for minimum equivalent graph problem can be utilized. Finally, analogous results also hold between the equivalent reduction problem and the transitive reduction problem. The structure of decompositions of the solutions are analogous, and both problems can be solved in polynomial-time using similar shortest path calculations.


\appendix

\section{Decomposing a walk into a path and cycles} \label{app:walk_decomp} 
In a directed graph $G = (V,E)$ without parallel edges, a walk can be represented by a sequence of node indices $w = (i_0, i_1, \ldots, i_m)$ where $(i_k, i_{k+1}) \in E$ for $k = 0,1,\ldots,m-1$. To decompose a walk $w$, the following ``scan'' operation is necessary:
\begin{flushleft}
$S = {\rm scan}(w)$
\end{flushleft}
\begin{itemize}
\item Initialize $S \leftarrow \emptyset$ and $k \leftarrow 0$.
\item While $k < {\rm length}(w)$, do
\begin{enumerate}
\item If there exists $r$ as the smallest index such that $r > k $ and $i_r = i_k$, then update $S \leftarrow S \cup \{(i_k, i_{k+1}, \ldots, i_r)\}$. The sequence $w$ is also updated according to
\begin{displaymath}
w \leftarrow 
\begin{cases} (\;\;) & \textrm{if $k = 0$ and $r = m$} \\
 (i_0, i_1, \ldots, i_k, i_{r+1}, i_{r+2}, \ldots, i_m) & \textrm{otherwise}
\end{cases}.
\end{displaymath}
In updating $w$, the sub-walk $(i_{r+1},\ldots,i_m)$ is empty by convention if $r = m$.
\item Increase $k \leftarrow k+1$.
\end{enumerate}
End (of While)
\item Update $S \leftarrow S \cup \{w\}$.
\end{itemize}
For any two nodes $u$ and $v$, applying the scan operation to a walk $w = (u = i_0, i_1, \ldots, i_m = v)$ results in $S$ a path from $u$ to $v$ (can be a degenerate path containing a single node if $u = v$) and a finite number of closed walks (if any). The scan operation can be applied to each closed walk and all closed walks generated subsequently that contain more than one cycle. The recursive application of the scan operation eventually decomposes all closed walks into cycles in finite number of steps. To see this, each time when ``children'' closed walks are generated by passing a ``parent'' closed walk through scan, the number of edges of the children closed walks must be smaller than that of the parent. The scan operation can be applied to each children which is not a cycle, which might generate more ``grand-children'' closed walks with even fewer edges. However, the recursive application of scan cannot continue indefinitely since all closed walks with two edges are cycles.

\section{Theorem~\ref{thm:decomposition} independent of choices of $v_k$'s} \label{app:decompose_well_defined}

For $E_{ij}^c$, we claim that $E_{ij}^c$ is independent of the choices of $v_k$'s. For any $1 \le i \neq j \le K$, (\ref{def:Eijc}) specifies that $(u,v) \in E_{ij}^c$ if and only if
\begin{equation} \label{eqn:app_Eijc}
d_{v_i s} + c_{st} + d_{t v_j} \ge d_{v_i u} + c_{uv} + d_{v v_i}, \quad \forall (s,t) \in E_{ij}.
\end{equation}
However, by lemma~\ref{lem:dij_dji} and \ref{lem:dis_dsj}, (\ref{eqn:app_Eijc}) can be rewritten as
\begin{displaymath}
c_{st} \ge d_{s v_i} + d_{v_i u} + c_{uv} + d_{v v_j} + d_{v_j t} = d_{s u} + c_{uv} + d_{v t}.
\end{displaymath}
This shows that $E_{ij}^c$ is independent of the choices of $v_k$'s. 

Next, we establish that the set of $R_0^\star$ is independent of the choices of $v_k$'s. Let $\tilde{G}$ be the condensation, where the representing nodes of the equivalence classes are $v_1, \ldots, v_K$. Now instead, let $\hat{v}_1, \ldots, \hat{v}_K$, where $\hat{v}_k$ can be different from $v_k$, be the representing nodes. Let $\hat{G}$ be the corresponding condensation. By Definition~\ref{def:condensation}, $(v_i, v_j)$ is an edge of $\tilde{G}$ if and only if $(\hat{v}_i, \hat{v}_j)$ is an edge of $\hat{G}$. In addition, it is claimed that if $\tilde{c}_{\hat{v}_i \hat{v}_j}$ is the edge weight of $(\hat{v}_i, \hat{v}_j)$ in $\hat{G}$, then it holds that 
\begin{equation} \label{eqn:app_tilde_c_eq}
\tilde{c}_{\hat{v}_i \hat{v}_j} = d_{\hat{v}_i v_i} + \tilde{c}_{v_i v_j} + d_{v_j \hat{v}_j}.
\end{equation}
First, we claim that 
\begin{equation}  \label{eqn:app_tilde_c_ineq}
\tilde{c}_{\hat{v}_i \hat{v}_j} \le d_{\hat{v}_i v_i} + \tilde{c}_{v_i v_j} + d_{v_j \hat{v}_j}.
\end{equation}
Note that, by (\ref{eqn:viji_vijj}) and Lemma~\ref{lem:dis_dsj}, it holds that
\begin{displaymath}
\begin{aligned}
d_{\hat{v}_i v_i} + \tilde{c}_{v_i v_j} + d_{v_j \hat{v}_j} & = d_{\hat{v}_i v_i} + d_{v_i v_{ij}^i} + c_{v_{ij}^i v_{ij}^j} + d_{v_{ij}^j v_j} + d_{v_j \hat{v}_j} \\
& = d_{\hat{v}_i v_{ij}^i} + c_{v_{ij}^i v_{ij}^j} + d_{v_{ij}^j \hat{v}_j}.
\end{aligned}
\end{displaymath}
Hence, the right-hand side of (\ref{eqn:app_tilde_c_ineq}) is in fact the weight of the walk $\hat{v}_i \leadsto v_{ij}^i \rightarrow v_{ij}^j \leadsto \hat{v}_j$. Consequently, by (\ref{eqn:tilde_c_ij}) (applied to $\hat{G}$), (\ref{eqn:app_tilde_c_ineq}) holds. Next, to complete the proof of (\ref{eqn:app_tilde_c_eq}), suppose that
\begin{equation} \label{eqn:app_tildec}
\tilde{c}_{\hat{v}_i \hat{v}_j} = d_{\hat{v}_i \hat{v}_{ij}^i} + c_{\hat{v}_{ij}^i \hat{v}_{ij}^j} + d_{\hat{v}_{ij}^j \hat{v}_j} < d_{\hat{v}_i v_i} + \tilde{c}_{v_i v_j} + d_{v_j \hat{v}_j}.
\end{equation}
Then, with Lemma~\ref{lem:dij_dji} and \ref{lem:dis_dsj}, (\ref{eqn:app_tildec}) implies that
\begin{displaymath}
d_{v_i \hat{v}_{ij}^i} + c_{\hat{v}_{ij}^i \hat{v}_{ij}^j} + d_{\hat{v}_{ij}^j v_j} < \tilde{c}_{v_i v_j}.
\end{displaymath}
This contradicts (\ref{eqn:tilde_c_ij}) in Definition~\ref{def:condensation}. Consequently, (\ref{eqn:app_tilde_c_eq}) holds. With (\ref{eqn:app_tilde_c_eq}), it can be shown that for any sequence $\hat{v}_{k(1)}, \hat{v}_{k(2)}, \ldots, \hat{v}_{k(m)}$, the inequality
\begin{equation} \label{eqn:app_tilde_hat}
\tilde{c}_{\hat{v}_{k(0)} \hat{v}_{k(m)}} \ge \tilde{c}_{\hat{v}_{k(0)} \hat{v}_{k(1)}} + \tilde{c}_{\hat{v}_{k(1)} \hat{v}_{k(2)}} + \ldots + \tilde{c}_{\hat{v}_{k(m-1)} \hat{v}_{k(m)}}
\end{equation}
is equivalent to
\begin{equation} \label{eqn:app_tilde_hat2}
\begin{aligned}
& d_{\hat{v}_{k(0)} v_{k(0)}} + \tilde{c}_{v_{k(0)} v_{k(m)}} + d_{v_{k(m)} \hat{v}_{k(m)}} \\
\ge \; & d_{\hat{v}_{k(0)} v_{k(0)}} + \tilde{c}_{v_{k(0)} v_{k(1)}} + d_{v_{k(1)} \hat{v}_{k(1)}} + d_{\hat{v}_{k(1)} v_{k(1)}} + \tilde{c}_{v_{k(1)} v_{k(2)}} + d_{v_{k(2)} \hat{v}_{k(2)}} \\
& + \ldots + d_{\hat{v}_{k(m-1)} v_{k(m-1)}} + c_{v_{k(m-1)} v_{k(m)}} + d_{v_{k(m)} \hat{v}_{k(m)}}.
\end{aligned}
\end{equation}
In (\ref{eqn:app_tilde_hat2}), Lemma~\ref{lem:dij_dji} yields $d_{v_{k(q)} \hat{v}_{k(q)}} + d_{\hat{v}_{k(q)} v_{k(q)}} = 0$ for $q = 1,\ldots,m-1$. Thus, (\ref{eqn:app_tilde_hat2}) becomes
\begin{equation} \label{eqn:app_tilde}
\tilde{c}_{v_{k(0)} v_{k(m)}} \ge \tilde{c}_{v_{k(0)} v_{k(1)}} + \tilde{c}_{v_{k(1)} v_{k(2)}} + \ldots + \tilde{c}_{v_{k(m-1)} v_{k(m)}}.
\end{equation}
That is, inequalities (\ref{eqn:app_tilde_hat}) and (\ref{eqn:app_tilde}) are equivalent. The equivalence implies that $\tilde{R}$ is a redundant edge set in $\tilde{G}$ if and only if $\hat{R} := \{(\hat{v}_i, \hat{v}_j) \mid (v_i, v_j) \in \tilde{R} \iff (\hat{v}_i, \hat{v}_j) \in \hat{R}\}$ is a redundant edge set in $\hat{G}$ with the same cardinality, and vice versa. Therefore, for $1 \le i \neq j \le K$, either there is an edge from $[v_i]$ to $[v_j]$ in the maximum redundant edge set of all condensations, or there is no edge from $[v_i]$ to $[v_j]$ in the maximum redundant edge set of any condensation. This is the second bullet of Remark~\ref{rmk:decomposition}. Consequently, the definition of $R_{ij}^\star$ (and hence $R_0^\star$) in Theorem~\ref{thm:decomposition} is independent of the choices of $v_k$'s.

\section{Theorem~\ref{thm:equiv_reduction} independent of arbitrary choices} \label{app:er_well_defined}
The first bullet has been shown in Appendix~\ref{app:decompose_well_defined}. The second bullet is shown as follows:
\begin{displaymath}
\begin{array}{cl}
& \underset{(s,t) \in E_{ij}}{\min} \; d_{us} + c_{st} + d_{tv} \vspace{1mm} \\
\overset{\text{Lemma~\ref{lem:dis_dsj}}}{=} & \underset{(s,t) \in E_{ij}}{\min} \; d_{u v_i} + d_{v_i s} + c_{st} + d_{t v_j} + d_{v_j v} \vspace{1mm} \\
= & d_{u v_i} + \big( \underset{(s,t) \in E_{ij}}{\min} \; d_{v_i s} + c_{st} + d_{t v_j} \big) + d_{v_j v} \vspace{1mm} \\
\overset{\text{Definition~\ref{def:eq_Eij}}}{=} & d_{u v_i} + d_{v_i v_{ij}^i} + c_{v_{ij}^i v_{ij}^j} + d_{v_{ij}^j v_j} + d_{v_j v} \vspace{1mm} \\
\overset{\text{Lemma~\ref{lem:dis_dsj}}}{=} & d_{u v_{ij}^i} + c_{v_{ij}^i v_{ij}^j} + d_{v_{ij}^j v}.
\end{array}
\end{displaymath}

\bibliographystyle{IEEEtran}
\bibliography{pgs}

\end{document}